\documentclass[12pt]{article}

\usepackage{fullpage}
\usepackage[T1]{fontenc}
\usepackage{ae,aecompl}
\usepackage{amssymb}
\usepackage{amsmath}
\usepackage{ntheorem}
\usepackage[round,authoryear]{natbib}
\usepackage{color}
\usepackage{array}
\usepackage{graphicx}
\usepackage{multirow}
\usepackage{multicol}
\usepackage{float}
\usepackage{epic}
\usepackage{tikz}
\usepackage{bm}
\usepackage{extarrows}
\usepackage[center,sc]{titlesec}
\usepackage[a4paper,margin=1.1in]{geometry}
\usepackage[pdfborder=false]{hyperref}
\hypersetup{colorlinks,citecolor=darkred,filecolor=black,linkcolor=darkblue,urlcolor=webgreen,pdfpagemode=None,
pdfstartview=FitH}

\definecolor{webgreen}{rgb}{0,0.4,0}
\definecolor{webbrown}{rgb}{0.6,0,0}
\definecolor{purple}{rgb}{0.5,0,0.25}
\definecolor{darkblue}{rgb}{0,0,0.7}
\definecolor{darkred}{rgb}{0.7,0,0}

\newtheorem{lemma}{{\sc Lemma}}

\newtheorem{theorem}{{\sc Theorem}}
\newtheorem{definition}{{\sc Definition}}
\newtheorem{example}{{\sc Example}}

\newtheorem{observation}{{\sc Observation}}
\newtheorem{remark}{{\sc Remark}}
\newenvironment{proof}{\noindent {\bf \sl Proof\/}:\enspace}{\hfill $\square$ \vspace{4pt}}

\begin{document}
	
\title{\vspace{-1em}\protect\textsc{Probabilistic Fixed Ballot Rules\\
and Hybrid Domains}\thanks{We would like to thank the Editor Andr\'{e}s Carvajal, the Associate Editor Olivier Bochet, and two anonymous Referees for their constructive comments and suggestions. We are also grateful to Lars Ehlers and Andrew Mackenzie for their valuable discussions and suggestions. We also would like to thank the participants of the Workshop on Economic Design, New Delhi, the Conference on Economic Design, Budapest, the International Conference on Economic Theory and Applications, Chengdu, SAET Conference, Taipei and the Conference on Mechanism and Institution Design, Durham, for helpful comments.
Shurojit Chatterji acknowledges the support by the Ministry of Education, Singapore under the grant MOE2016-T2-1-168.
Huaxia Zeng acknowledges that his work was supported by the National Natural Science Foundation of China (No.~71803116),
the Program for Professor of Special Appointment (Eastern Scholar) at Shanghai Institutions of Higher Learning (No.~2019140015) and
the Fundamental Research Funds for the Central Universities of China (No.~2018110153).
This paper was previously circulated under the title ``\emph{Restricted Probabilistic Fixed Ballot Rules and Hybrid Domains}''
\href{https://ink.library.smu.edu.sg/soe_research/2342/}{\textcolor[rgb]{0.00,0.00,1.00}{SMU Economics and Statistics Working Paper Series, Paper No.~03-2020.}}}}
\author{{\normalsize Shurojit Chatterji\thanks{School of Economics, Singapore Management University, Singapore.},~
Souvik Roy\thanks{Economic Research Unit, Indian Statistical Institute, Kolkata.},~
Soumyarup Sadhukhan\thanks{Department of Computer Science and Automation, Indian Institute of Science},}\\
{\normalsize Arunava Sen\thanks{Economics and Planning Unit, Indian Statistical Institute, Delhi Center.}~ and
Huaxia Zeng\thanks{School of Economics, Shanghai University of Finance and Economics,
and the Key Laboratory of Mathematical Economics (SUFE), Ministry of Education, Shanghai 200433, China.}} }
\date{\normalsize \today}
\maketitle

\begin{abstract}
\noindent
We study a class of preference domains
that satisfies the familiar properties of minimal richness, diversity and no-restoration.
We show that a specific preference restriction, \emph{hybridness}, has been embedded in these domains so that the preferences are single-peaked at the ``extremes'' and unrestricted in the ``middle''.
We also study the structure of strategy-proof and unanimous Random Social Choice Functions on these domains.
We show them to be special cases of probabilistic fixed ballot rules (introduced by \citet*{EPS2002}).

\medskip
		
		
\end{abstract}

\newpage
\section{Introduction}
Our goal in this paper is two-fold.
The first is to investigate a class of preference domains by means of some simple
and well-known axioms, without specifying any preference restriction.
We call such domains {\it hybrid domains} since, loosely speaking, the preferences of these domains are revealed to be single-peaked at the ``extremes'' and unrestricted in the ``middle''.
Our second objective is to comprehensively
analyze the class of unanimous and strategy-proof random social choice functions (RSCFs) defined on these domains.
We refer to these RSCFs as \textit{Probabilistic Fixed Ballots Rules}.

Two familiar preference domains in the literature on mechanism design in voting environments are the complete domain and the domain of single-peaked preferences.
The complete domain arises when there are no a priori restrictions on preferences (see \citet{G1973}, \citet{S1975} and \citet{G1977}).
 Single-peaked preferences on the other hand, require more structure on the set of alternatives.
They arise naturally in a variety of situations such as preference aggregation \citep{B1948}, strategic voting \citep{M1980}, public facility allocation \citep{BG2012}, fair division \citep{S1991,BJN1997} and object assignment \citep{B2019}.

In order to discuss hybrid domains, we briefly recall the definition of single-peaked preferences. The set of  alternatives is a finite set $A = \{a_1, a_2, \dots, a_m\}$ which is endowed  with the prior order $a_1 \prec a_2 \prec \dots \prec a_m$.
A preference ordering over $A$ is \emph{single-peaked} if there exists a top-ranked alternative, say $a_k$, such that preferences decline when alternatives move ``farther  away'' from $a_k$.
For instance, if ``$a_r \prec a_s \prec a_k$ or $a_k  \prec a_s \prec a_r$''
, then $a_s$ is strictly preferred to $a_r$.
A preference is \emph{hybrid} if there exist \textit{threshold} alternatives $a_{\underline{k}}$ and $a_{\overline{k}}$ with $a_{\underline{k}} \prec a_{\overline{k}}$ such that preferences over the alternatives in the interval between $a_{\underline{k}}$ and $a_{\overline{k}}$  are ``unrestricted'' relative to each other, while preferences over other alternatives retain features of single-peakedness.
Thus, the set $A$ can be decomposed  into three parts:
left interval $L = \{a_1, \dots, a_{\underline{k}}\}$, right interval $R = \{a_{\overline{k}}, \dots, a_m\}$ and
middle interval $M = \{a_{\underline{k}}, \dots, a_{\overline{k}}\}$.
Formally, a preference is \emph{$(\underline{k},\overline{k})$-hybrid} if the following holds:
(i) for a voter whose best alternative lies in $L$ (respectively in $R$), preferences over alternatives in the set $L \cup R$ are conventionally single-peaked, while  preferences over alternatives in $M$ are  arbitrary  subject to the restriction that the best alternative in $M$ is the left threshold $a_{\underline{k}}$ (respectively, the right threshold  $a_{\overline{k}}$), and
(ii) for a voter whose  peak lies in $M$, preferences restricted to $L \cup R$ are single-peaked but arbitrary over $M$. Observe that if $\underline{k}=1$ and $\overline{k}=m$, then preferences are unrestricted, while the case where  $\overline{k}-\underline{k}=1$ coincides with the case of single-peaked preferences.

%A $(\underline{k},\overline{k})$-hybrid preference is a preference ordering
%which is single-peaked everywhere except over the alternatives in the middle interval between $a_{\underline{k}}$ and $a_{\overline{k}}$ on the prior order $\prec$.
Hybrid preferences can plausibly arise when an inherently multi-dimensional model is embedded in a one-dimensional model.
For instance, to address the structure of preferences over political parties according to different aspects of a party's platform, \citet{R2015} adopts the value-restricted preferences of \citet{S1966} to introduce a domain of \emph{multidimensional (locally) single-peaked preferences}, and
shows that the domain degenerates to a union of single-peaked domains w.r.t.~multiple distinct prior orders.
Indeed, Reffgen's multidimensional (locally) single-peaked preference is a special case of our hybrid preferences.
As a second instance, in Section \ref{sec:Hybrid}, we provide an example to show that a domain of \emph{multidimensional (globally) single-peaked preferences} introduced by \citet{BGS1993} reduces to a domain of our $(\underline{k},\overline{k})$-hybrid preferences in the framework of voting under constraints studied in \citet{BMN1997}.

Hybrid preferences also arise in the public facility allocation environment.
Imagine a transportation system in a city/region which consists of two suburban zones on ``opposite'' sides of a central urban zone. The central urban zone has a well-connected transportation system. It has two hubs on its boundaries, each of which is connected to a proximate suburban zone via a unique railway line. Consider the location of a good public facility in the city. Each citizen naturally prefers the public facility being located closer to his/her own location. Thus, a citizen living in a suburban zone has a single-peaked preference on locations along the railway line towards its hub, prefers the location of the hub to any other location in the central urban zone, and has a single-peaked preference on locations along the railway line towards the other suburban zone. A citizen living in the central urban zone can have arbitrary preferences over all locations in the central zone due to the denseness of the transportation system. She also prefers a location along a suburban railway line that is closer to its hub.
All these preferences can be expressed in terms of $(\underline{k}, \overline{k})$-hybridness according to the prior order $\prec$,
where $a_1$ and $a_m$ are the ``remotest'' suburban locations and
the thresholds $a_{\underline{k}}$ and $a_{\overline{k}}$ are the two hubs.

%Imagine a world where voters care about two aspects of a candidates' platform - her position on ``economic'' issues and her position on ``social issues''. Voters' opinions
%on each issue is single-peaked but they may weight the two issues differently. It is likely that candidates on both the left and right ``extremes'' have similar positions on both
%issues. On the other hand, candidates in the ``center'' may be left on one issue and right in the other. When viewed in a single dimension, left-right spectrum voters
%preferences on candidates will be hybrid with single-peakedness breaking down with respect to candidates in the center.

Observe that the domain of \emph{all} $(1,m)$-hybrid preferences is the unrestricted domain. It follows
that {\it any} domain is a subset of  \emph{the $(\underline{k}, \overline{k})$-hybrid domain} (i.e., the domain of \emph{all} $(\underline{k}, \overline{k})$-hybrid preferences) for some $1 \leq \underline{k}< \overline{k}\leq m$.\footnote{By finiteness of $A$, we can push the gap between $\overline{k}$ and $\underline{k}$ to the minimum, in the sense that the domain is not contained in the $(\underline{k}', \overline{k}')$-hybrid domain for any $\underline{k} \leq \underline{k}'< \overline{k}'\leq \overline{k}$ such that $\overline{k}'-\underline{k}' < \overline{k}-\underline{k}$.}
One of our main results (Theorem \ref{thm:domaincharacterization}) is that every domain that satisfies three familiar axioms must be \emph{a $(\underline{k}, \overline{k})$-hybrid domain} for some $1 \leq \underline{k}< \overline{k}\leq m$, i.e.,
a domain of $(\underline{k}, \overline{k})$-hybrid preferences that \emph{in addition}, satisfies two special properties in terms of the {\it strong connectedness} introduced in \citet{CSS2013}. The details of the two properties are described in Section \ref{sec:Hybrid}. We note here that the domain of all $(\underline{k}, \overline{k})$-hybrid preferences is of course a $(\underline{k}, \overline{k})$-hybrid domain, and
the domain of single-peaked preferences is a $(k,k+1)$-hybrid domain for all $1 \leq k < m$. The multiple single-peaked domain introduced by \citet{R2015} is also included in the class of our hybrid domains, but
the semi-single-peaked domain of \citet{CSS2013}, is not.

The three familiar axioms that we impose are {\it minimal richness}, {\it diversity} and {\it no-restoration}. The first axiom requires every alternative
be first-ranked by some preference in the domain.
The second requires the existence of two completely reversed preferences in the domain.
The third axiom requires that for each given pair of preferences and each given pair of alternatives,
there exists a path from one preference of the given pair to the other in the domain, by a sequence of specific preference switches,
such that the relative ranking of the given pair of alternatives is switched at most once along the path.
The first two axioms focus on the richness of a domain in question, while the third one not only is concerned with the richness of the domain (the existence of a path in the domain that reconciles the difference between two preferences), but also ensures that all preferences of the domain be well organized; for instance, if two given alternatives are identically ranked in two given preferences,
they must be identically ranked in every preference along one path connecting the two given preferences.\footnote{In the cardinal model, the valuation space is usually assumed to be a convex set.
Given a valuation vector $\theta_i \in \mathbb{R}^m$, let $\theta_{ia}$ be the valuation of the alternative $a$.
Given two valuation vectors $\theta_i, \theta_i'\in \mathbb{R}^{m}$,
consider all convex combinations of these two valuation vectors $\{\theta_i^t\equiv (1-t)\theta_i+t\theta_i': 0\leq t \leq 1\}$, which formulate a line in the valuation space to reconcile the difference between $\theta_i$ and $\theta_i'$.
One can observe that given $a, b \in A$,
if $\theta_{ia} \geq \theta_{ib}$ and $\theta_{ia}' \geq \theta_{ib}'$,
then $\theta_{ia}^t \geq \theta_{ib}^t$ for all $0\leq t\leq 1$;
if $\theta_{ia} > \theta_{ib}$ and $\theta_{ia}' < \theta_{ib}'$, then
there exists a unique $0<\hat{t}<1$ such that
$\theta_{ia}^t > \theta_{ib}^t$ for all $0\leq t < \hat{t}$,
$\theta_{ia}^{\hat{t}} = \theta_{ib}^{\hat{t}}$,
and $\theta_{ia}^s < \theta_{ib}^s$ for all $\hat{t} < s \leq 1$. Therefore, the axiom of no-restoration in our ordinal model is analogous to but weaker than the convex-set assumption of the valuation space in the cardinal model.}
No-restoration has been shown to be a key property for the equivalence between local and global strategy-proofness -- see \citet{S2013}, \citet{KRSYZ2021a, KRSYZ2021b} and \citet{HK2020}.
Moreover, no-restoration is also shown to be a sufficient condition that endogenizes the tops-only property in all unanimous and strategy-proof RSCFs -- see \citet{CZ2018} and \citet{KRSYZ2021b}.
Related conditions on domains such as {\it connectedness} have been used extensively in the literature on Condorcet domains \citep[e.g.,][]{M2009,P2018,PS2019}.
Unlike the restricted domains studied in the literature, the axiom of no-restoration does not exogenously specify any explicit preference restriction on the domain; for instance, a no-restoration domain can be as permissive as the unrestricted domain, or be as restrictive as a single-peaked domain or a single-crossing domain of \citet{S2009} or \citet{C2012}.
Therefore, no-restoration creates a unified framework to cover some important domains in the literature, and our Theorem \ref{thm:domaincharacterization} reveals the key common feature, hybridness, satisfied by all these domains.
An important feature of Theorem \ref{thm:domaincharacterization}  is that the condition on the domain does not specify an underlying structure of single-peakedness or threshold alternatives. These are derived endogenously from our hypotheses.

The paper also investigates the structure of unanimous and strategy-proof RSCFs on regular domains, i.e., domains satisfying the three aforementioned familiar axioms. A RSCF associates a lottery over alternatives to each profile of preferences.
Randomization is a way to resolve conflicts of interest by ensuring a measure of ex-ante fairness in the collective decision process.
More importantly, it has recently been shown that randomization significantly enlarges the scope of designing well-behaved mechanisms, e.g., the compromise RSCF of \citet{CSZ2014} and the maximal-lottery mechanism of \citet{BBS2016}.

In order to define the notion of strategy-proofness, we follow the standard approach of \citet{G1977}.   For every voter, truthfully revealing her preference ordering must yield a lottery that stochastically dominates the lottery arising from any unilateral  misrepresentation of preferences according to the sincere preference. Unanimity is a weak efficiency  requirement which says that the alternative that is unanimously best at a preference profile is selected with probability one.
	
According to Theorem \ref{thm:rulecharacterization}, a RSCF defined on a regular domain, which is further revealed to be a $(\underline{k}, \overline{k})$-hybrid domain, is unanimous and strategy-proof if and only if it is a \emph{Probabilistic Fixed Ballot Rule} (or PFBR).
A PFBR is specified by a collection of probability distributions $\beta_S$, where $S$ is a coalition of voters, over the set of alternatives.
We formally call $\beta_S$ a \emph{probabilistic ballot}.
If $\overline{k} - \underline{k}=1$, then a PFBR reduces to a \emph{fixed-probabilistic-ballots rule} that is originally introduced by \citet{EPS2002}.
However, if $\overline{k} - \underline{k}>1$, then a PFBR requires an additional  restriction on the probabilistic ballots:
each voter $i$ has a fixed probability weight $\varepsilon_i \geq 0$ such that the probability of the right interval $R$ according to $\beta_S$ is the total weight $\sum_{i \in S} \varepsilon_i$ of the voters in $S$ and that of the left interval $L$ is the total weight  $\sum_{i \notin S} \varepsilon_i$ of the voters outside $S$.
Recall the equivalence between local and global strategy-proofness and the tops-only property related to  the no-restoration domains.
It is worth mentioning that Theorem \ref{thm:rulecharacterization} decodes that
these two important properties emerge from the class of PFBRs on a no-restoration domain that in addition satisfies minimal richness and diversity, and
and more importantly provides a clear, operable and exhaustive guideline for a social planner's task of designing strategy-proof mechanisms.
	
The paper is organized as follows.
Section \ref{sec:literature} reviews the literature,
while Section \ref{sec:model} sets out the model and no-restoration domains.
Sections \ref{sec:Hybrid} and \ref{sec_charaterization} introduce hybrid preferences, PFBRs and the characterization results,
while Section \ref{sec:conclusion} concludes.

\subsection{Relationship with the Literature}\label{sec:literature}

	The analysis  of strategy-proof deterministic social choice functions on single-peaked domains was initiated by \citet{M1980} and developed further by \citet{BGS1993}, \citet{C1997} and \citet{W2011}. In the deterministic setting, \citet{NP2007}, \citet{CSS2013}, \citet{R2015}, \citet{CM2018}, \citet{AR2020} and \citet{BM2020} analyze the structure of unanimous and strategy-proof social choice functions on domains closely related to single-peakedness.

The structure of unanimous and strategy-proof RSCFs on single-peaked domains was first studied by \citet{EPS2002}. They considered the case where the set of alternatives is an interval in the real line and characterized the unanimous and strategy-proof RSCFs in terms of probabilistic fixed ballot rules.
Recently, in the case of finite alternatives, \citet{RS2019} establish the same characterization result on some proper subdomains of single-peaked preferences.\footnote{Given $m\geq 3$ alternatives, there are in total $2^{m-1}$ preferences in the single-peaked domain. The cardinality of a subdomain investigated in \citet{RS2019} can be reduced to $2(m-1)$.}
	
Recently, \citet{PRS2021} have considered the case where the set of alternatives is endowed with a graph structure.
Single-peakedness is defined w.r.t.~such graphs as in \citet{D1982} and \citet{CSS2013}.
\citet{PRS2021} investigate the structure of unanimous and strategy-proof RSCFs.
Their characterization result (Theorem 5.6 of \citet{PRS2021}) implies our Theorem \ref{thm:rulecharacterization} for a special graph structure.
However, the extension of our result in Theorem \ref{thm:rulecharacterization} is more general than their result since we do not assume a prespecified graph over the set of alternatives.
In particular, our result covers many domains that are excluded by theirs.
Finally, we emphasize that the motivation, formulation, and proof techniques in the two papers are completely different.

	\section{The Basic Model}\label{sec:model}

Let $A = \{a_1, a_2, \dots, a_m\}$ be a finite set of alternatives with $m \geq 3$.
Let $N = \{1, 2, \dots, n\}$ be a finite set of voters with $n \geq 2$. Each voter $i$ has a preference ordering $P_i$ (i.e., a complete, transitive and antisymmetric binary relation) over the alternatives.
We interpret $a_{s}P_{i}a_{t}$ as  ``$a_{s}$ is strictly preferred to $a_{t}$ according to $P_{i}$".
For each $1 \leq k \leq m$, $r_{k}(P_{i})$ denotes the $k$th ranked alternative in $P_{i}$.
We use the following notational convention:  $P_{i}\equiv (a_k\,a_s\,a_t\,\cdots) $ refers to a preference ordering (or a linear order) where  $a_k$ is first-ranked, $a_s$
is second-ranked, and $a_t$ is third-ranked, while the rest of the rankings in $P_{i}$ are arbitrary.
We denote the set of all preference orderings by $\mathbb{P}$, which we call \textbf{the complete domain}.
 A domain $\mathbb{D}$ is a subset of $\mathbb{P}$.
In particular, if $\mathbb{D} \subset \mathbb{P}$, it is called a \emph{restricted domain}.\footnote{Henceforth, let $\subset$ and $\subseteq$ denote the strict and weak inclusion relations respectively.}
A preference profile is an $n$-tuple of preferences, i.e., $P = (P_{1}, P_{2}, \dots, P_n) = (P_{i}, P_{-i}) \in \mathbb{D}^n$.
	
We consider domains that satisfy certain properties.

\begin{definition} \label{def:minrich}
A domain $\mathbb{D}$ is \textbf{minimally rich} if for each $a_k \in A$, there exists a preference $P_i \in \mathbb{D}$ such that $r_1(P_i) = a_k$.
\end{definition}

\begin{definition} \label{def:div}
A domain $\mathbb{D}$ satisfies \textbf{diversity} if there exist preferences $P_i, P_i' \in {\mathbb D}$ that are \textbf{completely  reversed} i.e. for all $a_s, a_t \in A$, we have $[a_s P_i a_t] \Leftrightarrow [a_tP_i'a_s]$.
\end{definition}
	
Throughout the paper, we will \emph{fix} two specific completely reversed preferences $\underline{P}_i$ and  $\overline{P}_i$ such that $a_k \underline{P}_i a_{k+1}$ and $a_{k+1} \overline{P}_i a_k$ for all $k = 1, \dots, m-1$.
Whenever a domain ${\mathbb D}$ is assumed to satisfy diversity, we will assume w.l.o.g.~that $\underline{P}_i, \overline{P}_i \in {\mathbb D}$.

A more subtle property that will be imposed on domains considered in the paper, is that of  \emph{no-restoration}. This condition was introduced by \citet{S2013}. As is well known, this condition is
 satisfied by many important voting domains, e.g., the complete domain, the single-peaked domain and the single-crossing domain.

Two distinct preferences $P_i, P_i'\in \mathbb{D}$ are \textit{adjacent}, denoted $P_i \sim P_i'$,
if there exist $a_s, a_t \in A$ such that (i) $r_k(P_i) = r_{k+1}(P_i') = a_s$ and $r_{k}(P_i') =r_{k+1}(P_i) =  a_t$ for some $1 \leq k \leq m-1$, and
(ii) $r_{l}(P_i) = r_{l}(P_i')$ for all $l \notin \{k, k+1\}$.
In other words, alternatives $a_s$ and $a_t$ are contiguously ranked in  $P_i$ and $P_i'$ and swapped across the two preferences,
while the ordering of all other alternatives is unchanged.
Given distinct preferences $P_i, P_i'\in \mathbb{D}$,
a sequence of non-repeated preferences $(P_i^1, \dots, P_i^v)$ is called a \textit{path} connecting $P_i$ and $P_i'$ if
$P_i^1 = P_i$, $P_i^v = P_i'$,
$P_i^k \in \mathbb{D}$ for all $1< k < v$, and $P_i^k \sim P_i^{k+1}$ for all $k = 1, \dots, v-1$.
%A domain $\mathbb{D}$ is said \emph{connected} if every pair of two distinct preferences is connected by a path in the domain.
Fixing two distinct preferences $P_i, P_i' \in \mathbb{D}$,
a path $\pi = (P_i^1, \dots, P_i^v)$ connecting $P_i$ and $P_i'$, and two distinct alternatives $a_s, a_t \in A$,
we say that $\pi$ has \textit{$\{a_s,a_t\}$-restoration} if there exist $1 \leq \kappa < \ell < \eta \leq v$ such that
either $a_s P_{\kappa} a_t$, $a_tP_{\ell} a_s$ and $a_sP_{\eta}a_t$, or $a_t P_{\kappa}a_s$, $a_sP_{\ell} a_t$ and $a_tP_{\eta}a_s$ occur.
%
%
%\begin{align*}
%[a_sP_i^{k^{\ast}}a_t \;\textrm{and}&\; a_tP_i^{k^{\ast}+1}a_s\; \textrm{for some}\; 1 \leq k^{\ast}< v] \\
%&\Rightarrow [a_sP_i^ka_t \; \textrm{for all}\; k=1, \dots, k^{\ast},\; \textrm{and}\; a_tP_i^{l}a_s\; \textrm{for all}\; l = k^{\ast}+1,\dots, v], \; \textrm{and}\\[0.3em]
%[a_tP_i^{k^{\ast}}a_s \;\textrm{and}&\; a_sP_i^{k^{\ast}+1}a_t\; \textrm{for some}\; 1 \leq k^{\ast}< v] \\
%&\Rightarrow [a_tP_i^ka_s \; \textrm{for all}\; k=1, \dots, k^{\ast},\; \textrm{and}\; a_sP_i^{l}a_t\; \textrm{for all}\; l = k^{\ast}+1,\dots, v].
%\end{align*}
Equivalently, \textbf{no $\{a_s,a_t\}$-restoration} implies that the relative ranking of $a_s$ and $a_t$ is switched  \emph{at most} once on the path $\pi$.
In particular, if $a_s$ and $a_t$ are identically ranked in $P_i$ and $P_i'$, then their relative ranking does not change along the path.

\begin{definition}\label{def:connectedness}
A domain $\mathbb{D}$ is a \textbf{no-restoration domain} if given distinct $P_{i}, P_{i}' \in \mathbb{D}$ and $a_s, a_t \in A$,
there exists a path connecting $P_{i}$ and $P_{i}'$ that has no $\{a_s,a_t\}$-restoration.
%Furthermore, domain $\mathbb{D}$ is a \textbf{non-restorated\textsuperscript{$+$} domain} if given distinct $P_{i}, P_{i}' \in \mathbb{D}$,
%there exists a path $\{P_{i}^{k}\}_{k=1}^{v} \subseteq \mathbb{D}$ connecting $P_{i}$ and $P_{i}'$ that has no $\{a_s,a_t\}$-restoration for all $a_s, a_t \in A$.
\end{definition}

For convenience, we will say that a domain is {\bf regular} if it satisfies minimal richness, diversity and no-restoration.

We now introduce random and deterministic social choice functions and some well-known properties related to them.
Let $\Delta(A)$ denote the space of all lotteries over $A$.
	An element $\lambda \in \Delta(A)$ is a lottery or a probability distribution over $A$,
	where $\lambda(a_k)$ denotes the probability received by alternative $a_{k}$.
	For notational convenience, we let $\bm{e}_{a_k}$ denote the  degenerate lottery where alternative $a_{k}$ receives probability one.
	A \textbf{Random Social Choice Function} (or RSCF) is a map $\varphi: \mathbb{D}^n \rightarrow \Delta(A)$
	which associates each preference profile to a  lottery.
	Let $\varphi_{a_k}(P)$ denote the probability assigned to $a_k$ by $\varphi$ at the preference profile $P$. If a RSCF selects a degenerate lottery at every preference profile, it is called a \textbf{Deterministic Social Choice Function} (or DSCF). More formally, a DSCF is a mapping $f: \mathbb{D}^n \to A$.
	
	In this paper, we impose two axioms on RSCFs: unanimity and strategy-proofness.
	
	\begin{definition} \label{def:unanimity}
	A RSCF $\varphi: \mathbb{D}^{n} \rightarrow \Delta(A)$ is \textbf{unanimous} if for all $P \in \mathbb{D}^{n}$ and $a_k \in A$,
	$[r_1(P_{i}) = a_k\; \textrm{for all}\; i \in N] \Rightarrow [\varphi(P) = \bm{e}_{a_k}]$.
	\end{definition}

	We adopt the first-order stochastic dominance  notion of strategy-proofness proposed by \citet{G1977}.
    It requires that the lottery from truthtelling  stochastically dominate the lottery obtained by any  misrepresentation by any voter at any possible profile of other voters' preferences.

\begin{definition} \label{def:sp}
A RSCF $\varphi: \mathbb{D}^{n} \rightarrow \Delta(A)$ is \textbf{strategy-proof} if for all $i \in N, P_i, P_i' \in \mathbb{D}$ and $P_{-i} \in \mathbb{D}^{n-1}$,
	$\varphi(P_i, P_{-i})$ stochastically dominates $\varphi(P_i', P_{-i})$ according to $P_i$, i.e.,
	$\sum_{t=1}^{k}\varphi_{r_{t}(P_{i})}(P_{i}, P_{-i}) \geq \sum_{t=1}^{k}\varphi_{r_{t}(P_{i})}(P_{i}', P_{-i})$ for all $k= 1, \dots, m$.
\end{definition}	

For future reference, we introduce the  tops-onlyness property of RSCFs. It requires the outcome at each preference profile to depend only on the top-ranked alternatives at that preference profile.

\begin{definition} \label{def:tops}
A  RSCF $\varphi: \mathbb{D}^n \rightarrow \Delta(A)$ satisfies the \textbf{tops-only property}
	if for all $P, P'\in \mathbb{D}^n$, we have $[r_1(P_i) = r_1(P_i')\; \textrm{for all}\; i \in N] \Rightarrow [\varphi(P) = \varphi(P')]$.
\end{definition}

An important class of unanimous and strategy-proof RSCFs are random dictatorships.
Formally, a RSCF $\varphi: \mathbb{D}^n \rightarrow \Delta(A)$ is a \textbf{random dictatorship}
if there exists  $\varepsilon_i \geq 0$ for each $i \in N$ with $\sum_{i \in N}\varepsilon_i = 1$
such that $\varphi(P) = \sum_{i \in N}\varepsilon_i\, \bm{e}_{r_1(P_i)}$ for all $P \in \mathbb{D}^n$.

We now proceed to our results.

	\section{Hybrid Domains and their Salience}
	\label{sec:Hybrid}

Our goal in this section is to show that domains that satisfy the properties of minimal richness, diversity and no-restoration have a specific structure.

An important restricted domain is the domain of  single-peaked preferences introduced by \citet{B1948}.
    We fix the \emph{natural order}~$\prec$ over $A$, i.e., $a_k \prec a_{k+1}$ for all $k = 1, \dots, m-1$.\footnote{For notational convenience, let $a_s \preceq a_t$ denote either $a_s \prec a_t$ or $a_s = a_t$. Henceforth, we write $\prec$ to denote the natural order.}
    Given $a_s, a_t \in A$ with $a_s \preceq a_t$, let $[a_s, a_t] = \{a_k \in A: a_s \preceq a_k\; \textrm{and}\; a_k \preceq a_t\}$ denote the interval of alternatives located between $a_s$ and $a_t$ on~$\prec$.
    A preference $P_i$ is \textbf{single-peaked} w.r.t.~$\prec$ if for all $a_s, a_t \in A$,
	we have $\left[a_s \prec a_t \prec r_1(P_i)\; \textrm{or}\;  r_1(P_i) \prec a_t \prec a_s\right] \Rightarrow [a_tP_ia_s]$.
	Let $\mathbb{D}_{\prec}$ denote \textbf{the single-peaked domain} which contains \emph{all} single-peaked preferences w.r.t.~$\prec$.
    The single-peaked domain is regular: Minimal richness and diversity follow immediately, and no-restoration is shown in Proposition 4.2 of \citet{S2013}.

We show that regular domains may be viewed as a weakening of single-peakedness
where the requirement of single-peakedness is violated over a subset of alternatives that lie in the ``middle'' of the natural order~$\prec$, and holds everywhere else.
We call this preference restriction ``hybridness'' to emphasize the coexistence of such violations with other features of single-peakedness.
	
Recall the natural order $\prec$ over $A$. Fix two alternatives $a_{\underline{k}}$ and $a_{\overline{k}}$ with $a_{\underline{k}} \prec a_{\overline{k}}$, which we refer to as the \emph{left threshold} and the \emph{right threshold}, respectively.
	We define three subsets of $A$ using these two thresholds:
	\textit{left interval} $L= [a_1, a_{\underline{k}}]$, \textit{right interval} $R = [a_{\overline{k}}, a_m]$ and
\textit{middle interval} $M  = [a_{\underline{k}}, a_{\overline{k}}]$.\footnote{Note that $L \cap M = \{a_{\underline{k}}\}$, $R \cap M = \{a_{\overline{k}}\}$ and $L \cap R = \emptyset$.}  In what follows, we present the structure of  a hybrid preference ordering.
	
\begin{figure}[t]
\centering
\includegraphics[width=0.9\textwidth]{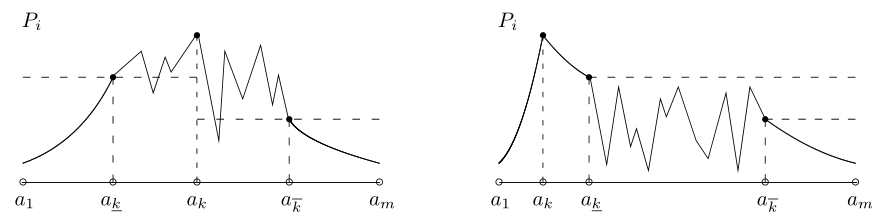}\\[-0.9em]
\caption{A graphic illustration of two hybrid preference orderings}\label{fig:illustration}
\end{figure}

	Consider a preference ordering  whose peak belongs to $M$ (see the first diagram of Figure \ref{fig:illustration}).
	The ranking of the  alternatives  in $M$ is completely arbitrary,
	while the ranking of the alternatives  in $L$ and $R$ follows the single-peakedness restriction w.r.t.~$\prec$.
    In other words, the only restriction  that the  preference ordering satisfies is that preference declines as one moves from $a_{\underline{k}}$ towards $a_1$, or  from $a_{\overline{k}}$ towards $a_m$.
	Note that this allows some alternatives in $L$ or $R$ be ranked above some alternatives in $M$.
	
	Next, consider a preference ordering  whose peak belongs to $L$ (see the second diagram of Figure \ref{fig:illustration}). The ranking of the alternatives in $L$ and $R$ follows single-peakedness w.r.t.~$\prec$. In other words, preference declines as one moves from the peak towards $a_1$ or  $a_{\underline{k}}$, or from $a_{\overline{k}}$ towards $a_m$.
	Furthermore, all alternatives in $M$  are ranked below $a_{\underline{k}}$ in an arbitrary manner.
	Notice that an alternative in $R$ may be ranked above some alternative in $M$, but  can never be ranked above $a_{\underline{k}}$.
	For a preference ordering  with the peak in $R$, the restriction is analogous.\medskip

	The formal definition of a hybrid preference is given below.
	
	\begin{definition}\label{defn:hybriddomain}
		Given the natural order $\prec$ and $1\leq \underline{k}< \overline{k}\leq m$,
   a preference $P_i$ is called a \bm{$(\underline{k},\overline{k})$}\textbf{-hybrid} preference
		if it satisfies the following two conditions:
\begin{itemize}
\item[\rm (i)] for all $a_r,a_s \in L$ or $a_r,a_s \in R$,\\
$[a_r \prec a_s \prec r_1(P_i)\; \textrm{or}\;\, r_1(P_i)\prec a_s\prec a_r] \Rightarrow [a_sP_ia_r]$, and
			
\item[\rm (ii)] $[r_1(P_i) \in L] \Rightarrow \big[a_{\underline{k}} P_i a_l \mbox{ for all } a_l \in M \backslash \{a_{\underline{k}}\}\big]$ and \\
$[r_1(P_i) \in R] \Rightarrow \big[a_{\overline{k}} P_i a_r \mbox{ for all } a_r \in M \backslash \{a_{\overline{k}} \}\big]$.
		\end{itemize}
Let $\mathbb{D}_{\prec} (\underline{k}, \overline{k})$ denote \textbf{the $\bm{(\underline{k}, \overline{k})}$-hybrid domain} containing all $(\underline{k}, \overline{k})$-hybrid preference orderings.
	\end{definition}

We here provide one example to explain how hybrid preferences arise in the model of voting under constraints studied in \citet{BMN1997}.

\begin{figure}[t]
\centering
\includegraphics[width=0.55\textwidth]{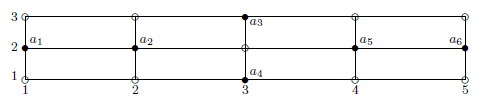}\\[-0.5em]
\caption{The Cartesian product set $X$ and the feasible set $A$}\label{fig:2lines}
\end{figure}

\begin{example}\rm
Given two component sets $X_1 = \{1, 2, 3, 4, 5\}$ and $X_2 = \{1, 2, 3\}$,
let $X = X_1 \times X_2$ denote a Cartesian product set.
Accordingly, we can locate all elements of $X$ on a grid of Figure \ref{fig:2lines}.
A preference $P_i$ over $X$, say $r_1(P_i) = x \equiv (x_1, x_2)$, is \emph{multidimensional single-peaked}  if for all distinct  $y\equiv (y_1, y_2) \in X$ and $z \equiv (z_1, z_2) \in X$,
we have $[z_k \leq y_k \leq x_k\; \textrm{or}\; x_k \leq y_k \leq z_k\; \textrm{for each}\; k \in \{1,2\}] \Rightarrow [yP_iz]$.
Meanwhile, let $A = \{a_1, a_2, a_3, a_4, a_5, a_6\} \subset X$ be the set of \emph{feasible} alternatives,
which are depicted by black nodes in Figure \ref{fig:2lines}.

Note that in a multidimensional single-peaked preference, (i) if $a_1$ is first-ranked,
then $a_2$ must be second-ranked within the feasible set $A$, and $a_5$ is preferred to $a_6$,
(ii) if $a_2$ is first-ranked, then $a_5$ is preferred to $a_6$,
and (iii) if $a_3$ is first-ranked, then $a_2$ is better than $a_1$, and $a_5$ is better than $a_6$.
Analogous preference restrictions over the ranking of feasible alternatives are observed for multidimensional single-peaked preferences with peaks $a_6, a_5$ and $a_4$.
These observations coincide with the preference restriction of $(2,5)$-hybridness
if we rearrange all feasible alternatives according to the natural order $\prec$.
In conclusion, when restrict attention to multidimensional single-peaked preferences whose peaks are feasible,
each induced preference over the feasible alternatives is a $(2,5)$-hybrid preference.
\hfill$\Box$
\end{example}

Indeed, the $(\underline{k}, \overline{k})$-hybrid domain $\mathbb{D}_{\prec} (\underline{k}, \overline{k})$ is a regular domain.\footnote{It is clear that the $(\underline{k}, \overline{k})$-hybrid domain $\mathbb{D}_{\prec} (\underline{k}, \overline{k})$ satisfies minimal richness and diversity.
The verification of no-restoration is put in Clarification 1 of Appendix \ref{app:clarification}.}
Note that $\mathbb{D}_{\prec} (\underline{k}, \overline{k})$ always includes the single-peaked domain $\mathbb{D}_{\prec}$, and reduces to the single-peaked domain $\mathbb{D}_{\prec}$ when $\overline{k}-\underline{k} = 1$.
In the other extreme case, $\overline{k}-\underline{k} = m-1$ (equivalently, $\underline{k} = 1$ and $\overline{k} = m$), the hybrid domain $\mathbb{D}_{\prec} (\underline{k}, \overline{k})$ expands  all the way to the complete domain $\mathbb{P}$.
This, in conjunction with finiteness of $A$, implies that for an arbitrary domain $\mathbb{D}$,
there always exist some $1 \leq \underline{k}< \overline{k} \leq m$ such that $\mathbb{D}\subseteq \mathbb{D}_{\prec} (\underline{k}, \overline{k})$, and $\mathbb{D} \nsubseteq (\underline{k}', \overline{k}')$ for any $\underline{k} \leq \underline{k}'< \overline{k}' \leq \overline{k}$ with
$\overline{k}'-\underline{k}'< \overline{k}-\underline{k}$.

We furthermore identify a class of domains of $(\underline{k}, \overline{k})$-hybrid preferences that satisfy a richness condition in terms of strong connectedness introduced by \citet{CSS2013}.
Given a domain $\mathbb{D}$,
two distinct alternatives $a_s, a_t \in A$ are said \textbf{strongly connected} (according to $\mathbb{D}$), denoted $a_s \approx a_t$, if
there exist $P_i, P_i' \in \mathbb{D}$ such that $r_1(P_i) = r_2(P_i')=a_s$, $r_1(P_i') = r_2(P_i) = a_t$ and $r_k(P_i) = r_k(P_i')$ for all $k = 3, \dots, m$.
One would immediately notice that strong connectedness has a close relation to the notion of adjacency between two preferences, and
it concentrates on the pair of adjacent preferences in the domain that differ in the peaks.
Furthermore, given an arbitrary non-empty subset $B \subseteq A$, we can induce a graph $G_{\approx}^B(\mathbb{D})$ such that
the vertex set is $B$, and two vertices of $B$ form an edge if and only if they are strongly connected.
A sequence of non-repeated vertices $(x_1, \dots, x_p)$ in $G_{\approx}^B(\mathbb{D})$ is called a \textbf{vertex-path} of $G_{\approx}^B(\mathbb{D})$ if
$x_{k} \approx x_{k+1}$ for all $k = 1, \dots, p-1$.
Accordingly, $G_{\approx}^B(\mathbb{D})$ is a \textbf{connected} graph if for all $a_s, a_t \in B$,
there exists a vertex-path $(x_1, \dots, x_p)$ connecting $a_s$ and $a_t$, i.e.,
$x_1 = a_s$, $x_p = a_t$ and $x_{k} \approx x_{k+1}$ for all $k = 1, \dots, p-1$.\footnote{As shown by Lemma \ref{lem:vertex-path} of Appendix \ref{app:domaincharacterization}, a minimally rich and no-restoration domain $\mathbb{D}$ induces a connected graph $G_{\approx}^A(\mathbb{D})$.
\citet{CSS2013} call a domain $\mathbb{D}$ a \textit{strongly path-connected} domain if $G_{\approx}^A(\mathbb{D})$ is a connected graph.}
Last, a vertex $a_s \in B$ is called a \textbf{leaf} in $G_{\approx}^B(\mathbb{D})$ if there exists a \emph{unique} $a_t \in B$ such that $a_s \approx a_t$.

\begin{definition}\label{def:hybrid*}
Given $1 \leq \underline{k}<\overline{k} \leq m$,
a domain $\mathbb{D} \subseteq \mathbb{D}_{\prec} (\underline{k}, \overline{k})$ is called \textbf{a \bm{$(\underline{k}, \overline{k})$}-hybrid domain},
if it satisfies the following two conditions:
\begin{itemize}
\item[\rm (i)] \textbf{Path-inclusion condition}:
for each $1< k < m$, $a_k$ is contained in a vertex-path connecting $a_1$ and $a_m$,
and

\item[\rm (ii)] \textbf{No-leaf condition}:
if $\overline{k}-\underline{k}>1$,
the subgraph $G_{\approx}^{M}(\mathbb{D})$ has no leaf.
\end{itemize}
Henceforth, a domain $\mathbb{D}$ is simply called \textbf{a hybrid domain}
if there exist $1 \leq \underline{k}< \overline{k} \leq m$ such that $\mathbb{D} \subseteq \mathbb{D}_{\prec} (\underline{k}, \overline{k})$, and $\mathbb{D}$ satisfies the path-inclusion and no-leaf conditions.
\end{definition}

We provide one example of a hybrid domain.

\begin{table}[t]
\begin{center}
\begin{tabular}{cccccccccccc}
  $P_1$ & $P_2$ & $P_3$ & $P_4$ & $P_5$ & $P_6$ & $P_7$ & $P_8$ & $P_9$ & $P_{10}$ & $P_{11}$ & $P_{12}$ \\
  $a_1$ & $a_2$ & $a_2$ & $a_2$ & $a_3$ & $a_3$ & $a_4$ & $a_4$ & $a_5$ & $a_5$    & $a_5$    & $a_6$    \\[-0.3em]
  $a_2$ & $a_1$ & $a_3$ & $a_4$ & $a_2$ & $a_5$ & $a_2$ & $a_5$ & $a_3$ & $a_4$    & $a_6$    & $a_5$    \\[-0.3em]
  $a_3$ & $a_3$ & $a_1$ & $a_1$ & $a_1$ & $a_2$ & $a_1$ & $a_2$ & $a_2$ & $a_2$    & $a_4$    & $a_4$    \\[-0.3em]
  $a_4$ & $a_4$ & $a_4$ & $a_3$ & $a_4$ & $a_1$ & $a_3$ & $a_1$ & $a_1$ & $a_1$    & $a_3$    & $a_3$    \\[-0.3em]
  $a_5$ & $a_5$ & $a_5$ & $a_5$ & $a_5$ & $a_4$ & $a_5$ & $a_3$ & $a_4$ & $a_3$    & $a_2$    & $a_2$    \\[-0.3em]
  $a_6$ & $a_6$ & $a_6$ & $a_6$ & $a_6$ & $a_6$ & $a_6$ & $a_6$ & $a_6$ & $a_6$    & $a_1$    & $a_1$
\end{tabular}
\caption{Domain $\mathbb{D}$}\label{tab:hybrid*}
\end{center}
\end{table}

\begin{figure}[t]
\centering
\includegraphics[width=0.42\textwidth]{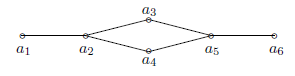}\\[-1em]
\caption{Graph $G_{\approx}^A(\mathbb{D})$}\label{fig:hybrid*}
\end{figure}

\begin{example}\label{exm:hybrid*}\rm
Fix $A = \{a_1, a_2, a_3, a_4, a_5, a_6\}$.
Let domain $\mathbb{D}$ consist of the 12 preferences in Table \ref{tab:hybrid*} and
graph $G_{\approx}^A(\mathbb{D})$ be specified in Figure \ref{fig:hybrid*}.
According to Table \ref{tab:hybrid*} and Figure \ref{fig:hybrid*},
one can easily tell that $\mathbb{D} \subseteq \mathbb{D}_{\prec}(2,5)$, and both the path-inclusion and no-leaf conditions are satisfied.
Therefore, $\mathbb{D}$ is a $(2, 5)$-hybrid domain.
Also, note that $\mathbb{D}$ is minimally rich and contains a pair of completely reversed preferences, $P_1$ and $P_{12}$. However, $\mathbb{D}$ is \emph{not} a regular domain, since it violates the no-restoration condition, for instance, there exists no path in $\mathbb{D}$ connecting $P_1$ and $P_{12}$.
\hfill$\Box$
\end{example}

\begin{remark}\rm
Fixing a $(\underline{k}, \overline{k})$-hybrid domain $\mathbb{D}$,
if $\overline{k}-\underline{k} =1$, condition (ii) of Definition \ref{def:hybrid*} is vacuously satisfied and $\mathbb{D}$ degenerates to a domain of single-peaked preferences w.r.t.~$\prec$;
if $\overline{k}-\underline{k} >1$, condition (ii) of Definition \ref{def:hybrid*} implies $\mathbb{D} \nsubseteq \mathbb{D}_{\prec}(k, k')$ for any $\underline{k} \leq k < k'< \overline{k}$ or $\underline{k} < k < k' \leq \overline{k}$.
\hfill$\Box$
\end{remark}

\begin{remark}\rm
The $(\underline{k}, \overline{k})$-hybrid domain $\mathbb{D}_{\prec}(\underline{k}, \overline{k})$ is naturally a $(\underline{k}, \overline{k})$-hybrid domain, since
either $\overline{k}-\underline{k}=1$ and $G_{\approx}^A\big(\mathbb{D}_{\prec}(\underline{k}, \overline{k})\big)$ is a line where all alternatives are located in the order of~$\prec$,
or $\overline{k}-\underline{k}>1$ and
$G_{\approx}^A\big(\mathbb{D}_{\prec}(\underline{k}, \overline{k})\big)$ is simply a combination
of the line $(a_1, \dots, a_{\underline{k}})$,
the subgraph $G_{\approx}^{M}\big(\mathbb{D}_{\prec}(\underline{k}, \overline{k})\big)$
where every pair of two alternatives of $M$ forms an edge, and the line $(a_{\overline{k}}, \dots, a_m)$.
In most cases, a $(\underline{k},\overline{k})$-hybrid domain is strictly contained in the $(\underline{k}, \overline{k})$-hybrid domain $\mathbb{D}_{\prec}(\underline{k},\overline{k})$.
\hfill$\Box$
\end{remark}

Our main result of this section is the following theorem,
which says that every regular domain is a hybrid domain.

\begin{theorem}\label{thm:domaincharacterization}
Let $\mathbb{D}$ be a regular domain, that is, a domain satisfying minimal richness, diversity and no-restoration. Then, $\mathbb{D}$ is a hybrid domain.
\end{theorem}

The proof of Theorem \ref{thm:domaincharacterization} is contained in Appendix \ref{app:domaincharacterization}.
Here, we provide a heuristic description of the proof.

An immediate implication of no-restoration, in conjunction with minimal richness, is that $G_{\approx}^A(\mathbb{D})$ is a connected graph.
According to the connected graph $G_{\approx}^A(\mathbb{D})$, we highlight two other important implications of no-restoration:
\begin{enumerate}
\item[\rm (a)] if an alternative is contained in a vertex-path connecting $a_1$ and $a_m$, then any other alternative that is strongly connected to it must be also contained in some vertex-path connecting $a_1$ and $a_m$,\footnote{This implication also relies on diversity.} and

\item[\rm (b)]
    given distinct $a_p, a_t, a_s \in A$, if $a_t$ is contained in \emph{every} vertex-path connecting $a_p$ and $a_s$, then $a_t$ is ranked above $a_s$ in \emph{every} preference with the peak $a_p$.
\end{enumerate}
The violation of the path-inclusion condition implies the violation of implication (a),
while the violation of the no-leaf condition implies the violation of implication (b),
which both indicate the failure of no-restoration.

First, let a domain $\mathbb{D}$ satisfies minimal richness and diversity, but violates the path-inclusion condition.
It is clear that either $G_{\approx}^A(\mathbb{D})$ is not a connected graph which directly implies the failure of no-restoration on $\mathbb{D}$, or $G_{\approx}^A(\mathbb{D})$ is a connected graph.
So let $G_{\approx}^A(\mathbb{D})$ be a connected graph.
By the violation of the path-inclusion condition, we can identify some $a_k$, where $1< k < m$, such that no vertex-path connecting $a_1$ and $a_m$, contains $a_k$.
Pick an arbitrary vertex-path connecting $a_1$ and $a_m$ and an arbitrary alternative $a_s$ contained in the vertex-path.
According to $G_{\approx}^A(\mathbb{D})$, we have a vertex-path $(y_1, \dots, y_q)$ connecting $a_s$ and $a_k$.
Since $a_s$ is contained in some vertex-path connecting $a_1$ and $a_m$, and
$a_k$ is not contained in any vertex-path connecting $a_1$ and $a_m$,
searching on the vertex-path $(y_1, \dots, y_q)$ from $y_1 = a_s$ to $y_q = a_k$,
we can identify some $1 \leq t < q$ such that
$y_t$ is contained in some vertex-path connecting $a_1$ and $a_m$,
while $y_{t+1}$ is not contained in any vertex-path connecting $a_1$ and $a_m$.
This contradicts implication (a).
Therefore, $\mathbb{D}$ is not a no-restoration domain.

Next, let a domain $\mathbb{D}$ satisfy minimal richness, diversity and the path-inclusion condition.
We first know that there exist $1\leq \underline{k}< \overline{k}\leq m$ such that
$\mathbb{D} \subseteq \mathbb{D}_{\prec}(\underline{k},\overline{k})$ and
$\mathbb{D} \nsubseteq \mathbb{D}_{\prec}(\underline{k}',\overline{k}')$ for any $\underline{k}\leq \underline{k}'< \overline{k}'\leq \overline{k}$ with $\overline{k}'-\underline{k}'< \overline{k}-\underline{k}$.\footnote{The proof of Theorem \ref{thm:domaincharacterization} provides the detail on identifying $a_{\underline{k}}$ and $a_{\overline{k}}$, using all vertex-path(s) connecting $a_1$ and $a_m$ in the graph $G_{\approx}^A(\mathbb{D})$.}
Suppose that $\mathbb{D}$ violates the no-leaf condition.
Then, it must be the case that $\overline{k}-\underline{k}>1$, and $G_{\approx}^M(\mathbb{D})$ has a leaf, where $M = [a_{\underline{k}}, a_{\overline{k}}]$.
It is clear that the restriction of $(\underline{k}, \overline{k})$-hybridness implies that $G_{\approx}^A(\mathbb{D})$ is simply a combination of the vertex-path $(a_1, \dots, a_{\underline{k}})$,
the connected graph $G_{\approx}^M(\mathbb{D})$, and the vertex-path $(a_{\overline{k}}, \dots, a_m)$.
Moreover, the path-inclusion condition implies that $a_k$ is never a leaf in $G_{\approx}^M(\mathbb{D})$ for all $\underline{k}< k < \overline{k}$. Thus, either $a_{\underline{k}}$ or $a_{\overline{k}}$ is a leaf in $G_{\approx}^M(\mathbb{D})$.
Assume w.l.o.g. that $a_{\underline{k}}$ is a leaf in $G_{\approx}^M(\mathbb{D})$.
Thus, we have a unique $a_k \in M\backslash \{a_{\underline{k}}\}$ such that $a_{\underline{k}} \approx a_k$.
Note that either $k> \underline{k}+1$ or $k = \underline{k}+1$ holds.
If $k > \underline{k}+1$, we have $a_{\underline{k}+1}$ contained in every vertex-path connecting $a_1$ and $a_k$, but $a_{\underline{k}+1}\underline{P}_ia_k$, which contradicts implication (b).
If $k = \underline{k}+1$,
according to $G_{\approx}^A(\mathbb{D})$ and the hypothesis $\mathbb{D} \subseteq \mathbb{D}_{\prec}(\underline{k},\overline{k})$ and $\mathbb{D} \nsubseteq \mathbb{D}_{\prec}(\underline{k}+1,\overline{k})$,
we can identify either alternatives $a_l \in [a_1, a_{\underline{k}}]$, $a_r \in [a_{\underline{k}+2}, a_{\overline{k}}]$ and a preference $P_i\in \mathbb{D}$ with $r_1(P_i) = a_l$ such that
$a_{\underline{k}+1}$ is of course contained in every vertex-path connecting $a_l$ and each $a_r$, but $a_rP_ia_{\underline{k}+1}$,
or an alternative $a_r \in [a_{\underline{k}+2}, a_{\overline{k}}]$ and a preference $P_i\in \mathbb{D}$ with $r_1(P_i) =a_r$ such that $a_{\underline{k}+1}$ is of course contained in every vertex-path connecting $a_r$ and each $a_{\underline{k}}$, but $a_{\underline{k}}P_ia_{\underline{k}+1}$.
This contradicts implication (b).
Therefore, domain $\mathbb{D}$ is not a no-restoration domain.

\begin{remark}\rm
We motivate the analysis provided in this section in the following manner.
The idea of single-peakedness has proved fundamental to the resolution of incentive problems.
It however takes as given an underlying order on alternatives.
A central question is ``where does the underlying order come from?''.
One approach to address this question is via investigating well-behaveness axioms on strategy-proof DSCFs or RSCFs acting on suitably rich preference domains -- see \citet{CSS2013,CSZ2016}.
Our alternative approach takes as its primitive the notion of connectedness (introduced by \citet{G1978} and \citet{M2009}).
 This idea has proved influential, with versions of it being utilized extensively in the recent literature on incentives and aggregation.
We pursue the idea that connectedness (no-restoration is a version of connectedness) by itself can provide insights into the structure of domains.
Our Theorem \ref{thm:domaincharacterization} directly investigates the effects of connectedness axioms on domains -- we show that
if a domain has nice connectivity properties and is therefore nicely structured,
a geometric structure over alternatives emerges. This structure is more permissive than the natural order $\prec$ in the middle part, and coincides with the natural order $\prec$ at the extremes.
\hfill$\Box$
\end{remark}

Recent literature has shown that the full power of single-peakedness is not needed to resolve problems of incentives. \citet{CSS2013} has introduced one weakening of single-peakedness, \emph{semi-single-peakedness},
which retains single-peakedness restriction structure on middle part and relaxes it on the extremes.
The preference restriction of hybridness introduced in this paper does the opposite, and can be viewed as a complementary approach, which for instance covers the multiple single-peaked domains of \citet{R2015} as special cases.\medskip

\begin{figure}[t]
\centering
\includegraphics[width=0.5\textwidth]{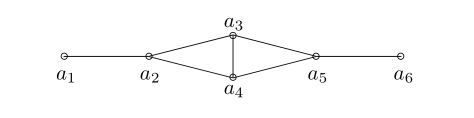}\\[-1em]
\caption{Graph $G_{\approx}^A(\mathbb{D}_{\Omega})$}\medskip\label{fig:multipleSP}
\end{figure}

\noindent \textbf{Multiple single-peaked domains:}
Every multiple single-peaked domain is a hybrid domain.
The detailed verification is relegated to Clarification 2 of Appendix \ref{app:clarification}.
We here introduce an example to illustrate.

Let $A = \{a_1, a_2, a_3, a_4, a_5, a_6\}$ and
$\Omega =  \{\prec_1, \prec_2\}$ be a collection of the  linear orders $\prec_1 = (a_1\,a_2\,a_3\,a_4\,a_5\,a_6)$ and $\prec_2 = (a_1\,a_2\,a_4\,a_3\,a_5\,a_6)$.
Note that $\prec_1$ and $\prec_2$ share \emph{the left maximum common part} $L_{\Omega} = (a_1 \,a_2)$ and
\emph{the right maximum common part} $R_{\Omega} = (a_5\,a_6)$.
A \emph{multiple single-peaked domain} $\mathbb{D}_{\Omega}$ is a union of the two single-peaked domains $\mathbb{D}_{\prec_1}$ and $\mathbb{D}_{\prec_2}$, i.e., $\mathbb{D}_{\Omega} = \mathbb{D}_{\prec_1}\cup \mathbb{D}_{\prec_2}$.\footnote{Note that the multiple single-peaked domain $\mathbb{D}_{\Omega}$ here is also a no-restoration domain.}

According to the natural order $\prec$ over $A$, selecting $a_2$ and $a_5$ as two thresholds, we have the $(2,5)$-hybrid domain $\mathbb{D}_{\prec}(2,5)$.
It is easy to verify that $\mathbb{D}_{\Omega} \subseteq \mathbb{D}_{\prec}(2,5)$.
The graph $G_{\approx}^A(\mathbb{D}_{\Omega})$ is specified in Figure \ref{fig:multipleSP},
which clearly indicates the path-inclusion condition and the no-leaf condition (i.e., $G_{\approx}^{[a_2, a_5]}(\mathbb{D}_{\Omega})$ has no leaf).
Therefore, $\mathbb{D}_{\Omega}$ is a $(2,5)$-hybrid domain.
It is worth mentioning that $\mathbb{D}_{\Omega}$ is ``strictly" contained in the $(2,5)$-hybrid domain $\mathbb{D}_{\prec}(2,5)$, e.g., the following two $(2,5)$-hybrid preferences are excluded by both $\mathbb{D}_{\prec_1}$ and $\mathbb{D}_{\prec_2}$: $P_i = (a_2\,a_5\,\cdots)$ and
$P_i' = (a_1\,a_2\,a_3\,a_5\,a_4\,a_6)$.

\medskip
	\noindent \textbf{Semi-single-peaked domains:}
	Consider the natural order  $\prec$ and fix \emph{one} threshold alternative. The semi-single-peakedness restriction on a preference requires that
	(i) the single-peakedness restriction prevail in the interval between the peak and the threshold, and
	(ii) each  alternative located beyond the threshold be ranked below the threshold.
	
	One can extend the semi-single-peakedness notion by adding more thresholds and requiring preferences to be semi-single-peaked w.r.t.~each threshold alternative.
	In particular, suppose that there are two distinct thresholds $a_{\underline{k}}$ and $a_{\overline{k}}$ with $a_{\underline{k}} \prec a_{\overline{k}}$.
	Consider a preference $P_i$ with $a_{\underline{k}} \prec r_1(P_i) \prec a_{\overline{k}}$.
	If $P_i$ is $(\underline{k}, \overline{k})$-hybrid,
	then the usual single-peakedness restriction prevails on the left and right intervals, and
	no restriction is imposed on the ranking of the alternatives in the middle interval (see the first diagram of Figure \ref{fig:ESPvsSSP}).
	On the contrary, if $P_i$ is semi-single-peaked w.r.t.~both $a_{\underline{k}}$ and $a_{\overline{k}}$,
	then the single-peakedness restriction prevails on the middle interval but fails on the left and right intervals (see the second diagram of Figure \ref{fig:ESPvsSSP}). Thus, the notions of  hybridness and  semi-single-peakedness are not entirely compatible with each other.
In fact, the semi-single-peaked domain is never a hybrid domain,\footnote{Let $\mathbb{D}_{\prec}(a_{s})$ be the
semi-single-peaked domain which contains all semi-single-peaked preferences on~$\prec$ w.r.t.~the threshold $a_s$.
Indeed, the graph $G_{\approx}^A\big(\mathbb{D}_{\prec}(a_s)\big)$ is a line over $A$ where all alternatives are located in the order of~$\prec$.
This meets the path-inclusion condition.
First, it is clear that $\mathbb{D}_{\prec}(a_{s})$ is never a domain of single-peaked preferences w.r.t.~$\prec$.
Then, consider arbitrary $1 \leq \underline{k}< \overline{k} \leq m$ such that $\overline{k}-\underline{k}>1$.
However, both $a_{\underline{k}}$ and $a_{\overline{k}}$ are leaves of
the subgraph $G_{\approx}^{[a_{\underline{k}}, a_{\overline{k}}]}\big(\mathbb{D}_{\prec}(a_s)\big)$.
This contradicts the no-leaf condition according to $a_{\underline{k}}$ and $a_{\overline{k}}$.}
which further implies by Theorem \ref{thm:domaincharacterization} that it is never a no-restoration domain.

\begin{figure}[t]
\centering
\includegraphics[width=0.9\textwidth]{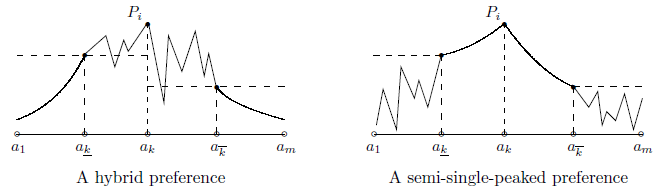}\\[-0.5em]
\caption{A hybrid preference  \emph{v.s.} a semi-single-peaked preference}\label{fig:ESPvsSSP}
\end{figure}

\section{Probabilistic Fixed Ballot Rules and\\ A Characterization}\label{sec_charaterization}
			
Our goal in this section is to provide a characterization of RSCFs defined over regular domains that satisfy unanimity and strategy-proofness.
For this purpose, we introduce the class of Probabilistic Fixed Ballot Rules (or PFBRs).

A PFBR is established on the natural order $\prec$ and is based on a collection of pre-specified parameters $(\beta_S)_{S \subseteq N}$, called \textbf{probabilistic ballots}.
Each probabilistic ballot $\beta_S$, which is associated to the coalition $S\subseteq N$, is a probability distribution over $A$, i.e., $\beta_S \in \Delta(A)$ for all $S \subseteq N$.
Given a coalition $S \subseteq N$,
for a subset $B \subseteq A$, let $\beta_S(B) = \sum_{a_k \in B}\beta_S(a_k)$ denote the probability of $B$ according to $\beta_S$.
Given a preference profile $P$, for each $1 \leq k \leq m$, let $S(k,P) = \{i\in N: a_k \preceq r_1(P_i)\}$ be the set of voters whose peaks are \emph{not} to the left of $a_k$.

\begin{definition}\label{def:PFBR}
A RSCF $\varphi: \mathbb{D}^n \rightarrow \Delta(A)$ is called a \textbf{Probabilistic Fixed Ballot Rule} (or \textbf{PFBR})
if there exists a collection of probabilistic ballots $(\beta_S)_{S \subseteq N}$ satisfying the following two properties:
\begin{itemize}
\item \textbf{Ballot unanimity}:  $\beta_N = \bm{e}_{a_m}$ and $\beta_{\emptyset} = \bm{e}_{a_1}$, and

\item \textbf{Monotonicity}: for all $a_k \in A$, $[S \subset T \subseteq N] \Rightarrow \big[\beta_S([a_k,a_m]) \leq \beta_T([a_k,a_m])\big]$,\footnote{Monotonicity implies that the bigger the coalition is, the higher probability the probabilistic ballot imposes on an interval towards $a_m$ in the natural order $\prec$.}

\end{itemize}
such that for all $P \in \mathbb{D}^{n}$ and $a_k \in A$, we have
\begin{align*}
\varphi_{a_k}(P) = \beta_{S(k, P)}([a_k, a_m]) - \beta_{S({k+1}, P)}([a_{k+1}, a_m]),
\end{align*}
where $S({m+1}, P) \equiv \emptyset$, $[a_{m+1}, a_m] \equiv \emptyset$ and hence $\beta_{S({m+1}, P)}([a_{m+1}, a_m])= 0$.
\end{definition}

\begin{remark}\rm
A PFBR $\varphi: \mathbb{D}^n \rightarrow \Delta(A)$ is a well defined RSCF:
Fix a preference profile $P \in \mathbb{D}^n$.
First, for all $a_k \in A$, since $S(k+1, P) \subseteq S(k,P)$ and $[a_{k+1}, a_m] \subset [a_k, a_m]$, monotonicity ensures
$\varphi_{a_k}(P) = \beta_{S(k,P)}([a_k,a_m]) - \beta_{{S}(k+1,P)}([a_{k+1},a_m]) \geq 0$.
Next, it is clear that $\sum_{k = 1}^m\varphi_{a_k}(P) = \sum_{k=1}^m
\big[\beta_{S(k,P)}([a_k,a_m]) - \beta_{{S}(k+1,P)}([a_{k+1},a_m]) \big]= \beta_{S(1, P)}([a_1, a_m]) = 1$.
Moreover, if $r_1(P_1) = \dots =r_1(P_n) \equiv a_k$, we have $S(k,P) = N$ and $S(k+1, P) = \emptyset$. Then, ballot unanimity implies $\varphi_{a_k}(P) = \beta_{S(k,P)}([a_k, a_m]) - \beta_{S(k+1,P)}([a_{k+1}, a_m]) = 1$.
Therefore, $\varphi$ satisfies unanimity.
Last, the construction of a PFBR clearly indicates that it satisfies the tops-only property.
\hfill$\Box$
\end{remark}

\begin{remark}\rm
\citet{EPS2002} introduced the class of PFBRs in the case of continuum of alternatives (for instance, the interval $[0,1]$ or the real line $\mathbb{R}$) and showed that a RSCF is unanimous and strategy-proof on the single-peaked domain if and only if it is a PFBR.
The deterministic version of PFBRs can be obtained by additionally requiring the probabilistic ballots be degenerate, i.e., $\beta_S(a_k) \in \{0,1\}$ for all $S \subseteq N$ and $a_k \in A$. We refer to these DSCFs as  \textbf{Fixed Ballot Rules} (or \textbf{FBRs}).\footnote{\citet{M1980} called these Augmented Median Voter Rules, while \citet{BGS1993} called these Generalized Median Voter Schemes. For an FBR, the subtraction form in Definition \ref{def:PFBR} can be reduced to a max-min form \citep[see Definition 10.3 in][]{NRTV2007book}. \citet{M1980} originally defined an augmented median voter rule in the min-max form,
which can be equivalently translated to a max-min form.}	
\citet{M1980} also considered the case of continuum of alternatives and showed that on the single-peaked domain,
a DSCF is unanimous, tops-only and strategy-proof  if and only if it is an FBR.
%In the finite setting, it can be easily verified that
%a mixture of FBRs is a unanimous and strategy-proof RSCF on the single-peaked domain $\mathbb{D}_{\prec}$, and is a PFBR.
%Furthermore, \citet{PRSS2014} and \citet{PU2015} prove that the converse is also true.
\hfill$\Box$
\end{remark}

We provide an example to illustrate a PFBR.

\begin{table}[t]
\centering
\begin{tabular}{|p{0.8cm}<{\centering}|p{0.8cm}<{\centering}|p{0.8cm}<{\centering}|p{0.8cm}<{\centering}|p{0.8cm}<{\centering}|}
					\hline
					& $\beta_{\emptyset}$ & $\beta_{\{1\}}$ & $\beta_{\{2\}}$  & $\beta_{N}$\\[0.2em]\hline
					$a_1$ & $1$           & $0.5$           & $0.4$            & $0$  \\\hline
					$a_2$ & $0$           & $0.2$           & $0.3$            & $0$  \\\hline
					$a_3$ & $0$           & $0.1$           & $0.2$            & $0$  \\\hline
					$a_4$ & $0$           & $0.2$           & $0.1$            & $1$  \\\hline
\end{tabular}
\caption{The probabilistic ballots $(\beta_S)_{S\subseteq N}$}\label{tab:0}
\end{table}	

\begin{example}\label{exm:PFBR}\rm
Let $N = \{1,2\}$ and $A = \{a_1,a_2,a_3,a_4\}$.
We fix a choice of probabilistic ballots in Table \ref{tab:0}.
It is easy to verify that both ballot unanimity and monotonicity are satisfied.
The PFBR $\varphi$ w.r.t.~probabilistic ballots of Table \ref{tab:0} works as follows.
Consider an arbitrary preference profile $P = (P_1, P_2)$ and an arbitrary alternative $a_k$.
We first identify the sets $S(k,P)$ and $S(k+1, P)$, and the probabilities $\beta_{S(k,P)}([a_k,a_4])$ and $\beta_{{S}(k+1,P)}([a_{k+1},a_4])$.
Then, the probability of the alternative $a_k$ selected at $P$ is defined as the difference between these two probabilities.
For instance, let $r_1(P_1)=a_2$ and $r_1(P_2)=a_4$. Then, we calculate
\begin{align*}
& \varphi_{a_4}(P) = \beta_{S(4, P)}([a_4,a_4]) - 0 = \beta_{\{2\}}([a_4,a_4])=0.1,\\
&\varphi_{a_3}(P) = \beta_{S(3, P)}([a_3,a_4]) - \beta_{S(4, P)}([a_4,a_4]) = \beta_{\{2\}}([a_3,a_4]) - \beta_{\{2\}}([a_4,a_4]) =0.2, \\
& \varphi_{a_2}(P) = \beta_{S(2, P)}([a_2,a_4])  - \beta_{S(3,P)}([a_3,a_4]) =\beta_{N}([a_2,a_4])  - \beta_{\{2\}}([a_3,a_4]) =0.7, \;\textrm{and}\\
& \varphi_{a_1}(P) =\beta_{S(1,P)}([a_1,a_4])  - \beta_{S(2, P)}([a_2,a_4]) = \beta_{N}([a_1,a_4])  - \beta_{N}([a_2,a_4]) =0.
~~~~~~~\Box
\end{align*}
\end{example}

%It is worth mentioning that the probabilistic ballot $\beta_S$ for a coalition $S\subseteq N$ represents the outcome of $\varphi$ at the ``boundary profile''
%where voters in $S$ have the preference $\overline{P}_i=(a_m\, \cdots\, a_k\,a_{k-1}\,\cdots\,  a_1)$, while the others have the preference $\underline{P}_i=(a_1\, \cdots\, a_{k-1}\,a_k\, \cdots\,  a_m)$.
%For ease of presentation, we call such a preference profile a \emph{$S$-boundary profile}.\footnote{Note that for every $S \subseteq N$, there is a unique $S$-boundary profile.} If a PFBR $\varphi$ is unanimous, then it follows that $\beta_{\emptyset}$ assigns probability $1$ to $a_1$  and $\beta_{N}$ assigns probability $1$ to $a_m$, which in turn implies ballot unanimity.
%In what follows, we argue that if $\varphi$ is strategy-proof, then $(\beta_S)_{S \subseteq N}$ must be monotonic.
%Consider a proper subset $S \subset N$ and $i \in N \setminus S$. Let $P$ and $P'$ be the $S$-boundary and  $S \cup \{i\}$-boundary profiles, respectively.
%In other words, only voter $i$ changes her preference $\overline{P}_i$ in the $S\cup \{i\}$-boundary profile to $\underline{P}_i$. Strategy-proofness of $\varphi$ implies that the probability of each upper contour set of $\overline{P}_i$ is weakly increased from $\varphi(P)$ to $\varphi(P')$.
%Since the interval $[a_k,a_m]$ coincides with the upper contour set of $a_k$ at $\overline{P}_i$, it follows that $\beta_{S}([a_k,a_m]) \leq \beta_{S \cup \{i\}}([a_k,a_m])$. Monotonicity of $(\beta_S)_{S \subseteq N}$ follows from the repeated application of this argument.

According to the PFBR specified in Example \ref{exm:PFBR}, we observe that it becomes manipulable on the hybrid domain $\mathbb{D}_{\prec}(2,4)$.
For instance, given two $(2,4)$-hybrid preferences $P_1 = (a_2\,a_4\,a_3\,a_1)$ and $P_2 = (a_4\,a_2\,a_3\,a_1)$,
we calculate $\varphi(P_1, P_2) = 0.7\bm{e}_{a_2}+0.2\bm{e}_{a_3}+0.1\bm{e}_{a_4}$, where that the alternative $a_3$ that is not the top of $P_1$ or $P_2$ receives some strictly positive probability.
However, voter 2 can manipulate via $P_2' = (a_2\,\cdots)$ and obtain $\varphi(P_1, P_2') = \bm{e}_{a_2}$ which stochastically dominates $\varphi(P_1, P_2)$ according to $P_2$.
This indicates that some additional condition must be imposed on the probabilistic ballots to restore strategy-proofness of PFBR on the hybrid domain.

\begin{definition}\label{defn:PFBR}
Given $1\leq \underline{k}<\overline{k}\leq m$, a PFBR $\varphi$ w.r.t.~probabilistic ballots $(\beta_S)_{S \subseteq N}$
is called a \textbf{$\bm{(\underline{k}, \overline{k})}$-PFBR} if all probabilistic ballots in addition satisfy the \textbf{constrained random-dictatorship condition} w.r.t.~$a_{\underline{k}}$ and $a_{\overline{k}}$, i.e.,
there is a ``conditional dictatorial coefficient'' $\varepsilon_i \geq 0$ for each $i \in N$ with $\sum_{i \in N} \varepsilon_i =1$ such that for all $S \subseteq N$, $\beta_S([a_{\overline{k}},a_m]) = \sum_{ i \in S} \varepsilon_i$ and $\beta_S([a_1,a_{\underline{k}}\,]) = \sum_{ i \in N \setminus S} \varepsilon_i$.
\end{definition}

Note that the constrained random-dictatorship condition implies that \emph{no} probabilistic ballot assigns positive probability to an alternative that lies  (strictly) between $a_{\underline{k}}$ and $a_{\overline{k}}$, i.e., $\beta_S(a_k)=0$ for all $S \subseteq N$ and $a_k \in [a_{\underline{k}+1}, a_{\overline{k}-1}]$.
We call it the constrained random-dictatorship condition to emphasize that it ensures
a PFBR behave like a random dictatorship on the middle interval $[ a_{\underline{k}}, a_{\overline{k}}]$ when
all voters' preference tops are located in $[a_{\underline{k}}, a_{\overline{k}}]$, i.e.,
$\varphi(P)=\sum_{i\in N}\varepsilon_i\, \bm{e}_{r_1(P_i)}$ for each preference profile $P$ such that $r_1(P_i) \in [a_{\underline{k}},a_{\overline{k}}]$ for all $i \in N$. Consequently, when all voters' preference tops are located in $[a_{\underline{k}}, a_{\overline{k}}]$, no alternative other than voters' preference tops can receive any positive probability.
In the extreme case where $\underline{k}=1$ and $\overline{k}=m$, a PFBR satisfying the constrained random-dictatorship condition reduces to a random dictatorship.

In what follows, we present an example of a PFBR that satisfies the constrained random-dictatorship condition.

\begin{table}[t]
\centering				 \begin{tabular}{|p{0.8cm}<{\centering}|p{0.8cm}<{\centering}|p{0.8cm}<{\centering}|p{0.8cm}<{\centering}|p{0.8cm}<{\centering}|p{0.8cm}<{\centering}|p{0.8cm}<{\centering}|p{0.8cm}<{\centering}|p{0.8cm}<{\centering}|}
					\hline
					& $\beta_{\emptyset}$ & $\beta_{\{1\}}$ & $\beta_{\{2\}}$ & $\beta_{\{3\}}$ & $\beta_{\{1,2\}}$ & $\beta_{\{1,3\}}$ & $\beta_{\{2,3\}}$ & $\beta_{N}$\\[0.2em]\hline
					$a_1$ & $1$                 & $1/3$           & $1/3$           & $1/3$           & $1/3$             & $1/3$           & $1/3$            & $0$\\\hline
					$a_2$ & $0$                 & $1/3$           & $1/3$           & $1/3$           & $0$               & $0$             & $0$             & $0$\\\hline
					$a_3$ & $0$                 & $0$             & $0$             & $0$             & $0$               & $0$             & $0$               & $0$\\\hline
					$a_4$ & $0$                 & $0$             & $0$             & $0$             & $1/3$             & $1/3$           & $1/3$               & $0$\\\hline
					$a_5$ & $0$                 & $1/3$           & $1/3$           & $1/3$           & $1/3$             & $1/3$           & $1/3$            & $1$\\\hline
				\end{tabular}
\caption{The probabilistic ballots $(\beta_S)_{S\subseteq N}$}\label{tab:1}
\end{table}

\begin{example}\label{exm:CRD}\rm
Let $N= \{1,2,3\}$ and $A = \{a_1,a_2,a_3,a_4,a_5\}$.
Consider the 8 probabilistic ballots specified in Table \ref{tab:1}.
Clearly, both ballot unanimity and monotonicity are satisfied.
Given all voters' conditional dictatorial coefficients $\varepsilon_1=\varepsilon_2=\varepsilon_3=\frac{1}{3}$,
it is easy to verify the constrained random-dictatorship condition w.r.t.~$a_2$ and $a_4$, i.e.,
$\beta_S([a_4,a_5]) = \sum_{i \in S} \varepsilon_i$ and $\beta_S([a_1,a_2]) = \sum_{i \in N\setminus S} \varepsilon_i$ for all $S \subseteq N$.
Therefore, the corresponding PFBR is a $(2,4)$-PFBR.
\hfill$\Box$
\end{example}

We now present our characterization result for unanimous and strategy-proof RSCFs defined on regular domains.
	
\begin{theorem}\label{thm:rulecharacterization}
Let $\mathbb{D}$ be a regular domain. The following two statements hold:
\begin{itemize}
\item[\rm (i)] Domain $\mathbb{D}$ is a $(\underline{k},\overline{k})$-hybrid domain for some $1 \leq \underline{k} < \overline{k} \leq m$.

\item[\rm (ii)] A RSCF $\varphi: \mathbb{D}^n \rightarrow \Delta(A)$ is unanimous and strategy-proof if and only if it is a PFBR and moreover, a $(\underline{k}, \overline{k})$-PFBR when $\overline{k}-\underline{k}>1$.
\end{itemize}
\end{theorem}

We present a formal proof of Theorem \ref{thm:rulecharacterization} in Appendix \ref{app:rulecharacterization}.
Here, we provide some intuitive explanations.

First, statement (i) of Theorem \ref{thm:rulecharacterization} follows exactly from Theorem \ref{thm:domaincharacterization}, and reveals the natural order $\prec$, the thresholds $a_{\underline{k}}$ and $a_{\overline{k}}$, and the preference restriction of $(\underline{k}, \overline{k})$-hybridness that are instruments for the characterization of unanimous and strategy-proof RSCFs in statement (ii).

Second, if statement (i) of Theorem \ref{thm:rulecharacterization} holds at some $\overline{k}-\underline{k} =1$, we know that $\mathbb{D}$ reduces to a single-peaked domain w.r.t.~$\prec$, and all alternatives of the graph $G_{\approx}^A(\mathbb{D})$ are located in the order of $\prec$.
Then, statement (ii) follows from Theorem 1 of \citet{RS2019}.

Third, the proof of Theorem \ref{thm:rulecharacterization} focuses on the verification of statement (ii) in the case that $\overline{k}-\underline{k}>1$.
To show the sufficiency part, we first refer to Theorem 1 of \citet{KRSYZ2021b}, and note that on domain $\mathbb{D}$, as a no-restoration domain, every unanimous and locally strategy-proof RSCF is strategy-proof.\footnote{Formally, a RSCF $\psi: \mathbb{D}^{n} \rightarrow \Delta(A)$ is \textbf{locally strategy-proof} if for all $i \in N$, $P_i, P_i' \in \mathbb{D}$ with $P_i \sim P_i'$ and $P_{-i} \in \mathbb{D}^{n-1}$, $\psi(P_i, P_{-i})$ stochastically dominates $\psi(P_i', P_{-i})$ according to $P_i$.\label{footnote:localSP}}
Therefore,
the verification of strategy-proofness of a $(\underline{k}, \overline{k})$-PFBR $\varphi$ on $\mathbb{D}$ is enormously simplified to a verification of local strategy-proofness.
Furthermore, on domain $\mathbb{D}$, as a $(\underline{k}, \overline{k})$-hybrid domain,
local strategy-proofness of $\varphi$ intuitively follows from the two observations below:
\begin{itemize}
\item[\rm (i)] at each preference profile, $\varphi$ assigns an alternative $a_k$ located strictly between $a_{\underline{k}}$ and $a_{\overline{k}}$, the probability that equals the sum of the conditional dictatorial coefficients of the voters who prefer $a_k$ the most, and
\item[\rm (ii)] $\varphi$ behaves like a PFBR over the two intervals $[a_1,a_{\underline{k}}]$ and $[a_{\overline{k}},a_m]$
where single-peakedness continues to prevail.
\end{itemize}
The proof of the necessity part is more comprehensive. It consists of the following three steps.
\begin{description}
\item[\bf Step 1.] We refer to the $(\underline{k}, \overline{k})$-hybrid domain $\mathbb{D}_{\prec}(\underline{k}, \overline{k})$, and show independently that every unanimous and locally strategy-proof RSCF on $\mathbb{D}_{\prec}(\underline{k}, \overline{k})$ is a $(\underline{k}, \overline{k})$-PFBR.

\item[\bf Step 2.] We turn back to the regular domain $\mathbb{D}$ of Theorem \ref{thm:rulecharacterization}. By Corollary 1 of \citet{KRSYZ2021b}, we first know that every unanimous and strategy-proof rule on $\mathbb{D}$, as a no-restoration domain, satisfies the tops-only property.
    Furthermore, we reveal an important property on every unanimous and strategy-proof RSCF defined on domain $\mathbb{D}$, as a $(\underline{k},\overline{k})$-hybrid domain: for all $i \in N$, $P_i, P_i' \in \mathbb{D}$ and $P_{-i} \in \mathbb{D}^{n-1}$,
    if $r_1(P_i), r_1(P_i') \in M = [a_{\underline{k}},a_{\overline{k}}]$,
    then the social lotteries at profiles $(P_i, P_{-i})$ and $(P_i', P_{-i})$ coincide to each other at every alternative other than the peaks of $P_i$ and $P_i'$.

\item[\bf Step 3.] We fix an arbitrary unanimous and strategy-proof RSCF $\varphi$ on the regular domain $\mathbb{D}$, and construct a RSCF $\psi$ on the hybrid domain $\mathbb{D}_{\prec}(\underline{k}, \overline{k})$:
    given $(P_1, \dots, P_n) \in \big[\mathbb{D}_{\prec}(\underline{k}, \overline{k})\big]^n$,
    $\psi(P_1, \dots, P_n) = \varphi(\bar{P}_1, \dots, \bar{P}_n)$  for all $(\bar{P}_1, \dots, \bar{P}_n) \in \mathbb{D}^n$ such that $r_1(\bar{P}_1) = r_1(P_1)$, ..., $r_1(\bar{P}_n) = r_1(P_n)$.
    Note that $\psi$ is well defined since both $\mathbb{D}$ and $\mathbb{D}_{\prec}(\underline{k}, \overline{k})$ are minimally rich, and $\varphi$ is shown to satisfy the tops-only property in \textbf{Step 2}.
    By construction, $\psi$ inherits unanimity and the tops-only property of $\varphi$.
    More importantly, according to the construction of $\psi$, we can adopt strategy-proofness of $\varphi$ and the property revealed in \textbf{Step 2} to show that $\psi$ is locally strategy-proof on $\mathbb{D}_{\prec}(\underline{k}, \overline{k})$.
    Then, the characterization result established in \textbf{Step 1} implies that $\psi$ is a $(\underline{k}, \overline{k})$-PFBR.
    Consequently, by construction, we know that $\varphi$ inherits the probabilistic ballots of $\psi$ and turns out to be a $(\underline{k}, \overline{k})$-PFBR as well.
\end{description}

%According to Theorem \ref{thm:domaincharacterization}, we know that the regular domain investigated in Theorem \ref{thm:rulecharacterization} is a $(\underline{k},\overline{k})$-hybrid domain.
%In particular, if $\overline{k}-\underline{k} = 1$, $\mathbb{D}$ reduces to a domain of single-peaked preferences,
%and then Theorem \ref{thm:rulecharacterization} follows from Theorem 1 of \citet{RS2019}.
%
%
%For the ``only-if part'', the proof consists of two steps.
%In the first step, we consider the hybrid domain $\mathbb{D}_{\prec}(\underline{k},\overline{k})$ which strictly contains the single-peaked domain $\mathbb{D}_{\prec}$, and then characterize all unanimous and strategy-proof RSCFs as $(\underline{k},\overline{k})$-PFBRs.
%In the second step, we move back to the domain $\mathbb{D}$,
%and complete the proof by showing that the unanimous and strategy-proof RSCF $\varphi: \mathbb{D}^n \rightarrow \Delta(A)$ satisfies the tops-only property, and
%hence is able to be extended to a unanimous RSCF on the hybrid domain $\mathbb{D}_{\prec}(\underline{k}, \overline{k})$, and
%last the extension which nests $\varphi$ is shown to be strategy-proof.\footnote{Since $\mathbb{D}$ is minimally rich, we extend the tops-only RSCF $\varphi$ to a tops-only RSCF $\hat{\varphi}: \big[\mathbb{D}_{\prec}(\underline{k}, \overline{k})\big]^n \rightarrow \Delta(A)$ where
%for each profile $\hat{P} \in \big[\mathbb{D}_{\prec}(\underline{k}, \overline{k})\big]^n$,
%$\hat{\varphi}(\hat{P}) = \varphi(P)$ for all $P \in \mathbb{D}^n$ such that $r_1(P_i) = r_1(\hat{P}_i)$ for all $i \in N$.}

\section{Conclusion}\label{sec:conclusion}

We study a class of preference domains that satisfies minimal richness, diversity and no-restoration.
We first show that such a domain is indeed a hybrid domain, i.e., (i) two thresholds $a_{\underline{k}}, a_{\overline{k}} \in A$, where $1 \leq \underline{k}< \overline{k} \leq m$, are identified on the natural order $\prec$, so that the preferences are revealed to be single-peaked on the intervals $[a_1, a_{\underline{k}}]$ and $[a_{\overline{k}},a_m]$ respectively, and unrestricted in the middle interval $[a_{\underline{k}}, a_{\overline{k}}]$, and (ii) the domain contains sufficiently many preferences so that it satisfies the path-inclusion condition and the no-leaf condition.
We next characterize all unanimous and strategy-proof RSCFs on such a domain to be PFBRs.

In the Supplementary Material, we further investigate decomposability of strategy-proof RSCFs on these domains: Can every unanimous and strategy-proof RSCF be decomposed as a mixture of finitely many deterministic counterparts?\footnote{See the following link: \href{https://drive.google.com/file/d/1L0VlXASqjiNueYemRchs_yKepYK6aDn4/view}
{\scriptsize{\textcolor[rgb]{0.00,0.00,1.00}{https://drive.google.com/file/d/1L0VlXASqjiNueYemRchs\_yKepYK6aDn4/view}}}.}
Indeed, decomposability holds on the complete domain and the single-peaked domain.
Thus, decomposability holds for the cases when $\overline{k}-\underline{k}=1$ or $\overline{k} - \underline{k}=m-1$.
Surprisingly, it does \textit{not} hold for any intermediate values of $\overline{k}$ and $\underline{k}$.
We identify a necessary and sufficient condition for decomposability under an additional assumption of anonymity, which requires the RSCF be non-sensitive to the identities of voters.

\setlength{\bibsep}{0.3ex}

{\small	
	\newpage
	\appendix

	\section*{Appendix}

\section{Proof of Theorem \ref{thm:domaincharacterization}}\label{app:domaincharacterization}

We fix a regular domain $\mathbb{D}$, i.e., minimal richness, diversity and no-restoration.
We are going to identify two thresholds $a_{\underline{k}}$ and $a_{\overline{k}}$,
where $1\leq \underline{k}< \overline{k} \leq m$, and show that $\mathbb{D}$ is a $(\underline{k}, \overline{k})$-hybrid domain.

We first introduce a sketch of the proof, which consists of the following four steps.
\begin{description}
\item[\textbf{Step 1.}]
We first provide three implications of no-restoration, in conjunction with minimal richness and diversity (see Lemmas \ref{lem:vertex-path}, \ref{lem:onestep} and \ref{lem:preferencerestriction}).
First two implications are related to the graph $G_{\approx}^A(\mathbb{D})$, while the third one reveals an important preference restriction embedded in the domain $\mathbb{D}$ based on the graph $G_{\approx}^A(\mathbb{D})$.

\item[\textbf{Step 2.}]
Using the first two implications in \textbf{Step 1}, we show that $\mathbb{D}$ satisfies the path-inclusion condition (see Lemma \ref{lem:degree}).

\item[\textbf{Step 3.}]
According to $G_{\approx}^A(\mathbb{D})$,
we know that either there exists a unique vertex-path connecting $a_1$ and $a_m$, or there are multiple vertex-paths connecting $a_1$ and $a_m$.
We consider the first case, and show via the preference restriction revealed in \textbf{Step 1}, that
domain $\mathbb{D}$ is single-peaked w.r.t.~$\prec$, and hence is a hybrid domain (see Lemma \ref{lem:sp}).

\item[\textbf{Step 4.}] We last consider the case that there are multiple vertex-paths connecting $a_1$ and $a_m$. Let each of these vertex-path start at $a_1$ and end at $a_m$.
    We then identify their left maximum common part and right maximum common part, and show by the preference restriction revealed in \textbf{Step 1} and the completely reversed preferences $\underline{P}_i$ and $\overline{P}_i$, that the left maximum common part must be $(a_1, \dots, a_{\underline{k}})$ and the right maximum common part must be $(a_{\overline{k}}, \dots, a_m)$, where $\overline{k}-\underline{k}>1$ (see statement (i) of Lemma \ref{lem:coincidence}).
    We further show that the graph $G_{\approx}^A(\mathbb{D})$ is simply a combination of the vertex-path
    $(a_1, \dots, a_{\underline{k}})$, a connected subgraph $G_{\approx}^{[a_{\underline{k}}, a_{\overline{k}}]}(\mathbb{D})$ and the vertex-path
    $(a_{\overline{k}}, \dots, a_m)$ (see statement (ii) of Lemma \ref{lem:coincidence}), which immediately by the preference restriction revealed in \textbf{Step 1}, implies $\mathbb{D} \subseteq \mathbb{D}_{\prec}(\underline{k}, \overline{k})$. Finally, we prove that $\mathbb{D}$ is a $(\underline{k}, \overline{k})$-hybrid domain
    by showing the no-leaf condition (see Lemma \ref{lem:hybrid*}).
\end{description}

Now, we proceed to the proof.

\begin{lemma}\label{lem:vertex-path}
The graph $G_{\approx}^A(\mathbb{D})$ is a connected graph.
\end{lemma}
\begin{proof}
Given arbitrary distinct $a_s, a_t \in A$,
by minimal richness, we have $P_i \in \mathbb{D}$ with $r_1(P_i) = a_s$ and $P_i'\in \mathbb{D}$ with $r_1(P_i') = a_t$.
Moreover, since $\mathbb{D}$ is a no-restoration domain,
there exists a path $(P_{i}^1, \dots, P_i^t)$ connecting $P_{i}$ and $P_{i}'$.
We partition $(P_{i}^1, \dots, P_i^t)$ according to the peaks of preferences (without rearranging preferences in the path), and
elicit all preference peaks:
\begin{displaymath}
\left(\frac{P_i^1, \dots, P_i^{k_1}}{\textrm{the same peak}\; x_1}, \frac{P_i^{k_1+1}, \dots, P_i^{k_2}}{\textrm{the same peak}\; x_2}, \dots, \frac{P_i^{k_{q-1}+1}, \dots, P_i^{t}}{\textrm{the same peak}\; x_q}\right)
\xlongrightarrow{\textrm{Elicit peaks}} (x_1, x_2, \dots, x_q),
\end{displaymath}
where $x_k \neq x_{k+1}$ for all $k = 1, \dots, q-1$.
For each $v = 1, \dots, q-1$, according to $P_i^{k_v}$ and $P_i^{k_v+1}$, we know that $x_v \approx x_{v+1}$.
Note that $(x_1, x_2, \dots, x_q)$ may contain repetitions.
Whenever a repetition appears,
we remove all alternatives strictly between the repetition and one alternative of the repetition.
For instance, if $x_k = x_l$ where $1\leq k< l\leq q$, we remove $x_k, x_{k+1}, \dots, x_{l-1}$, and refine the sequence to
$(x_1, \dots, x_{k-1}, x_l, \dots, x_q)$.
By repeatedly eliminating repetitions,
we finally elicit a vertex-path connecting $a_s$ and $a_t$.
\end{proof}

Given $a_s, a_t \in A$,
let $\Pi(a_s, a_t)$ denote the set of \emph{all} vertex-paths connecting $a_s$ and $a_t$ in $G_{\approx}^A(\mathbb{D})$.\footnote{In particular, if $a_s = a_t$, then $\Pi(a_s, a_t) = \big\{(a_s)\big\}$ is a singleton set of a null vertex-path.}
Clearly, Lemma \ref{lem:vertex-path} implies $\Pi(a_s, a_t) \neq \emptyset$.
Given a vertex-path $\mathcal{P} \in \Pi(a_s, a_t)$ and $a_p, a_q \in \mathcal{P}$,
let $\langle a_p, a_q\rangle^{\mathcal{P}}$ denote the interval between $a_p$ and $a_q$ on $\mathcal{P}$.

\begin{lemma}\label{lem:onestep}
Given $a_s\in A\backslash \{a_1, a_m\}$ and $a_t \in A$, let $a_s \approx a_t$.
If one vertex-path of $\Pi(a_1, a_m)$ contains $a_t$,
there exists a vertex-path of $\Pi(a_1, a_m)$ that includes $a_s$.
\end{lemma}

\begin{proof}
Let $(x_1, \dots, x_p) \in \Pi(a_1, a_m)$ and $a_t = x_{\eta}$ for some $1\leq \eta \leq p$.
If $a_s \in \{x_1, \dots, x_p\}$, the lemma holds evidently.
Henceforth, assume $a_s \notin \{x_1, \dots, x_p\}$.
Thus, we have two vertex-paths $(a_1= x_1, x_2, \dots, x_{\eta}= a_t, a_s) \in \Pi(a_1, a_s)$,
and $(a_s, a_t= x_{\eta}, \dots, x_{p-1}, x_p = a_m) \in \Pi(a_s, a_m)$.
	
Since $\underline{P}_i$ and $\overline{P}_i$ are completely reversed, either $a_s\underline{P}_ia_t$ or $a_s\overline{P}_ia_t$ holds.
Assume w.l.o.g.~that $a_s\underline{P}_ia_t$.
Thus, $x_{\eta} = a_t \neq a_1 = x_1$.
Pick $P_{i} \in \mathbb{D}$ with $r_1(P_i) = a_s$ by minimal richness.
Clearly, $\underline{P}_i$ and $P_i$ are distinct.
By no-restoration, we have a path $(P_{i}^1, \dots, P_i^{\nu})$ connecting $\underline{P}_i$ and $P_{i}$
such that $a_s P_{i}^{k}a_t$ for all $k = 1, \dots, \nu$.
Thus, $r_1(P_i^k) \neq a_t$ for all $k = 1, \dots, \nu$.
According to the path $(P_{i}^1, \dots, P_i^{\nu})$, we elicit a vertex-path $(y_1, \dots, y_q) \in  \Pi(a_1,a_s)$.
Clearly $a_t \notin \{y_1, \dots, y_q\}$.

It is evident that $\{a_1\}\subseteq \{y_1, \dots, y_q\}\cap \{x_1, \dots, x_p\}$.
If $\{y_1, \dots, y_q\}\cap \{x_1, \dots, x_p\}= \{a_1\}$,
the concatenated vertex-path $(a_1= y_1, \dots, y_q = a_s, a_t= x_{\eta}, \dots, x_p = a_m)$ includes $a_s$.
Next, we assume $\{a_1\}\subset \{y_1, \dots, y_q\}\cap \{x_1, \dots, x_p\}$.
We identify the alternative in the vertex-path $(y_1, \dots, y_q)$ that has the maximum index and is contained in the vertex-path $(x_1, \dots, x_p)$, i.e.,
$y_{\hat{k}} = x_{k^{\ast}}$ for some $1< \hat{k}< q$ and $1<k^{\ast}\leq p$ and
$\{y_{\hat{k}+1}, \dots, y_q\} \cap \{x_1, \dots, x_p\} = \emptyset$.
Recall that $a_t = x_{\eta}$, $1<\eta \leq p$ and $a_t \notin \{y_1, \dots, y_q\}$.
Therefore, either $1< k^{\ast}< \eta$ or $\eta< k^{\ast}\leq p$ holds.
If $1< k^{\ast}< \eta$, the concatenated vertex-path
$(a_1= x_1,\dots, x_{k^{\ast}} = y_{\hat{k}}, y_{\hat{k}+1}, \dots, y_q = a_s, a_t= x_{\eta}, \dots, x_p = a_m)$
includes $a_s$.
If $\eta< k^{\ast}\leq p$,
the concatenated vertex-path $(a_1= x_1, \dots, x_{\eta}= a_t, a_s =y_{q}, \dots, y_{\hat{k}+1}, y_{\hat{k}} = x_{k^{\ast}}, \dots, x_p = a_m)$
	includes $a_s$.
\end{proof}

\begin{lemma}\label{lem:preferencerestriction}
	Given distinct $a_p, a_t, a_s \in A$,
	let $a_t$ be included in every vertex-path of $\Pi(a_p, a_s)$.
	Given a preference $P_i \in \mathbb{D}$, we have $[r_1(P_i) = a_p] \Rightarrow [a_tP_ia_s]$ and
$[r_1(P_i) = a_s] \Rightarrow [a_tP_ia_p]$.
\end{lemma}
\begin{proof}
	Let $r_1(P_i) = a_p$. Suppose $a_s P_i a_t$.
	Pick arbitrary $P_{i}' \in \mathbb{D}$ with $r_1(P_i') = a_s$ by minimal richness.
	Since $\mathbb{D}$ is a no-restoration domain, we have a path $(P_{i}^1, \dots, P_i^l)$ connecting $P_{i}$ and $P_{i}'$ such that $a_{s}P_{i}^{k}a_t$ for all $k = 1, \dots, l$.
	Thus, $r_1(P_i^k) \neq a_t$ for all $k = 1, \dots, l$.
Consequently, according to path $(P_{i}^1, \dots, P_i^l)$, we elicit a vertex-path that connects $a_p$ and $a_s$, and excludes $a_t$.
	This contradicts the hypothesis of the lemma.
	Therefore, $a_tP_ia_s$.
	Symmetrically, if $r_1(P_i) = a_s$, then $a_tP_ia_p$.
\end{proof}

This completes the verification of the first step.

\begin{lemma}\label{lem:degree}
Given $a_s \in A\backslash \{a_1, a_m\}$,
there exists a vertex-path of $\Pi(a_1, a_m)$ that includes $a_s$.
\end{lemma}

\begin{proof}
Pick arbitrary $P_{i} \in \mathbb{D}$ with $r_1(P_i) = a_s$ by minimal richness.
Since $a_s\underline{P}_ia_m$ and $a_s P_ia_m$,
we by no-restoration have a path $(P_{i}^1, \dots, P_i^l)$ connecting $\underline{P}_{i}$ and $P_{i}$ such that
$a_{s}P_{i}^{k}a_m$ for all $k = 1, \dots, l$.
Thus, $r_1(P_{i}^{k}) \neq a_m$ for all $k = 1, \dots, l$.
According to the path $(P_{i}^1, \dots, P_i^l)$, we elicit a vertex-path $(x_1, \dots, x_p) \in \Pi(a_1, a_s)$ that excludes $a_m$.
Symmetrically, we have a vertex-path $(y_1, \dots, y_q) \in \Pi(a_s, a_m)$ that excludes $a_1$.
Thus, $\{a_s\} \subseteq \{x_1, \dots, x_p\}\cap \{y_1, \dots, y_q\}$.
If $\{x_1, \dots, x_p\}\cap \{y_1, \dots, y_q\} = \{a_s\}$, the concatenated vertex-path
$(a_1= x_{1},\dots, x_{p} = a_s = y_1, \dots, y_q =a_m)$ includes $a_s$.
If $\{a_s\} \subset \{x_1, \dots, x_p\}\cap \{y_1, \dots, y_q\}$,
we identify the alternative $a_t$ included in both vertex-paths $(x_1, \dots, x_p)$ and $(y_1, \dots, y_q)$ with the maximum index in $(x_1, \dots, x_p)$ and the minimum index in $(y_1, \dots, y_q)$, i.e., $a_t = x_{\hat{k}}=y_{k^{\ast}}$ for some $1 < \hat{k} < p$ and $1< k^{\ast} < q$ such that
$\{x_{1}, \dots, x_{\hat{k}-1}\}\cap\{y_{k^{\ast} +1}, \dots, y_q\} = \emptyset$.
Thus, the concatenated vertex-path $(a_1=x_1, \dots, x_{\hat{k}-1}, x_{\hat{k}} = a_t =  y_{k^{\ast}}, y_{k^{\ast}+1}, \dots, y_q = a_m)$
includes $a_t$ and excludes $a_s$.
Last, along the vertex-path $(a_t= x_{\hat{k}},  \dots,  x_{p}= a_s)$, by repeatedly applying Lemma \ref{lem:onestep} step by step from $a_t$ to $a_s$,
we eventually find a vertex-path of $\Pi(a_1, a_m)$ that includes $a_s$.
\end{proof}

This completes the verification of the second step, and shows the path-inclusion condition.

\begin{lemma}\label{lem:sp}
	If $\Pi(a_1, a_m)$ is a singleton set, then $\mathbb{D}$ is a single-peaked domain w.r.t.~$\prec$, and hence a hybrid domain.
\end{lemma}

\begin{proof}
Let $\mathcal{P}$ be the unique vertex-path of $\Pi(a_1, a_m)$.
Then, Lemma \ref{lem:degree} implies that all alternatives are included in $\mathcal{P}$.
Thus, the path-inclusion condition is satisfied.
Furthermore, by uniqueness of $\mathcal{P}$, we can infer that $G_{\approx}^A(\mathbb{D})$ is the line of $\mathcal{P}$.
Then Lemma \ref{lem:preferencerestriction} implies that $\mathbb{D}$ is a single-peaked domain on $\mathcal{P}$, i.e.,
for all $P_i \in \mathbb{D}$ and distinct $a_s, a_t \in A$, $\big[a_s \in \langle r_1(P_i), a_t\rangle^{\mathcal{P}}\big] \Rightarrow [a_sP_ia_t]$.
Furthermore, by the labeling of alternatives in the preferences $\underline{P}_i$ and $\overline{P}_i$,
it must be the case that all alternatives of $\mathcal{P}$ are located in the order of $\prec$, i.e.,
$\mathcal{P} = (a_1, \dots, a_k,a_{k+1}, \dots, a_m)$.
Therefore, $\mathbb{D}$ is a single-peaked domain w.r.t.~$\prec$, and hence
$\mathbb{D} \subseteq \mathbb{D}_{\prec}(k, k+1)$ for all $k= 1, \dots, m-1$,
which indicates that the no-leaf condition is vacuously satisfied.
In conclusion, $\mathbb{D}$ is a hybrid domain.
\end{proof}

This completes the verification of the third step.\medskip

In the rest of proof, we assume that $\Pi(a_1, a_m)$ is not a singleton set.
Note that $\Pi(a_1, a_m)$ is a finite nonempty set.
Hence, we label $\Pi(a_1, a_m) = \{\mathcal{P}_1, \dots, \mathcal{P}_{\omega}\}$, $\omega\geq 2$,
and make sure that each vertex-path of $\Pi(a_1, a_m)$ starts from $a_1$ (at left) and ends at $a_m$ (at right).
Thus, we can identify the left maximum common part and the right maximum common part of all vertex-paths of $\Pi(a_1, a_m)$, i.e.,
there exist two alternatives $a_{\underline{k}}, a_{\overline{k}} \in A$ (either $\underline{k}\leq \overline{k}$ or $\underline{k} \geq \overline{k}$, so far)
such that the following three conditions are satisfied:
\begin{itemize}
\item[\bf (1)] $a_{\underline{k}}, a_{\overline{k}} \in \mathcal{P}_l$ for all $\mathcal{P}_l \in \Pi(a_1, a_m)$,

\item[\bf (2)] $\langle a_1, a_{\underline{k}}\rangle^{\mathcal{P}_l} = \langle a_1, a_{\underline{k}}\rangle^{\mathcal{P}_{\nu}}$
and $\langle a_{\overline{k}}, a_m\rangle^{\mathcal{P}_l} = \langle a_{\overline{k}}, a_m\rangle^{\mathcal{P}_{\nu}}$ for all $\mathcal{P}_l, \mathcal{P}_{\nu} \in \Pi(a_1, a_m)$, and

\item[\bf (3)] there exist two vertex-paths $\mathcal{P}, \mathcal{P}' \in \Pi(a_1, a_m)$ and
two alternatives $a_s \in \langle a_{\underline{k}}, a_m\rangle^{\mathcal{P}}$ and $a_t \in \langle a_{\underline{k}}, a_m\rangle^{\mathcal{P}'}$ such that $a_s \neq a_t$, $a_s \approx a_{\underline{k}}$ and $a_t \approx a_{\underline{k}}$; and
there exist two vertex-paths $\overline{\mathcal{P}}, \overline{\mathcal{P}}' \in \Pi(a_1, a_m)$ and
two alternatives $a_p \in \langle a_1, a_{\overline{k}}\rangle^{\overline{\mathcal{P}}}$ and
$a_q \in \langle a_1, a_{\overline{k}}\rangle^{\overline{\mathcal{P}}'}$ such that $a_p \neq a_q$, $a_p \approx a_{\overline{k}}$ and $a_q \approx a_{\overline{k}}$.
\end{itemize}

According to condition {\bf (2)}, to avoid that $\Pi(a_1, a_m)$ degenerates to a singleton set,
it must be the case that $a_{\underline{k}}\neq a_{\overline{k}}$.

Fix arbitrary $\mathcal{P}_l \in \Pi(a_1, a_m)$.
We first claim $\langle a_1, a_{\underline{k}}\rangle^{\mathcal{P}_l} \cap  \langle a_{\overline{k}}, a_m\rangle^{\mathcal{P}_l} = \emptyset$.
Suppose not, i.e., there exists $a_s \in \langle a_1, a_{\underline{k}}\rangle^{\mathcal{P}_l} \cap  \langle a_{\overline{k}}, a_m\rangle^{\mathcal{P}_l}$ such that
$\langle a_1, a_{s}\rangle^{\mathcal{P}_l} \cap  \langle a_{s}, a_m\rangle^{\mathcal{P}_l} = \{a_s\}$.
Since $a_{\underline{k}}\neq a_{\overline{k}}$, we know either $a_s \neq a_{\underline{k}}$ or $a_s \neq a_{\overline{k}}$.
Consequently, the concatenation of $\langle a_1, a_{s}\rangle^{\mathcal{P}_l}$ and $\langle a_{s}, a_m\rangle^{\mathcal{P}_l}$ forms a vertex-path
that connects $a_1$ and $a_m$, but excludes either $a_{\underline{k}}$ or $a_{\overline{k}}$. This contradicts condition {\bf (1)}.
Therefore, $\langle a_1, a_{\underline{k}}\rangle^{\mathcal{P}_l} \cap  \langle a_{\overline{k}}, a_m\rangle^{\mathcal{P}_l} = \emptyset$.
This implies $1\leq \underline{k}<m$ and $1< \overline{k}\leq m$.
Furthermore, we claim $\langle a_1, a_{\underline{k}}\rangle^{\mathcal{P}_l} \cup \langle a_{\overline{k}}, a_m\rangle^{\mathcal{P}_l} \neq A$.
Otherwise, condition {\bf (2)} implies $\langle a_1, a_{\underline{k}}\rangle^{\mathcal{P}_{\nu}} \cup  \langle a_{\overline{k}}, a_m\rangle^{\mathcal{P}_{\nu}} = A$ for all $\mathcal{P}_{\nu} \in \Pi(a_1, a_m)$, and $\Pi(a_1, a_m)$ degenerates to a singleton set. Contradiction.

\begin{lemma}\label{lem:coincidence}
The following two statements hold:
\begin{itemize}
\item[\rm (i)] Set $\Pi(a_1, a_{\underline{k}})$ is a singleton set of the vertex-path $(a_1, \dots, a_k, a_{k+1},\dots, a_{\underline{k}})$,
set $\Pi(a_{\overline{k}}, a_m)$ is a singleton set of the vertex-path $(a_{\overline{k}}, \dots,a_{k}, a_{k+1},\dots, a_m)$, and $\overline{k}-\underline{k} > 1$.

\item[\rm (ii)]
Given $M = \{a_{\underline{k}}, \dots, a_k, a_{k+1}, \dots, a_{\overline{k}}\}$,
graph $G_{\approx}^A(\mathbb{D})$ is a combination of the vertex-path
$(a_1, \dots, a_k, a_{k+1},\dots, a_{\underline{k}})$, the subgraph $G_{\approx}^M(\mathbb{D})$ and the vertex-path $(a_{\overline{k}}, \dots,a_{k}, a_{k+1},\dots, a_m)$, i.e.,
each pair $(a_s, a_t)$ such that $a_s \approx a_t$ is an edge in $(a_1, \dots, a_k, a_{k+1},\dots, a_{\underline{k}})$, $G_{\approx}^M(\mathbb{D})$ or $(a_{\overline{k}}, \dots,a_{k}, a_{k+1},\dots, a_m)$.
\end{itemize}
\end{lemma}
\begin{proof}
By symmetry, we focus on showing that $\Pi(a_1, a_{\underline{k}})$ is a singleton set of the vertex-path $(a_1, \dots, a_k, a_{k+1},\dots, a_{\underline{k}})$.

We first show $\Pi(a_1, a_{\underline{k}})$ is a singleton set.
If $a_1 = a_{\underline{k}}$, it is evident that $\Pi(a_1, a_{\underline{k}})$ is a singleton set.
We further assume $a_1 \neq a_{\underline{k}}$.
Suppose that $\Pi(a_1, a_{\underline{k}})$ is not a singleton set.
Pick arbitrary $\mathcal{P}_l \equiv (x_1, \dots, x_q) \in \Pi(a_1, a_m)$.
By condition {\bf (1)}, we know $a_{\underline{k}} = x_p$ for some $1<p<q$.
Thus, $\langle a_1, a_{\underline{k}}\rangle^{\mathcal{P}_l} = (x_1, \dots, x_p) \in \Pi(a_1, a_{\underline{k}})$.
Since $\Pi(a_1, a_{\underline{k}})$ is not a singleton set,
we have another vertex-path $(y_1, \dots, y_u) \in \Pi(a_1, a_{\underline{k}})$.
Clearly, we can identify $1< \hat{k} < \min(p, u)-1$ such that $x_k = y_k$ for all $k = 1, \dots, \hat{k}$ and $x_{\hat{k}+1} \neq y_{\hat{k}+1}$.
According to vertex-path $(a_1 =y_1, \dots, y_u = a_{\underline{k}})$ and $(a_{\underline{k}}=x_p, \dots, x_q=a_m)$,
we know either $\{a_{\underline{k}}\}=\{y_1, \dots, y_u\}\cap \{x_p, \dots, x_q\}$ or
$\{a_{\underline{k}}\} \subset \{y_1, \dots, y_u\}\cap \{x_p, \dots, x_q\}$.
If $\{a_{\underline{k}}\}=\{y_1, \dots, y_u\}\cap \{x_p, \dots, x_q\}$,
we have the concatenated vertex-path $\mathcal{P}_{\nu} \equiv (a_1 =y_1, \dots, y_u = a_{\underline{k}}=x_p, \dots, x_q=a_m) \in \Pi(a_1, a_m)$.
However, $\langle a_1, a_{\underline{k}}\rangle^{\mathcal{P}_l} = (x_1, \dots, x_p) \neq (y_1, \dots, y_u)
= \langle a_1, a_{\underline{k}}\rangle^{\mathcal{P}_{\nu}}$ contradicts condition {\bf (2)}.
If $\{a_{\underline{k}}\} \subset \{y_1, \dots, y_u\}\cap \{x_p, \dots, x_q\}$,
we can identify $1\leq s<u$ and $p < t \leq q$ such that $y_s = x_t$ and $\{y_1, \dots, y_{s-1}\}\cap \{x_{t+1}, \dots, x_q\} = \emptyset$.
Then, we have the concatenated vertex-path $(y_1, \dots, y_{s-1}, y_s = x_t, x_{t+1}, \dots, x_q) \in \Pi(a_1, a_m)$ which excludes $a_{\underline{k}}$.
This contradicts condition {\bf (1)}.
In conclusion, $\Pi(a_1, a_{\underline{k}})$ is a singleton set.

Now, let $\Pi(a_1, a_{\underline{k}}) = \{\mathcal{P}\}$.
We next show that $\mathcal{P}$ contains and only contains vertices of $a_1, \dots, a_k, a_{k+1},\dots, a_{\underline{k}}$.
It is equivalent to show $[a_s \notin \mathcal{P}] \Leftrightarrow [s >\underline{k}]$.
Suppose by contradiction that there exists $\underline{k}< s \leq m$ such that $a_s \in \mathcal{P}$.
Clearly, $\underline{k}>1$, and $a_s$ is located between $a_1$ and $a_{\underline{k}}$ in every vertex-path of $\Pi(a_1, a_{\underline{k}})$.
Consequently, Lemma \ref{lem:preferencerestriction} implies $a_sP_ia_{\underline{k}}$ for all $P_i \in \mathbb{D}$ with $r_1(P_i) = a_1$.
However, by the labeling of alternatives in $\underline{P}_i$ who has the peak $a_1$, we have $a_{\underline{k}}\underline{P}_ia_s$. Contradiction.
Therefore, we have $[a_s \notin \mathcal{P}] \Leftarrow [s >\underline{k}]$.
Next, we show $[a_s \notin \mathcal{P}] \Rightarrow [s >\underline{k}]$.
Given $a_s \notin \mathcal{P}$, we claim that $a_{\underline{k}}$ is included in every vertex-path of $\Pi(a_1, a_s)$.
Suppose not, i.e., we have a path $(x_1, \dots, x_p) \in \Pi(a_1, a_s)$ that excludes $a_{\underline{k}}$.
By Lemma \ref{lem:degree}, we also have a path $\mathcal{P}_v \equiv (y_1, \dots, y_q) \in \Pi(a_1, a_m)$ that includes $a_s$.
Let $a_s = y_{\ell}$ for some $1\leq \ell \leq q$.
By condition {\bf (1)} and {\bf (2)}, we know $a_{\underline{k}} = y_{\eta}$ for some $1 \leq \eta < q$.
Moreover, since $(y_1, \dots, y_{\eta}) \in \Pi(a_1, a_{\underline{k}}) = \{\mathcal{P}\}$, it is true that $\mathcal{P} = (y_1, \dots, y_{\eta})$.
Therefore, $a_s \notin \mathcal{P}$ implies $\eta < \ell \leq q$, and hence we have a vertex-path $(y_{\ell}, \dots, y_{q}) \in \Pi(a_s, a_m)$ that excludes $a_{\underline{k}}$.
Thus, according to the vertex-paths $(a_1 = x_1, \dots, x_p = a_s)$ and
$(a_s = y_{\ell}, \dots, y_{q} = a_m)$,
we can construct a vertex-path that connects $a_1$ and $a_m$, and excludes $a_{\underline{k}}$.
This contradicts condition {\bf (1)}.
Therefore, $a_{\underline{k}}$ is included in every vertex-path of $\Pi(a_1, a_s)$.
Then, Lemma \ref{lem:preferencerestriction} implies $a_{\underline{k}}\underline{P}_ia_s$.
Moreover, according to the labeling of alternatives in $\underline{P}_i$,
$a_{\underline{k}}\underline{P}_ia_s$ implies $s >\underline{k}$.
Therefore, we have $[a_s \notin \mathcal{P}] \Rightarrow [s >\underline{k}]$.
In conclusion, $[a_s \notin \mathcal{P}] \Leftrightarrow [s >\underline{k}]$.

Now, since $\mathcal{P}$ contains and only contains vertices of $a_1, \dots, a_k, a_{k+1},\dots, a_{\underline{k}}$,
according to the labeling of alternatives in $\underline{P}_i$,
we know that all alternatives of $\mathcal{P}$ must be located in the order of the natural order $\prec$.

Recall that $\langle a_1, a_{\underline{k}}\rangle^{\mathcal{P}_l}\cap \langle a_{\overline{k}}, a_m\rangle^{\mathcal{P}_l} = \emptyset$ and $\langle a_1, a_{\underline{k}}\rangle^{\mathcal{P}_l}\cup \langle a_{\overline{k}}, a_m\rangle^{\mathcal{P}_l} \neq A$ for all $\mathcal{P}_l \in \Pi(a_1, a_m)$.
Therefore, according to statement (i), it must be the case that $\overline{k}-\underline{k}>1$.
This completes the verification of statement (i).\medskip

To prove statement (ii),
it suffices to show that
(a) no alternative of $\{a_1, \dots, a_{\underline{k}-1}\}$ is strongly connected to any alternative of $\{a_{\underline{k}+1}, \dots, a_m\}$, and
(b) no alternative of $\{a_{\overline{k}+1}, \dots, a_m\}$ is strongly connected to any alternative of $\{a_1, \dots, a_{\overline{k}-1}\}$.
Indeed, items (a) and (b) are symmetric.
Hence, we focus on the verification of item (a).
Suppose by contradiction that there exists a pair $(a_s, a_t)$ such that $1\leq s< \underline{k}$, $\underline{k}< t \leq m$ and $a_s \approx a_t$.
Thus, we have the vertex-path $(a_1, \dots, a_s, a_t) \in \Pi(a_1, a_t)$ that excludes $a_{\underline{k}}$.
Meanwhile, by Lemma \ref{lem:degree}, we have a vertex-path $\mathcal{P}_v \equiv (y_1, \dots, y_q) \in \Pi(a_1, a_m)$ that includes $a_t$.
By conditions {\bf (1)} and {\bf (2)}, we know $a_{\underline{k}} = y_{\eta}$ for some $1 \leq \eta < q$.
Thus, $(a_1, \dots, a_k,a_{k+1}, \dots, a_{\underline{k}}) = (y_1, \dots, y_{\eta})$, and moreover
$\underline{k} < t\leq m$ implies $a_t = y_{\ell}$ for some $\eta< \ell \leq q$.
Hence, we have a vertex-path $(y_{\ell}, \dots, y_{q}) \in \Pi(a_t, a_m)$ that excludes $a_{\underline{k}}$.
Now, according to the vertex-paths $(a_1, \dots, a_s, a_t)$ and $(a_t = y_{\ell}, \dots, y_{q} = a_m)$,
we can construct a vertex-path that connects $a_1$ and $a_m$, and excludes $a_{\underline{k}}$.
This contradicts condition {\bf (1)}.
In conclusion, $G_{\approx}^A(\mathbb{D})$ is a combination of
$\langle a_1, a_{\underline{k}}\rangle$, $G_{\approx}^M(\mathbb{D})$ and $\langle a_{\overline{k}}, a_m\rangle$.
This completes the verification of statement (ii) and hence proves the lemma.
\end{proof}

Henceforth, let $L = \{a_1, \dots, a_k, a_{k+1},\dots, a_{\underline{k}}\}$ and
$R = \{a_{\overline{k}}, \dots,a_{k}, a_{k+1},\dots, a_m\}$.
According to the natural order $\prec$, we know
$L = [a_1, a_{\underline{k}}]$, $M = [a_{\underline{k}}, a_{\overline{k}}]$ and $R = [a_{\overline{k}}, a_m]$.

Now, we are ready to show that $\mathbb{D}$ is a $(\underline{k}, \overline{k})$-hybrid domain.

\begin{lemma}\label{lem:hybrid*}
Domain $\mathbb{D}$ is a $(\underline{k}, \overline{k})$-hybrid domain.
\end{lemma}

\begin{proof}
By Lemma \ref{lem:degree}, we know $\mathbb{D}$ satisfies the path-inclusion condition.
Therefore, to prove this lemma, we show that $\mathbb{D} \subseteq \mathbb{D}_{\prec}(\underline{k}, \overline{k})$, and $\mathbb{D}$ satisfies the no-leaf condition.

Since $G_{\approx}^A(\mathbb{D})$ is a connected graph by Lemma \ref{lem:vertex-path},
statement (ii) of Lemma \ref{lem:coincidence} implies that $G_{\approx}^M(\mathbb{D})$ is a connected subgraph.
Thus, $G_{\approx}^A(\mathbb{D})$ is simply a combination of the vertex-path $(a_1, \dots, a_k, a_{k+1},\dots, a_{\underline{k}})$, the connected subgraph $G_{\approx}^M(\mathbb{D})$ and the vertex-path $(a_{\overline{k}}, \dots,a_{k}, a_{k+1},\dots, a_m)$.
Thus, by applying Lemma \ref{lem:preferencerestriction}, one can easily show that every preference of $\mathbb{D}$ is a $(\underline{k}, \overline{k})$-hybrid preference, as required, and moreover,
by Lemma \ref{lem:degree} and statement (ii) of Lemma \ref{lem:coincidence},
we can easily infer that for each $\underline{k}< k < \overline{k}$,
there exists a vertex-path in $G_{\approx}^M(\mathbb{D})$ that connects $a_{\underline{k}}$ and $a_{\overline{k}}$, and includes $a_k$.
Hence, for all $\underline{k}< k < \overline{k}$, $a_k$ is not a leaf in $G_{\approx}^M(\mathbb{D})$.
Moreover, by condition {\bf (3)} above and statement (i) of Lemma \ref{lem:coincidence}, we clearly know that neither $a_{\underline{k}}$ nor $a_{\overline{k}}$ is a leaf in $G_{\approx}^M(\mathbb{D})$.
Therefore, $G_{\approx}^M(\mathbb{D})$ has no leaf, as required.
This proves the lemma, completes the verification of the fourth step, and hence proves Theorem \ref{thm:domaincharacterization}.
\end{proof}

\section{Proof of Theorem \ref{thm:rulecharacterization}}\label{app:rulecharacterization}
Let $\mathbb{D}$ be a regular domain, i.e., minimal richness, diversity and no-restoration.
By Theorem \ref{thm:domaincharacterization}, we know that $\mathbb{D}$ is a $(\underline{k},\overline{k})$-hybrid domain for some $1 \leq \underline{k}< \overline{k} \leq m$.
This proves statement (i) of Theorem \ref{thm:rulecharacterization}.

Now, we know that either $\overline{k}-\underline{k} = 1$ or $\overline{k}-\underline{k} > 1$.
If $\overline{k}-\underline{k} = 1$, then $\mathbb{D}$ is a single-peaked domain.
Furthermore, by Theorem 1 of \citet{RS2019}, we know that a RSCF $\varphi: \mathbb{D}^n \rightarrow \Delta(A)$ is unanimous and strategy-proof if and only if it is a PFBR.\footnote{Lemma \ref{lem:vertex-path} first shows that $G_{\approx}^A$ is a connected graph.
Next, since $\mathbb{D}$ is a single-peaked domain w.r.t.~$\prec$, it must be the case that $G_{\approx}^A$ is a line such that all alternatives are located in the order of $\prec$. Consequently, $\mathbb{D}$ satisfies the hypothesis of Theorem 1 of \citet{RS2019}.}
This proves statement (ii) of Theorem \ref{thm:rulecharacterization} in the case that $\overline{k}-\underline{k} = 1$.

In the rest of proof, we assume $\overline{k}-\underline{k} > 1$, and show statement (ii) of Theorem \ref{thm:rulecharacterization}.
We start from the sufficiency part.

\medskip
\noindent(\textbf{Sufficiency})~
We fix a $(\underline{k}, \overline{k})$-PFBR
$\varphi: \mathbb{D}^{n} \rightarrow \Delta(A)$
w.r.t.~the probabilistic ballots $(\beta_S)_{S\subseteq N}$.
It is clear that ballot unanimity implies that $\varphi$ is unanimous.
We henceforth focus on verifying strategy-proofness of $\varphi$.

Since $\mathbb{D}$ is a $(\underline{k}, \overline{k})$-hybrid domain by statement (i), we know that all preferences of $\mathbb{D}$ are $(\underline{k}, \overline{k})$-hybrid preferences (Definition \ref{defn:hybriddomain}).
More importantly,
since $\mathbb{D}$ is a no-restoration domain,
by Theorem 1 of \citet{KRSYZ2021b}, we know that every unanimous and locally strategy-proof RSCF on $\mathbb{D}$ is strategy-proof.\footnote{Recall the formal definition of local strategy-proofness in Footnote \ref{footnote:localSP}}
Therefore, to complete the verification, it suffices to show that the $(\underline{k}, \overline{k})$-PFBR
$\varphi$ is locally strategy-proof on $\mathbb{D}$.
Also, note that the $(\underline{k}, \overline{k})$-PFBR $\varphi$ by definition satisfies the tops-only property.

%First, we introduce the notion of \textit{local strategy-proofness},
%which is weaker than strategy-proofness as it only requires a RSCF be immune to the misrepresentation of preferences that are adjacent to the sincere one.
%Formally, a RSCF $\phi: \mathbb{D}^{n} \rightarrow \Delta(A)$ is \textbf{locally strategy-proof} if for all $i \in N$, $P_i, P_i' \in \mathbb{D}$ with $P_i \sim P_i'$ and $P_{-i} \in \mathbb{D}^{n-1}$, $\phi(P_i, P_{-i})$ stochastically dominates $\phi(P_i', P_{-i})$ according to $P_i$.
%In Fact 1 of the Supplementary Material, the hybrid domain $\mathbb{D}_{\prec}(\underline{k}, \overline{k})$ is shown to be a no-restoration domain.
%Immediately, by Theorem 3 of \citet{KRSYZ2021b}, we know that every unanimous and locally strategy-proof RSCF on $\mathbb{D}_{\prec}(\underline{k}, \overline{k})$ is strategy-proof.
%Therefore, the verification of strategy-proofness of $\varphi$ is simplified to a verification of local strategy-proofness.	

Fixing $i \in N$, $P_i, P_i' \in \mathbb{D}$ with $P_i \sim P_i'$ and $P_{-i} \in \mathbb{D}^{n-1}$,
we show that $\varphi(P_i, P_{-i})$ stochastically dominates $\varphi(P_i', P_{-i})$ according to $P_i$.
Let $r_1(P_i) = a_s$ and $r_1(P_i') = a_t$.
Evidently, if $a_s = a_t$, the tops-only property implies $\varphi(P_i, P_{-i}) = \varphi(P_i', P_{-i})$.
Henceforth, assume $a_s \neq a_t$.
Then, $P_i \sim P_i'$ implies $r_1(P_i) = r_{2}(P_i') = a_s$, $r_1(P_i') = r_2(P_i) = a_t$ and $r_{k}(P_i) = r_{k}(P_i')$ for all $k = 3, \dots, m$.
Thus, to show local strategy-proofness, it suffices to show the following condition:
		\begin{eqnarray*}
\qquad \qquad\quad
\begin{array}{l}
		\varphi_{a_{s}}(P_i, P_{-i}) \geq \varphi_{a_{s}}(P_i', P_{-i})\; \textrm{or}\;
\varphi_{a_{t}}(P_i, P_{-i}) \leq \varphi_{a_{t}}(P_i', P_{-i}),\;\textrm{and}\\
        \varphi_{a_{k}}(P_i, P_{-i}) = \varphi_{a_{k}}(P_i', P_{-i})\;
		\textrm{for all}\; a_k \notin \{a_s, a_t\}.
\end{array}\qquad\;\,   \;\,\;\,\;\,\;\,\;\,\;\,(\#)
		\end{eqnarray*}
Moreover, by $(\underline{k}, \overline{k})$-hybridness,
		$P_i\sim P_i'$ implies one of the following three cases: (i) $a_s, a_t \in L$ and $|s-t| =1$,
		(ii) $a_s, a_t \in R$ and $|s-t|=1$, or (iii) $a_s, a_t \in M$.
		Note that the first two cases are symmetric. Therefore, we focus on cases (i) and (iii).
		\medskip
		
		\noindent
		\textsc{Claim 1}:
		In case (i), condition (\#) holds.
		
		\medskip
		If $s-t=1$ (equivalently, $a_t = a_{s-1}$), we know
		$S(s, (P_i, P_{-i})) \supset S(s, (P_i', P_{-i}))$ and
		$S(k, (P_i, P_{-i})) = S(k, (P_i', P_{-i}))$ for all $a_{k} \in A\backslash \{a_s\}$, and
		derive
		\begin{align*}
		\varphi_{a_{s}}(P_i, P_{-i})
		&= \beta_{S(s, (P_i, P_{-i}))}([a_{s}, a_m]) - \beta_{S(s+1, (P_i, P_{-i}))}([a_{s+1}, a_m])\\
		&\geq \beta_{S(s, (P_i', P_{-i}))}([a_{s}, a_m]) - \beta_{S(s+1, (P_i', P_{-i}))}([a_{s+1}, a_m])\quad \textrm{by monotonicity}\\
		&= \varphi_{a_{s}}(P_i', P_{-i}),
		\end{align*}
		and for all $a_k \notin \{a_{s-1}, a_s\}$,
		\begin{align*}
		\varphi_{a_{k}}(P_i, P_{-i}) &= \beta_{S(k, (P_i, P_{-i}))}([a_{k}, a_m]) - \beta_{S(k+1, (P_i, P_{-i}))}([a_{k+1}, a_m])\\
		&= \beta_{S(k, (P_i', P_{-i}))}([a_{k}, a_m]) - \beta_{S(k+1, (P_i', P_{-i}))}([a_{k+1}, a_m])\\
		&= \varphi_{a_{k}}(P_i', P_{-i}).
		\end{align*}
		
		If $s-t=-1$ (equivalently, $a_t = a_{s+1}$), we know $S(s+1, (P_i, P_{-i})) \subset S(s+1, (P_i', P_{-i}))$ and
		$S(k, (P_i, P_{-i})) = S(k, (P_i', P_{-i}))$ for all $a_{k} \in A\backslash \{a_{s+1}\}$, and derive
		\begin{align*}
		\varphi_{a_{s+1}}(P_i, P_{-i})
		&= \beta_{S(s+1, (P_i, P_{-i}))}([a_{s+1}, a_m]) - \beta_{S(s+2, (P_i, P_{-i}))}([a_{s+2}, a_m])\\
		&\leq \beta_{S(s+1, (P_i', P_{-i}))}([a_{s+1}, a_m]) - \beta_{S(s+2, (P_i', P_{-i}))}([a_{s+2}, a_m])\quad \textrm{by monotonicity}\\
		&= \varphi_{a_{s+1}}(P_i', P_{-i}).
		\end{align*}
		and for all $a_k \notin \{a_s, a_{s+1}\}$,
		\begin{align*}
		\varphi_{a_{k}}(P_i, P_{-i})
		&= \beta_{S(k, (P_i, P_{-i}))}([a_{k}, a_m]) - \beta_{S(k+1, (P_i, P_{-i}))}([a_{k+1}, a_m])\\
		&= \beta_{S(k, (P_i', P_{-i}))}([a_{k}, a_m]) - \beta_{S(k+1, (P_i', P_{-i}))}([a_{k+1}, a_m])\\
		&= \varphi_{a_{k}}(P_i', P_{-i}).
		\end{align*}
		This completes the verification of the claim.
		\medskip

		\noindent
		\textsc{Claim 2}:
		In case (iii), condition (\#) holds.
		
		\medskip
		We assume $a_t \prec a_s$. The verification related to the situation $a_s \prec a_t$ is symmetric, and we hence omit it.
		First, note that for all $a_k \in A$ such that $a_{k} \preceq a_t$ or $a_s \prec a_k$, $S(a_{k}, (P_i, P_{-i})) = S(a_{k}, (P_i', P_{-i}))$,
and hence we have
		\begin{align*}
		\varphi_{a_{k}}(P_i, P_{-i})
		&= \beta_{S(k, (P_i, P_{-i}))}([a_k, a_m]) - \beta_{S(k+1, (P_i, P_{-i}))}([a_{k+1}, a_m])\\
		&= \beta_{S(k, (P_i', P_{-i}))}([a_k, a_m]) - \beta_{S(k+1, (P_i', P_{-i}))}([a_{k+1}, a_m])
		= \varphi_{a_{k}}(P_i', P_{-i}).
		\end{align*}
		
		Next, given $a_t \prec a_k \prec a_s$,
		it is clear that $a_{\underline{k}} \prec a_k \prec a_{\overline{k}}$ and $a_{\underline{k}} \prec a_{k+1} \preceq a_{\overline{k}}$.
		Since the constrained random-dictatorship condition implies $\beta_S(a_v) = 0$ for all $S \subseteq N$ and $a_v \in [a_{\underline{k}+1}, a_{\overline{k}-1}]$, we have $\beta_{S}([a_k, a_m]) = \sum_{j \in S}\varepsilon_j = \beta_{S}([a_{k+1}, a_m])$.
		Moreover, note that $S(k, (P_i, P_{-i}))\backslash S(k+1, (P_i, P_{-i})) = \{j \in N \backslash \{i\}: r_1(P_j) = a_k\}
		= S(k, (P_i', P_{-i}))\backslash S(k+1, (P_i', P_{-i}))$.
		Hence, we have
		\begin{align*}
		\varphi_{a_k}(P_i, P_{-i})
		&=\beta_{S(k, (P_i, P_{-i}))}([a_{k}, a_m])- \beta_{S(k+1, (P_i, P_{-i}))}([a_{k+1}, a_m])\\
		&=\sum\nolimits_{j \in S(k, (P_i, P_{-i}))\backslash S(k+1, (P_i, P_{-i}))}\,\varepsilon_{j}\\
		&=\sum\nolimits_{j \in S(k, (P_i', P_{-i}))\backslash S(k+1, (P_i', P_{-i}))}\,\varepsilon_{j}\\
		&= \beta_{S(k, (P_i', P_{-i}))}([a_{k}, a_m]) - \beta_{S(k+1, (P_i', P_{-i}))}([a_{k+1}, a_m])\\
		&= \varphi_{a_k}(P_i', P_{-i}).
		\end{align*}
		Overall, $\varphi_{a_k}(P_i, P_{-i}) = \varphi_{a_k}(P_i', P_{-i})$ for all $a_k \notin \{a_s, a_t\}$.

		Last, since $a_t \prec a_s$, we know $S(s, (P_i, P_{-i}))\supset S(s, (P_i', P_{-i}))$ and
		$S(a_{s+1}, (P_i, P_{-i})) = S(a_{s+1}, (P_i', P_{-i}))$, and derive
		\begin{align*}
		\varphi_{a_{s}}(P_i, P_{-i})
		&= \beta_{S(s, (P_i, P_{-i}))}([a_{s}, a_m]) - \beta_{S(s+1, (P_i, P_{-i}))}([a_{s+1}, a_m])\\
		&\geq \beta_{S(s, (P_i', P_{-i}))}([a_{s}, a_m]) - \beta_{S(s+1, (P_i', P_{-i}))}([a_{s+1}, a_m])\quad
		\textrm{by monotonicity}\\
		&= \varphi_{a_{s}}(P_i', P_{-i}).
		\end{align*}
		This completes the verification of the claim.
		\medskip
		
		In conclusion, $\varphi$ is locally strategy-proof, as required.
		This completes the verification of the sufficiency part.
		\medskip

\noindent
(\textbf{Necessity}) The proof of the necessity part consists of three steps which are mentioned right below Theorem \ref{thm:rulecharacterization}.
\textbf{Step 1} focuses on the hybrid domain $\mathbb{D}_{\prec}(\underline{k}, \overline{k})$, and consists of the following lemma.

%In the first step, we restrict attention to the hybrid domain $\mathbb{D}_{\prec}(\underline{k}, \overline{k})$ and show that every unanimous and strategy-proof RSCF defined on $\mathbb{D}_{\prec}(\underline{k}, \overline{k})$ is a $(\underline{k}, \overline{k})$-PFBR.
%In the second step, we return to the regular domain $\mathbb{D}$ and partially characterize unanimous and strategy-proof RSCFs defined on $\mathbb{D}$.
%In the last step, we fix a unanimous and strategy-proof RSCF $\varphi: \mathbb{D}^n \rightarrow \Delta(A)$,
%transfer it to a \emph{random voting scheme} $\varphi: A^n \rightarrow \Delta(A)$ by the result established in the second step and
%show that the random voting scheme is unanimous and strategy-proof on the hybrid domain $\mathbb{D}_{\prec}(\underline{k}, \overline{k})$ and therefore is a $(\underline{k}, \overline{k})$-PFBR by the characterization result established in the first step.

\begin{lemma}\label{lem:PFBR}	
Given a unanimous and locally strategy-proof RSCF $\psi: \left[\mathbb{D}_{\prec}(\underline{k}, \overline{k})\right]^{n} \rightarrow \Delta(A)$, it is a $(\underline{k}, \overline{k})$-PFBR.
\end{lemma}
\begin{proof}
First, recall that $\mathbb{D}_{\prec}(\underline{k}, \overline{k})$ is shown to be a no-restoration domain in Clarification 1 of Appendix \ref{app:clarification}.
        Then, Theorem 1 and Corollary 1 of \citet{KRSYZ2021b} imply that $\psi$ satisfies strategy-proofness and the tops-only property respectively.
Next, since $\mathbb{D}_{\prec} \subseteq \mathbb{D}_{\prec}(\underline{k}, \overline{k})$,
		we can elicit a unanimous and strategy-proof RSCF $\phi: [\mathbb{D}_{\prec}]^{n} \rightarrow \Delta(A)$ such that $\phi(P) = \psi(P)$ for all  $P \in [\mathbb{D}_{\prec}]^{n}$.
		Then, by the characterization theorem of \citet{EPS2002}, we know that $\phi$ is a PFBR.
        Let $(\beta_S)_{S \subseteq N}$ be the corresponding probabilistic ballots.
        Evidently, $(\beta_S)_{S \subseteq N}$ satisfies ballot unanimity and monotonicity.
		Last, since both $\mathbb{D}_{\prec}$ and $\mathbb{D}_{\prec}(\underline{k}, \overline{k})$ are minimally rich,
		by the tops-only property of $\psi$ and the construction of $\phi$, it is true that $\psi$ is also a PFBR and inherits the probabilistic ballots $(\beta_S)_{S \subseteq N}$.
		Therefore, for all $P \in \left[\mathbb{D}_{\prec}(\underline{k}, \overline{k})\right]^{n}$ and $a_k \in A$, we have
		$\psi_{a_k}(P) = \beta_{S(k, P)}([a_k, a_m]) - \beta_{S(k+1, P)}([a_{k+1}, a_m])$,
        where $S(m+1, P) \equiv \emptyset$, $[a_{m+1}, a_m] \equiv \emptyset$ and $\beta_{S(m+1, P)}([a_{m+1}, a_m])=0$.

		To complete the verification, we show that $(\beta_S)_{S \subseteq N}$ satisfy the constrained random-dictatorship condition w.r.t.~$a_{\underline{k}}$ and $a_{\overline{k}}$.
		Let $\overline{\mathbb{D}}_{\prec}(\underline{k}, \overline{k}) = \left\{P_i \in \mathbb{D}_{\prec}(\underline{k}, \overline{k}): r_1(P_i) \in M\right\}$
		denote the subdomain of preferences which have peaks in $M$.
		First, since $|M| \geq 3$ and $\overline{\mathbb{D}}_{\prec}(\underline{k}, \overline{k})$ has no restriction on the ranking of alternatives in $M$,
		according to the random dictatorship characterization theorem of \citet{G1977},
		we easily infer that
		there exists a ``conditional dictatorial coefficient'' $\varepsilon_i \geq 0$ for each $i \in N$ with $\sum_{i \in N} \varepsilon_i = 1$ such that
		$\psi(P) = \sum_{i \in N}\varepsilon_i \, \bm{e}_{r_1(P_i)}$ for all $P \in \left[\,\overline{\mathbb{D}}_{\prec}(\underline{k}, \overline{k})\right]^{n}$. Next, fixing an arbitrary coalition $S \subseteq N$,
we show $\beta_{S}([a_{\overline{k}}, a_m]) = \sum_{j \in S}\varepsilon_j$ and $\beta_{S}([a_1, a_{\underline{k}}]) = \sum_{j \in N\backslash S}\varepsilon_j$.
We construct a profile $P  \in \left[\,\overline{\mathbb{D}}_{\prec}(\underline{k}, \overline{k})\right]^n$
where every voter of $S$ has a preference with the peak $a_{\overline{k}}$ and every voter out of $S$ has a preference with the peak $a_{\underline{k}}$.
Thus, $S = S(\overline{k}, P)$ and $\psi(P) = \sum_{j \in S} \varepsilon_j\, \bm{e}_{a_{\overline{k}}}+ \sum_{j \in N \backslash S} \varepsilon_j\, \bm{e}_{a_{\underline{k}}}$.
Furthermore, as a PFBR, we derive
$\beta_S([a_{\overline{k}}, a_m]) = \beta_{S(\overline{k},P)}([a_{\overline{k}}, a_m])
= \sum\nolimits_{k=\overline{k}}^{m}\big[\beta_{S(k,P)}([a_{k}, a_m])-\beta_{S(k+1,P)}([a_{k+1}, a_m])\big]
= \sum\nolimits_{k=\overline{k}}^{m}\psi_{a_k}(P)
= \psi_{a_{\overline{k}}}(P)=\sum\nolimits_{j \in S}\varepsilon_j$, as required.
Next, given arbitrary $a_k \in [a_{\underline{k}+1},a_{\overline{k}-1}]$,
since $S(k, P) =  S = S(k+1, P)$, we have
$\beta_{S}(a_k)=\beta_{S}([a_k, a_m]) - \beta_{S}([a_{k+1}, a_m])
= \beta_{S(k, P)}([a_k, a_m]) - \beta_{S(k+1, P)}([a_{k+1}, a_m])
= \psi_{a_k}(P)  = 0$.
Thus, $\beta_{S}([a_1, a_{\underline{k}}]) = 1 - \beta_{S}([a_{\overline{k}}, a_{m}]) - \beta_{S}([a_{\underline{k}+1}, a_{\overline{k}-1}])
= 1 - \beta_{S}([a_{\overline{k}}, a_{m}]) = \sum_{j \in N\backslash S}\varepsilon_j$, as required.
This proves the lemma, and completes the verification of \textbf{Step 1}.
\end{proof}

\medskip
\noindent
\textbf{Step 2} considers the regular domain $\mathbb{D}$ of Theorem 2, and
establishes four lemmas (Lemmas \ref{lem:uncompromising} - \ref{lem:refineuncompromise}), which together partially characterize all unanimous and strategy-proof RSCFs on $\mathbb{D}$.

Since $\mathbb{D}$ is a no-restoration domain,
by Corollary 1 of \citet{KRSYZ2021b}, we know that every unanimous and strategy-proof RSCF defined on $\mathbb{D}$ satisfies the tops-only property.
Therefore, all unanimous and strategy-proof RSCFs investigated in Lemmas \ref{lem:uncompromising} - \ref{lem:refineuncompromise} satisfy the tops-only property as well.
Henceforth, for notational convenience,
with a little notational abuse, we write $(a_s, a_t)$ as a two-voter preference profile where the first voter presents a preference with peak $a_s$ while the second reports a preference with peak $a_t$.
We also write $(a_s, P_{-i})$ as an $n$-voter preference profile where voter $i$ presents a preference with peak $a_s$ and $P_{-i} = (P_1, \dots, P_{i-1}, P_{i+1}, \dots, P_n)$.

\begin{lemma}[The uncompromising property]\label{lem:uncompromising}
	Let $\varphi: \mathbb{D}^{n} \rightarrow \Delta(A)$ be a unanimous and strategy-proof RSCF.
	Given a vertex-path $(x_1, \dots, x_t)$ in $G_{\approx}^A(\mathbb{D})$, $i \in N$ and $P_{-i} \in \mathbb{D}^{n-1}$,
	we have $\varphi_{a_s}(x_1, P_{-i}) = \varphi_{a_s}(x_t, P_{-i})$ for all $a_s \notin \{x_{k}\}_{k=1}^{t}$, and hence
	$\sum_{k=1}^{t}\varphi_{x_k}(x_1, P_{-i}) = \sum_{k=1}^{t}\varphi_{x_k}(x_t, P_{-i})$.
\end{lemma}

\begin{proof}
	We start with $\varphi(x_1, P_{-i})$ and $\varphi(x_2, P_{-i})$.
	Since $x_1 \approx x_2$,
	we have $P_{i} ,P_{i}' \in \mathbb{D}$ such that $r_1(P_i) =r_2(P_i') = x_1$, $r_1(P_i') = r_2(P_i) = x_2$ and $r_l(P_i) = r_l(P_i)$ for all $l = 3, \dots, m$.
	Then, the tops-only property and strategy-proofness imply
	$\varphi_{a_s}(x_1, P_{-i}) =\varphi_{a_s}(P_{i}, P_{-i}) =  \varphi_{a_s}(P_{i}', P_{-i}) = \varphi_{a_s}(x_2, P_{-i})$ for all $a_s \notin \{x_{1}, x_2\}$.
	
	We next introduce an induction hypothesis: given $2<k\leq t$, for all $2\leq k'< k$,
	$\varphi_{a_s}(x_1, P_{-i}) = \varphi_{a_s}(x_{k'}, P_{-i})$ for all $a_s \notin \{x_{l}\}_{l=1}^{k'}$.
	We show $\varphi_{a_s}(x_1, P_{-i}) = \varphi_{a_s}(x_{k}, P_{-i})$ for all $a_s \notin \{x_{l}\}_{l=1}^{k}$.
	Since $x_{k} \approx x_{k-1}$,
	we have $P_{i} ,P_{i}' \in \mathbb{D}$ such that $r_1(P_i) =r_2(P_i') = x_k$, $r_1(P_i') = r_2(P_i) = x_{k-1}$ and $r_l(P_i) = r_l(P_i)$ for all $l = 3, \dots, m$.
	Then, the tops-only property and strategy-proofness imply
	$\varphi_{a_s}(x_k, P_{-i}) =\varphi_{a_s}(P_{i}, P_{-i}) =  \varphi_{a_s}(P_{i}', P_{-i}) = \varphi_{a_s}(x_{k-1}, P_{-i})$ for all $a_s \notin \{x_{k-1}, x_k\}$.
	Moreover, since $\varphi_{a_s}(x_1, P_{-i}) = \varphi_{a_s}(x_{k-1}, P_{-i})$ for all $a_s \notin \{x_{l}\}_{l=1}^{k-1}$ by the induction hypothesis,
	it is true that $\varphi_{a_s}(x_1, P_{-i}) = \varphi_{a_s}(x_{k}, P_{-i})$ for all $a_s \notin \{x_{l}\}_{l=1}^{k}$.
	This completes the verification of the induction hypothesis.
	Therefore, $\varphi_{a_s}(x_1, P_{-i}) = \varphi_{a_s}(x_t, P_{-i})$ for all $a_s \notin \{x_{k}\}_{k=1}^{t}$.
	Then, we have
	$\sum_{k=1}^{t}\varphi_{x_k}(x_1, P_{-i}) = 1- \sum_{a_s \notin  \{x_{k}\}_{k=1}^{t}}\varphi_{a_s}(x_1, P_{-i})
	= 1- \sum_{a_s \notin  \{x_{k}\}_{k=1}^{t}}\varphi_{a_s}(x_t, P_{-i}) = \sum_{k=1}^{t}\varphi_{x_k}(x_t, P_{-i})$.
\end{proof}

\medskip
By statement (i) of Theorem \ref{thm:rulecharacterization}, we know that $\mathbb{D}$ is a $(\underline{k}, \overline{k})$-hybrid domain.

\begin{lemma}\label{lem:bounded}
Let $\varphi: \mathbb{D}^n \rightarrow \Delta(A)$ be a unanimous and strategy-proof RSCF.
For all $P \in \mathbb{D}^n$, we have $[r_1(P_i) \in M\; \textrm{for all}\; i \in N] \Rightarrow \big[\sum_{a_k \in M}\varphi_{a_k}(P) = 1\big]$.
\medskip
\end{lemma}

\begin{proof}
We first establish a claim.\medskip

\noindent
\textsc{Claim 1}:
Given $i \in N$, $P_{-i}\in \mathbb{D}^{n-1}$ with $r_1(P_1), \dots, r_1(P_{i-1}), r_1(P_{i+1}), \dots, r_1(P_n) \in M$, and $a_s, a_t \in M$ with
$a_s \approx a_t$,
if $\sum_{a_k \in M}\varphi_{a_k}(a_s, P_{-i}) = 1$, then $\sum_{a_k \in M}\varphi_{a_k}(a_t, P_{-i}) = 1$.\medskip

Since $a_s, a_t \in M$ and $a_s \approx a_t$,
we have $P_i, P_i' \in \mathbb{D}$ such that
$r_1(P_i) = r_2(P_i') = a_s$, $r_1(P_i') = r_2(P_i) = a_t$ and $r_k(P_i) = r_k(P_i')$ for all $k = 3, \dots, m$.
Clearly, strategy-proofness implies
$\varphi_{a_s}(P_i, P_{-i})+\varphi_{a_t}(P_i, P_{-i}) = \varphi_{a_s}(P_i', P_{-i})+\varphi_{a_t}(P_i', P_{-i})$ and
$\varphi_{a_k}(P_i, P_{-i}) = \varphi_{a_k}(P_i', P_{-i})$ for all $a_k \in A\backslash \{a_s, a_t\}$.
Then, by the tops-only property, we have
\begin{align*}
\sum\nolimits_{a_k \in M}\varphi_{a_k}(a_t, P_{-i})
&=\sum\nolimits_{a_k \in M}\varphi_{a_k}(P_i', P_{-i})\\
&= \sum\nolimits_{a_k \in M\backslash \{a_s, a_t\}}\varphi_{a_k}(P_i', P_{-i})+\Big[\varphi_{a_s}(P_i', P_{-i})+\varphi_{a_t}(P_i', P_{-i})\Big]\\
&= \sum\nolimits_{a_k \in M\backslash \{a_s, a_t\}}\varphi_{a_k}(P_i, P_{-i})+\Big[\varphi_{a_s}(P_i, P_{-i})+\varphi_{a_t}(P_i, P_{-i})\Big]\\
&= \sum\nolimits_{a_k \in M}\varphi_{a_k}(P_i, P_{-i}) \\
& = \sum\nolimits_{a_k \in M}\varphi_{a_k}(a_s, P_{-i}) \\
&= 1.
\end{align*}
This completes the verification of the claim.
\medskip

Now, fixing an arbitrary profile $P \in \mathbb{D}^n$ such that $r_1(P_i) \in M$ for all $i \in N$, we show $\sum_{a_k \in M}\varphi_{a_k}(P) = 1$.
Pick a profile $P' \in \mathbb{D}^n$ such that $r_1(P_1') =\dots = r_1(P_n') \equiv a_l \in M$.
Clearly, unanimity implies $\sum_{a_k \in M}\varphi_{a_k}(P') = 1$.
According to $P'$ and $P$, for notational convenience,
we assume $r_1(P_i) \neq a_l$ for $i = 1, \dots, \eta$, $r_1(P_j) = a_l$ for $j = \eta+1, \dots, n$, and
let $S = \{1, \dots, \eta\}$.
Since $\mathbb{D}$ is a $(\underline{k}, \overline{k})$-hybrid domain,
the path-inclusion condition implies that $G_{\approx}^{M}(\mathbb{D})$ is a connected graph.
Let $(x_1, \dots, x_v)$ be a vertex-path in $G_{\approx}^{M}(\mathbb{D})$ connecting $r_1(P_1)$ and $a_l$.
By repeatedly applying Claim 1 on the vertex-path $(x_1, \dots, x_v)$ and the tops-only property,
we have
\begin{align*}
\sum\nolimits_{a_k \in M}\varphi_{a_k}(P_1,P_{S\backslash \{1\}}, P_{N\backslash S})
& =\sum\nolimits_{a_k \in M}\varphi_{a_k}(x_1,P_{S\backslash \{1\}}, P_{N\backslash S}') \\
& = \sum\nolimits_{a_k \in M}\varphi_{a_k}(x_2,P_{S\backslash \{1\}}, P_{N\backslash S}') \\
&~\, \vdots\\
&=
\sum\nolimits_{a_k \in M}\varphi_{a_k}(x_v,P_{S\backslash \{1\}}, P_{N\backslash S}')\\
&= \sum\nolimits_{a_k \in M}\varphi_{a_k}(P_1',P_{S\backslash \{1\}}, P_{N\backslash S}').
\end{align*}
Continuing in this manner for voters $2, \dots, \eta$ step by step,
we eventually have
\begin{align*}
\sum\nolimits_{a_k \in M}\varphi_{a_k}(P_S, P_{N\backslash S})
&= \sum\nolimits_{a_k \in M}\varphi_{a_k}(P_1',P_{S\backslash \{1\}}, P_{N\backslash S}')\\
&~\, \vdots\\
& =\sum\nolimits_{a_k \in M}\varphi_{a_k}(P_1',\dots, P_l', P_{S\backslash \{1, \dots, l\}}, P_{N\backslash S}')\\
&=\sum\nolimits_{a_k \in M}\varphi_{a_k}(P_1',\dots, P_l', P_{l+1}', P_{S\backslash \{1, \dots, l,\, l+1\}}, P_{N\backslash S}')\\
&~\, \vdots\\
&= \sum\nolimits_{a_k \in M}\varphi_{a_k}(P_1',\dots, P_{\eta}', P_{N\backslash S}') \\
&=\sum\nolimits_{a_k \in M}\varphi_{a_k}(P')\\
&= 1.
\end{align*}
This proves the lemma.
\end{proof}

\begin{lemma}\label{lem:rd}
	Every unanimous and strategy-proof RSCF $\varphi: \mathbb{D}^{n} \rightarrow \Delta(A)$ behaves like a random dictatorship on
	$M$, i.e.,
	there exists a conditional dictatorial coefficient $\varepsilon_i \geq 0$ for each $i \in N$ with $\sum_{i \in N}\varepsilon_i = 1$ such that
for all $P \in \mathbb{D}^n$, we have $[r_1(P_i) \in M\; \textrm{for all}\; i \in N] \Rightarrow \big[\varphi(P) =\sum\nolimits_{i \in N}\varepsilon_i\, \bm{e}_{r_1(P_i)}\big]$.
\end{lemma}

\begin{proof}
	We verify the lemma in two steps.
	In the first step, we restrict attention to the case $n = 2$, i.e., $N = \{1, 2\}$, and show by Claims 1 - 4 below that every two-voter unanimous and strategy-proof RSCF on $\mathbb{D}$ behaves like a random dictatorship on $M$.
	In the second step, we extend to the case $n>2$ by adopting the Ramification Theorem of \citet{CSZ2014}.

	Fix a unanimous and strategy-proof RSCF $\varphi: \mathbb{D}^2 \rightarrow \Delta(A)$.
	It is clear that $\varphi$ satisfies the tops-only property.\medskip

	\noindent
	\textsc{Claim 1}: The following two statements hold:
	\begin{itemize}
		\item[\rm (i)] Given a vertex-path $(z_1, \dots, z_l)$ in $G_{\approx}^A(\mathbb{D})$, we have $\sum_{k=1}^{l}\varphi_{z_k}(z_1, z_l) = 1$.
		\item[\rm (ii)] Given a circle $(z_1, \dots, z_l, z_1)$ in $G_{\approx}^A(\mathbb{D})$, we have $\varphi_{z_s}(z_s, z_t) +\varphi_{z_t}(z_s, z_t) = 1$ for all $s\neq t$.\footnote{A vertex-path $(z_1, \dots, z_l)$, where $l \geq 3$, and an edge $(z_l, z_1)$ in $G_{\approx}^A(\mathbb{D})$ formulate a circle. For notational convenience, we simply write this circle as $(z_1, \dots, z_l, z_1)$.}
	\end{itemize}
	
	The first statement follows immediately from unanimity and the uncompromising property.
	Next, consider the circle $(z_1, \dots, z_l, z_1)$.
    Fixing $z_s$ and $z_t$, assume w.l.o.g.~that $s< t$.
	There are two vertex-paths connecting $z_s$ and $z_{t}$: the clockwise vertex-path $\mathcal{P}= (z_s, z_{s+1}, \dots, z_t)$ and
	the counter clockwise vertex-path $\mathcal{P}' = (z_{s}, z_{s-1}, \dots, z_1, z_l, z_{l-1},\dots, z_t)$.
	It follows immediately from statement (i) that $\sum_{z \in \mathcal{P}}\varphi_z(z_s, z_t) = 1$ and $\sum_{z \in \mathcal{P}'}\varphi_z(z_s, z_t) = 1$.
	Last, since the intersection of $\mathcal{P}$ and $\mathcal{P}'$ consists of $z_s$ and $z_{t}$,
	it is true that $\varphi_{z_s}(z_s, z_t) +\varphi_{z_t}(z_s, z_t) = 1$.
	This completes the verification of the claim.
	\medskip

Recall that $\mathbb{D}$ is a $(\underline{k}, \overline{k})$-hybrid domain.
By the path-inclusion condition and the no-leaf condition,
it is clear to infer that there exist two vertex-paths in $G_{\approx}^M(\mathbb{D})$ connecting
such that they start from $a_{\underline{k}}$ and end at $a_{\overline{k}}$, and
disagree on the the second vertices.
This further implies that $a_{\underline{k}}$ must be included in a circle of $G_{\approx}^M(\mathbb{D})$.
Symmetrically, $a_{\overline{k}}$ is also included in a circle of $G_{\approx}^M(\mathbb{D})$.
Let $\mathcal{C}_1 \equiv (x_1, \dots, x_p, x_1)$ be a circle of $G_{\approx}^M(\mathbb{D})$ that includes $a_{\underline{k}}$, and
$\mathcal{C}_2 \equiv (y_1, \dots, y_q, y_1)$ be a circle of $G_{\approx}^M(\mathbb{D})$ that includes $a_{\overline{k}}$.

\medskip
	
	\noindent
	\textsc{Claim 2}: RSCF
	$\varphi$ behaves like a random dictatorship on the circle ${\mathcal{C}_1}$, i.e.,
	there exists $0\leq \varepsilon \leq 1$ such that $\varphi(x_k, x_{k'}) = \varepsilon \bm{e}_{x_k}+(1-\varepsilon)\bm{e}_{x_{k'}}$ for all $1 \leq k, k' \leq p$.
	\medskip
	
	Claim 1(ii) first implies $\varphi_{x_1}(x_1,x_2)+\varphi_{x_2}(x_1,x_2) = 1$.
	Let $\varepsilon = \varphi_{x_1}(x_1, x_2)$ and $1-\varepsilon =\varphi_{x_2}(x_1, x_2)$.
	Fix another profile $(x_k, x_{k'})$.
	If $x_k = x_{k'}$, unanimity implies $\varphi(x_k, x_{k'}) = \varepsilon \bm{e}_{x_k}+(1-\varepsilon) \bm{e}_{x_{k'}}$.
	We next assume $x_k \neq x_{k'}$.
	There are four cases:
	(i) $x_1\neq x_k$ and $x_2 = x_{k'}$,
	(ii) $x_1= x_k$ and $x_2 \neq x_{k'}$,
	(iii) $x_1\neq x_k$, $x_2 \neq x_{k'}$ and $(x_k, x_{k'}) \neq (x_2, x_1)$, and
	(iv) $(x_k, x_{k'}) = (x_2, x_1)$.
	
	Since cases (i) and (ii) are symmetric, we focus on the verification of case (i).
	We first have $\varphi_{x_k}(x_k, x_2)+\varphi_{x_2}(x_k, x_2) = 1$ by Claim 1(ii).
We next show $\varphi_{x_2}(x_k, x_2) = 1-\varepsilon$.
	Note that there exists a vertex-path in $\mathcal{C}_1$ that connects $x_1$ and $x_k$, and excludes $x_2$.
	Then, the uncompromising property implies $\varphi_{x_2}(x_k, x_2) = \varphi_{x_2}(x_1, x_2) = 1-\varepsilon$, as required.
	
	In case (iii), we first know either $x_k \notin \{x_1, x_2\}$ or $x_{k'} \notin \{x_1, x_2\}$.
	Assume w.l.o.g.~that $x_k \notin \{x_1, x_2\}$. Then, by the verification of cases (i),
	from $(x_1, x_2)$ to $(x_k, x_2)$,
	we have $\varphi(x_k, x_2) = \varepsilon \bm{e}_{x_k}+(1-\varepsilon)\bm{e}_{x_2}$.
	Furthermore, by case (ii), from $(x_k, x_2)$ to $(x_k, x_{k'})$, we eventually have $\varphi(x_k, x_{k'}) = \varepsilon \bm{e}_{x_k}+(1-\varepsilon)\bm{e}_{x_{k'}}$.
	
	Last, in case (iv), since the circle $\mathcal{C}_1$ contains at least three alternatives,
	we first consider the profile $(x_3, x_2)$ and have $\varphi(x_3, x_2) = \varepsilon \bm{e}_{x_3}+(1-\varepsilon)\bm{e}_{x_2}$ by the verification of case (i).
	Next, according to the verification of case (iii), from $(x_3, x_2)$ to $(x_2, x_1)$,
	we induce $\varphi(x_2, x_1) = \varepsilon \bm{e}_{x_2}+(1-\varepsilon)\bm{e}_{x_1}$.
	This completes the verification of the claim.
	\medskip
	
	Symmetrically,
	$\varphi$ also behaves like a random dictatorship on the circle $\mathcal{C}_2$, i.e.,
	there exists $0 \leq \varepsilon'\leq 1$ such that $\varphi(y_k, y_{k'}) = \varepsilon' \bm{e}_{y_k}+(1-\varepsilon')\bm{e}_{y_{k'}}$
	for all $1 \leq k, k' \leq q$.
	\medskip
	
	\noindent
	\textsc{Claim 3}:
	We have (i) $\varepsilon = \varepsilon'$,
	(ii) $\varphi(a_{\underline{k}}, a_{\overline{k}}) = \varepsilon\, \bm{e}_{a_{\underline{k}}}+ (1-\varepsilon) \bm{e}_{a_{\overline{k}}}$,
	and (iii) $\varphi(a_{\overline{k}},a_{\underline{k}}) = \varepsilon\, \bm{e}_{a_{\overline{k}}}+ (1-\varepsilon) \bm{e}_{a_{\underline{k}}}$.
	
	\medskip
    If $a_{\underline{k}}$ and $a_{\overline{k}}$ are contained in the same circle $\mathcal{C}_1$ or $\mathcal{C}_2$,
    the claim holds evidently.
    Next, we assume that $a_{\underline{k}}$ is not included in $\mathcal{C}_2$, and $a_{\overline{k}}$ is not included in $\mathcal{C}_1$.
    Since $G_{\approx}^M(\mathbb{D})$ is a connected subgraph, loosely speaking, there must exist a vertex-path of $G_{\approx}^M(\mathbb{D})$ that connects the two circles $\mathcal{C}_1$ and $\mathcal{C}_2$.
	Precisely, we can construct a vertex-path $\mathcal{P} \equiv (z_1, z_2, \dots, z_{l-1},z_l)$ in $G_{\approx}^M(\mathbb{D})$ such that
	(i) $l \geq 3$, (ii) $z_1$ and $z_2$ are included in $\mathcal{C}_1$, and $a_{\underline{k}} \in \{z_1, z_2\}$, and
	(iii) $z_{l-1}$ and $z_l$ are included in $\mathcal{C}_2$, and $a_{\overline{k}} \in \{z_{l-1}, z_l\}$.
	First, Claim 2 and the uncompromising property imply $\varepsilon= \varphi_{z_1}(z_1, z_2) =\varphi_{z_1}(z_1, z_l)$ and
	$1-\varepsilon = \varphi_{z_1}(z_2, z_1) = \varphi_{z_1}(z_l, z_1)$.
	Symmetrically,
	$1-\varepsilon'= \varphi_{z_{l}}(z_{l-1}, z_l) = \varphi_{z_{l}}(z_1, z_l)$ and
	$\varepsilon'= \varphi_{z_l}(z_l, z_{l-1}) = \varphi_{z_l}(z_l, z_1)$.
	Thus, we have $\varepsilon + 1-\varepsilon'= \varphi_{z_1}(z_1, z_l) +\varphi_{z_{l}}(z_1, z_l)\leq 1$ which implies $\varepsilon\leq \varepsilon'$, and
	$1-\varepsilon + \varepsilon'= \varphi_{z_1}(z_l, z_1) +\varphi_{z_l}(z_l, z_1)\leq 1$ which implies $\varepsilon \geq \varepsilon'$.
	Therefore, $\varepsilon = \varepsilon'$. This completes the verification of item (i).
	
	Since items (ii) and (iii) are symmetric, we focus on showing item (ii).
	First, by the verification of item (i), we have $\varphi(z_1, z_l) = \varepsilon \,\bm{e}_{z_1} + (1-\varepsilon)\bm{e}_{z_l}$.
	Second, according to the vertex-path $\mathcal{P}$, the uncompromising property implies
	$\varphi_{z_l}(z_2, z_l) = \varphi_{z_l}(z_1, z_l)= 1-\varepsilon$ and $\varphi_{z_k}(z_2, z_{l}) = \varphi_{z_k}(z_1, z_{l}) = 0$ for all $2<k< l$.
	Moreover, since $\sum_{k =2}^l\varphi_{z_k}(z_2, z_l) = 1$ by Claim 1(i),
	we have $\varphi_{z_2}(z_2, z_l) = 1-\varphi_{z_l}(z_2, z_l) = \varepsilon$,
    and hence $\varphi(z_2, z_l) = \varepsilon \,\bm{e}_{z_2} + (1-\varepsilon)\bm{e}_{z_l}$.
	Symmetrically, we also have $\varphi(z_1, z_{l-1}) = \varepsilon \,\bm{e}_{z_1} + (1-\varepsilon)\bm{e}_{z_{l-1}}$.
    Recall that $a_{\underline{k}} \in \{z_1, z_2\}$ and $a_{\overline{k}} \in \{z_{l-1}, z_l\}$.
	We hence conclude that when $a_{\underline{k}} = z_1$ or $a_{\overline{k}} = z_{l}$,
	$\varphi(a_{\underline{k}}, a_{\overline{k}}) = \varepsilon \,\bm{e}_{a_{\underline{k}}} + (1-\varepsilon)\bm{e}_{a_{\overline{k}}}$.
	Last, we show that when $a_{\underline{k}} = z_2$ and $a_{\overline{k}} = z_{l-1}$,
	$\varphi(a_{\underline{k}}, a_{\overline{k}}) = \varepsilon \,\bm{e}_{a_{\underline{k}}} + (1-\varepsilon)\bm{e}_{a_{\overline{k}}}$.
	According to the vertex-path $\mathcal{P}$, the uncompromising property implies
    $\varphi_{a_{\underline{k}}}(a_{\underline{k}}, a_{\overline{k}}) =\varphi_{z_2}(z_2, z_{l-1}) = \varphi_{z_2}(z_2, z_l) = \varepsilon$ and
    $\varphi_{a_{\overline{k}}}(a_{\underline{k}}, a_{\overline{k}}) = \varphi_{z_{l-1}}(z_2, z_{l-1}) =\varphi_{z_{l-1}}(z_1, z_{l-1}) = 1-\varepsilon$, as required.
	This completes the verification of item (ii), and hence proves the claim.

	\medskip
	\noindent
	\textsc{Claim 4}:
	Given distinct $a_s, a_t \in M$,
	$\varphi(a_s, a_t) = \varepsilon \,\bm{e}_{a_s}+(1-\varepsilon)\bm{e}_{a_t}$.
	\medskip
	
	First, consider the situation that there exists a vertex-path $(x_1, \dots, x_l)$ in $G_{\approx}^M(\mathbb{D})$ that connects $a_{\underline{k}}$ and $a_{\overline{k}}$, and includes $a_s$ and $a_t$.
	By Claim 3, we first have $\varphi(x_1, x_l) = \varepsilon\, \bm{e}_{x_1}+ (1-\varepsilon)\bm{e}_{x_l}$ and
	$\varphi(x_l, x_1) = \varepsilon\, \bm{e}_{x_l}+(1-\varepsilon) \bm{e}_{x_1}$.
    Then, according to the vertex-path $(x_1, \dots, x_l)$, by repeatedly applying Claim 1(i) and the uncompromising property,
	we have $\varphi(x_k, x_{k'}) = \varepsilon\, \bm{e}_{x_k}+(1-\varepsilon)\bm{e}_{x_{k'}}$ for all distinct $1 \leq k, k' \leq l$.
	Hence, $\varphi(a_s, a_t) = \varepsilon\, \bm{e}_{a_s}+ (1-\varepsilon)\bm{e}_{a_t}$.
	
	Next, consider the situation that there exists no vertex-path of $G_{\approx}^M(\mathbb{D})$ that connects $a_{\underline{k}}$ and $a_{\overline{k}}$, and includes both $a_s $ and $a_t$.\footnote{For instance, see Figure \ref{fig:hybrid*}.}
	According to the path-inclusion condition, it must be the case that $a_s \notin \{a_{\underline{k}}, a_{\overline{k}}\}$ and
	$a_t \notin \{a_{\underline{k}}, a_{\overline{k}}\}$, and
moreover, there exists a vertex-path $(b_1, \dots, b_v)$ in $G_{\approx}^M(\mathbb{D})$ that connects $a_{\underline{k}}$ and $a_{\overline{k}}$, and includes $a_s$, and there exists
a vertex-path $(c_1, \dots, c_u)$ in $G_{\approx}^M(\mathbb{D})$ that connects $a_{\underline{k}}$ and $a_{\overline{k}}$, and includes $a_t$.
	Clearly, $a_s \notin \{c_{k}\}_{k=1}^{u}$ and $a_t \notin \{b_{k}\}_{k=1}^{v}$.
	Let $a_s = b_{p}$ for some $1< p < v$ and $a_t = c_{q}$ for some $1< q < u$.
	According to the vertex-paths $(b_{1}, \dots, b_{p})$ and $(c_1, \dots, c_{q})$,
	since $b_1 = c_1 = a_{\underline{k}}$, $b_p \notin \{c_1, \dots, c_{q}\}$ and $c_{q}\notin \{b_{1}, \dots, b_{p}\}$,
	we can identify $1\leq \eta < p$ and $1 \leq \nu < q$ such that $b_{\eta} = c_{\nu}$ and
	$\{b_{\eta+1}, \dots, b_p\}\cap \{c_{\nu+1}, \dots, c_{q}\} = \emptyset$.
	Then, we have the concatenated vertex-path $\mathcal{P} = (a_s= b_p,  \dots, b_{\eta}= c_{\nu},  \dots, c_{q} = a_t)$ in $G_{\approx}^M(\mathbb{D})$
	that connects $a_s$ and $a_t$.
	By the verification in the first situation, we have $\varphi_{b_p}(b_p, b_{\eta}) = \varepsilon$ and $\varphi_{c_{q}}(c_{\nu}, c_{q}) = 1-\varepsilon$.
	Furthermore, according to $\mathcal{P}$, the uncompromising property implies
	$\varphi_{a_s}(a_s, a_t) = \varphi_{b_p}(b_p, c_{q}) =\varphi_{b_p}(b_p, c_{\nu})=\varphi_{b_p}(b_p, b_{\eta})= \varepsilon$ and
	$\varphi_{a_t}(a_s, a_t) = \varphi_{c_{q}}(b_p, c_{q}) = \varphi_{c_{q}}(b_{\eta}, c_{q}) = \varphi_{c_{q}}(c_{\nu}, c_{q}) = 1-\varepsilon$.
	Therefore, $\varphi(a_s, a_t) = \varepsilon\, \bm{e}_{a_s}+ (1-\varepsilon)\bm{e}_{a_t}$.
	This completes the verification of the claim.
	\medskip
	
	In conclusion, every two-voter unanimous and strategy-proof RSCF on $\mathbb{D}$ behaves like a random dictatorship on $M$.
	For the general case $n > 2$, we adopt an induction argument.

\medskip
	\noindent
	\textsc{Induction Hypothesis}:
	Given $n \geq 3$, for all $2\leq n' < n$, every unanimous and strategy-proof $\psi: \mathbb{D}^{n'} \rightarrow \Delta(A)$
	behaves like a random dictatorship on $M$.
	\medskip
	
	Given a unanimous and strategy-proof RSCF $\varphi: \mathbb{D}^{n} \rightarrow \Delta(A)$, $n > 2$,
	we show that it behaves like a random dictatorship on $M$.
    Since Lemma \ref{lem:bounded} implies that no alternative out of $M$ receives a positive probability at a preference profile where all voters' preference peaks belong to $M$, we can adopt the proof of the Ramification Theorem of \citet{CSZ2014} to verify the induction hypothesis.
    More specifically, if $n \geq 4$, the verification follows exactly from Propositions 5 and 6 of \citet{CSZ2014}.
	When $n = 3$, i.e., $N = \{1, 2, 3\}$,
	analogous to Propositions 4 and 6 of \citet{CSZ2014},
	we split the verification into the following two parts:
	\begin{enumerate}
		\item There exists $\varepsilon_1, \varepsilon_2, \varepsilon_3 \geq 0$ with $\varepsilon_1+\varepsilon_2+\varepsilon_3 = 1$
		such that for all $(P_1,P_2,P_3) \in \mathbb{D}^{3}$ with $r_1(P_1), r_1(P_2), r_1(P_3) \in M$, we have
		$\left[P_{i} = P_{j}\; \textrm{for some distinct}\; i, j \in N\right] \Rightarrow \left[\varphi(P_1,P_2,P_3) =\varepsilon_1\, e_{r_1(P_{1})}+ \varepsilon_2\, e_{r_1(P_{2})}+\varepsilon_3\, e_{r_1(P_{3})}\right]$.
		
		\item For all $(P_1,P_2,P_3) \in \mathbb{D}^{3}$ with $r_1(P_1), r_1(P_2), r_1(P_3) \in M$, we have
		$\varphi(P_1,P_2,P_3) =\varepsilon_1\, e_{r_1(P_{1})}+ \varepsilon_2\, e_{r_1(P_{2})}+\varepsilon_3\, e_{r_1(P_{3})}$.
	\end{enumerate}
If the first part holds, we can directly apply Proposition 6 of \citet{CSZ2014} to establish the second.
Therefore, in the rest of the proof, we focus on showing the first part.\footnote{Proposition 4 of \citet{CSZ2014} is not applicable for verifying the first part
    since they impose an additional condition (see their Definition 18) that cannot be confirmed in the subdomain of $\mathbb{D}$ where all preference peaks belong to $M$.}
	
	According to $\varphi$, we first induce three two-voter RSCFs by merging two voters respectively:
	for all $P_{1}, P_{2}, P_{3} \in \mathbb{D}$,
	let $\psi^{1}(P_{1}, P_{2}) = \varphi(P_{1}, P_{2}, P_{2})$, $\psi^{2}(P_{1}, P_{2}) = \varphi(P_{1}, P_{2}, P_{1})$ and
	$\psi^{3}(P_{1}, P_{3}) = \varphi(P_{1}, P_{1}, P_{3})$.
	It is easy to verify that all $\psi^1, \psi^2$ and $\psi^3$ are unanimous and strategy-proof on $\mathbb{D}$.
	Therefore, the induction hypothesis implies that there exist $0\leq \varepsilon_1,\varepsilon_2, \varepsilon_3 \leq 1$ such that
	for all $(P_1,P_2,P_3) \in \mathbb{D}^{3}$ with $r_1(P_1), r_1(P_2), r_1(P_3) \in M$,
	$\psi^{1}(P_1, P_2) =\varepsilon_1\, \bm{e}_{r_1(P_{1})}+ (1-\varepsilon_1) \bm{e}_{r_1(P_{2})}$,
	$\psi^{2}(P_1, P_2) =(1-\varepsilon_2)\bm{e}_{r_1(P_{1})}+ \varepsilon_2\, \bm{e}_{r_1(P_{2})}$ and
	$\psi^{3}(P_1, P_3) =(1-\varepsilon_3)\bm{e}_{r_1(P_{1})}+ \varepsilon_3\, \bm{e}_{r_1(P_{3})}$.
	Note that to show the first part, it suffices to prove $\varepsilon_1+\varepsilon_2+\varepsilon_3=1$.
	
	Recall the circle $\mathcal{C}_1 = (x_1, \dots, x_p, x_1)$ identified before.
	First, according to the vertex-paths $(x_2, x_3)$, $(x_1, x_2)$ and $(x_1, x_p, \dots, x_4, x_3)$ in $\mathcal{C}_1$,
    the uncompromising property implies respectively that
    (i) $\varphi_{x_1}(x_1, x_2, x_3) = \varphi_{x_1}(x_1, x_2, x_2) = \psi_{x_1}^{1}(x_1, x_2) = \varepsilon_1$ and
    $\varphi_{a_s}(x_1, x_2, x_3) = \varphi_{a_s}(x_1, x_2, x_2) = \psi_{a_s}^{1}(x_1, x_2) = 0$ for all $a_s \notin \{x_1, x_2, x_3\}$,
    (ii) $\varphi_{x_3}(x_1, x_2, x_3) = \varphi_{x_3}(x_2, x_2, x_3) = \psi_{x_3}^{3}(x_2, x_3) = \varepsilon_3$, and
	(iii) $\varphi_{x_2}(x_1, x_2, x_3) = \varphi_{x_2}(x_3, x_2, x_3) = \psi_{x_2}^{2}(x_3, x_2) = \varepsilon_2$.
	Therefore, we have $\varepsilon_1+\varepsilon_2+\varepsilon_3 = \varphi_{x_1}(x_1, x_2, x_3)+ \varphi_{x_2}(x_1, x_2, x_3)+\varphi_{x_3}(x_1, x_2, x_3)+\sum_{a_s \notin \{x_1, x_2,x_3\}}\varphi_{a_s}(x_1, x_2, x_3)
= \sum_{a_s \in A}\varphi_{a_s}(x_1, x_2, x_3) = 1$, as required.
	This completes the verification of the induction hypothesis, and hence proves Lemma \ref{lem:rd}.
\end{proof}

\vspace{-1em}
\begin{lemma}\label{lem:refineuncompromise}
	Let $\varphi: \mathbb{D}^n \rightarrow \Delta(A)$ be a unanimous and strategy-proof RSCF.
	Given distinct $a_s, a_t \in M$ and $P_{-i} \in \mathbb{D}^{n-1}$, we have $\varphi_{a_k}(a_s, P_{-i}) = \varphi_{a_k}(a_{t}, P_{-i})$ for all $a_k \notin \{a_s, a_t\}$.
\end{lemma}

\begin{proof}
	Clearly, $\varphi$ satisfies the tops-only property.
	Moreover, Lemma \ref{lem:rd} shows that $\varphi$ behaves like a random dictatorship on $M$.\medskip
	
	\noindent
	\textsc{Claim 1}: We have that (i) if $a_{\underline{k}} \notin \{a_s, a_t\}$,
then $\varphi_{a_{\underline{k}}}(a_s, P_{-i}) = \varphi_{a_{\underline{k}}}(a_{t}, P_{-i})$, and (ii)
if $a_{\overline{k}} \notin \{a_s, a_t\}$, then $\varphi_{a_{\overline{k}}}(a_s, P_{-i}) = \varphi_{a_{\overline{k}}}(a_{t}, P_{-i})$.
\medskip
	
	By symmetry, we focus on showing item (i).
	Note that if there exists a vertex-path in $G_{\approx}^M(\mathbb{D})$ that connects $a_s$ and $a_t$ and excludes $a_{\underline{k}}$, then
	the uncompromising property implies $\varphi_{a_{\underline{k}}}(a_s, P_{-i}) = \varphi_{a_{\underline{k}}}(a_{t}, P_{-i})$.
	Therefore, to complete the verification, we construct such a vertex-path.
	If $a_s \neq a_{\overline{k}}$, then $a_s \notin \{a_{\underline{k}}, a_{\overline{k}}\}$.
Then, by Definition \ref{def:hybrid*}, we have a vertex-path $\mathcal{P}$ in $G_{\approx}^M(\mathbb{D})$ that connects $a_{\underline{k}}$ and $a_{\overline{k}}$, and includes $a_s$. Thus, the vertex-path $\langle a_s, a_{\overline{k}}\rangle^{\mathcal{P}}$ connects $a_{s}$ and $a_{\overline{k}}$, and excludes $a_{\underline{k}}$.
    If $a_s = a_{\overline{k}}$, the null vertex-path $(a_s)$ connects $a_{s}$ and $a_{\overline{k}}$, and excludes $a_{\underline{k}}$.
	Overall, we have a vertex-path in $G_{\approx}^M(\mathbb{D})$ that connects $a_{s}$ and $a_{\overline{k}}$, and excludes $a_{\underline{k}}$.
	Similarly, we have a vertex-path in $G_{\approx}^M(\mathbb{D})$ that connects $a_{t}$ and $a_{\overline{k}}$, and excludes $a_{\underline{k}}$.
	According to these two vertex-paths, we can construct a vertex-path in $G_{\approx}^M(\mathbb{D})$ that connects $a_s$ and $a_t$, and excludes $a_{\underline{k}}$, as required.
	This completes the verification of the claim.

	\medskip
	
	Since $a_s, a_t \in M$ and $G_{\approx}^M(\mathbb{D})$ is a connected subgraph,
    there exists a vertex-path $(x_1, \dots, x_p)$ in $G_{\approx}^M(\mathbb{D})$ connecting $a_s$ and $a_t$.
    According to the vertex-path $(x_1, \dots, x_p)$, the uncompromising property first implies $\varphi_{a_k}(a_s, P_{-i}) = \varphi_{a_k}(a_t, P_{-i})$ for all $a_k \notin \{x_{k}\}_{k=1}^{p}$.
    Therefore, to complete the proof, it suffices to show $\varphi_{x_k}(a_s, P_{-i}) = \varphi_{x_k}(a_t, P_{-i})$ for all $k = 2, \dots, p-1$.
    If $x_k \in \{a_{\underline{k}},  a_{\overline{k}} \}$, it follows immediately from Claim 1 that $\varphi_{x_k}(a_s, P_{-i}) = \varphi_{x_k}(a_t, P_{-i})$.
    Hence, we let $\Theta =\{x_2, \dots, x_{p-1}\}\backslash \{a_{\underline{k}}, a_{\overline{k}}\}$ and show $\varphi_{z}(a_s, P_{-i}) = \varphi_{z}(a_t, P_{-i})$ for all $z \in \Theta$.
	
	For notational convenience, let $i = n$.
	We partition $\{1, \dots, n-1\}$ into three parts: $\underline{S} = \{1, \dots, j\}$,
	$\overline{S} = \{j+1, \dots, l\}$ and $\hat{S} = \{l+1, \dots, n-1\}$, and assume w.l.o.g that
	$r_1(P_{1}), \dots, r_1(P_j) \in [a_1, a_{\underline{k}-1}]$,
	$r_1(P_{j+1}), \dots, r_1(P_{l}) \in [a_{\overline{k}+1},  a_m]$ and
	$r_1(P_{l+1}), \dots, r_1(P_{n-1}) \in [a_{\underline{k}}, a_{\overline{k}}]$.
	Note that if $l=0$, all voters' preference peaks according to the profiles $(a_s, P_{-i})$ and $(a_t, P_{-i})$ are located in $M$, and then
    Lemma \ref{lem:rd} implies $\varphi_{z}(a_s, P_{-n}) = \varphi_{z}(a_t, P_{-n})$ for all $z \in \Theta$.
	Next, assume $l>0$.
	We construct the following preference profiles:
\begin{align*}
P^{\eta} =& \Big(P_{1}, \dots, P_{\eta}, \frac{a_{\underline{k}}}{~\{\eta+1, \dots, j\}~}, \frac{~a_{\overline{k}}~}{\overline{S}}, P_{\hat{S}}, a_{s}\Big)\; \textrm{for all}\; \eta = 0, 1, \dots, j,\;\textrm{and} \\
P^{\nu} =& \Big(P_{\underline{S}},  P_{j+1}, \dots, P_{\nu}, \frac{a_{\overline{k}}}{~\{\nu+1, \dots, l\}~}, P_{\hat{S}}, a_{s}\Big)\; \textrm{for all}\; \nu = j+1, \dots, l.
\end{align*}
    Note that $P^{0} = \big(\frac{~a_{\underline{k}}~}{\underline{S}}, \frac{~a_{\overline{k}}~}{\overline{S}}, P_{\hat{S}}, a_{s}\big)$ and
    $P^{l} = (a_s, P_{-n})$. In other words, from $P^{0}$ to $P^{l}$, we first change step by step voter $\eta$'s preference that has the peak $a_{\underline{k}}$, back to $P_{\eta}$, and then change step by step voter $\nu$'s preference that has the peak $a_{\overline{k}}$, back to $P_{\nu}$.

	Given arbitrary $0 \leq \eta < j$, consider $P^{\eta}$ and $P^{\eta+1}$.
	Note that voter $\eta+1$ has the preference peak $a_{\underline{k}}$ at the profile $P^{\eta}$,
	and has the preference peak $r_1(P_{\eta+1}) = a_k \prec a_{\underline{k}}$ at the profile $P^{\eta+1}$.
	By Lemma \ref{lem:coincidence}, $(a_k, a_{k+1}, \dots, a_{\underline{k}})$ is the unique vertex-path in $G_{\approx}^A(\mathbb{D})$ that connects $a_k$ and $a_{\underline{k}}$, and excludes all alternatives of $\Theta$.
	Then, the uncompromising property implies
	$\varphi_{z}(P^{\eta}) = \varphi_{z}(P^{\eta+1})$ for all $z \in \Theta$.
	Therefore, we have $\varphi_{z}(P^{0}) = \dots = \varphi_{z}(P^{j})$ for all $z \in \Theta$.
    Next, given arbitrary $j\leq \nu< l$, consider $P^{\nu}$ and $P^{\nu+1}$.
	Note that voter $\nu+1$ has the preference peak $a_{\overline{k}}$ at $P^{\nu}$,
	and has the preference peak $r_1(P_{\nu+1}) = a_k \succ a_{\overline{k}}$ at $P^{\nu+1}$.
	By Lemma \ref{lem:coincidence}, $(a_{\underline{k}}, \dots, a_{k-1}, a_k)$ is the unique vertex-path in $G_{\approx}^A(\mathbb{D})$ that connects $a_{\overline{k}}$ and $a_k$, and excludes all alternatives of $\Theta$.
	Then, the uncompromising property implies
	$\varphi_{z}(P^{\nu}) = \varphi_{z}(P^{\nu+1})$ for all $z \in \Theta$.
	Therefore, we have $\varphi_{z}(P^{j}) = \dots = \varphi_{z}(P^{l})$ for all $z \in \Theta$.
	In conclusion, $\varphi_{z}\big(\frac{~a_{\underline{k}}~}{\underline{S}}, \frac{~a_{\overline{k}}~}{\overline{S}}, P_{\hat{S}}, a_{s}\big)
	= \varphi_{z}(P^{0}) =\dots = \varphi_{z}(P^{l}) = \varphi_{z}(a_s, P_{-n})$ for all $z \in \Theta$.
    Symmetrically, we also derive
	$\varphi_{z}\big(\frac{~a_{\underline{k}}~}{\underline{S}}, \frac{~a_{\overline{k}}~}{\overline{S}}, P_{\hat{S}}, a_t\big) = \varphi_{z}(a_t, P_{-n})$
	for all $z \in \Theta$.
	Last, since Lemma \ref{lem:rd} implies
	$\varphi_{z}\big(\frac{~a_{\underline{k}}~}{\underline{S}}, \frac{~a_{\overline{k}}~}{\overline{S}}, P_{\hat{S}}, a_s\big) =\varphi_{z}\big(\frac{~a_{\underline{k}}~}{\underline{S}}, \frac{~a_{\overline{k}}~}{\overline{S}}, P_{\hat{S}}, a_t\big)$ for all $z \in \Theta$,
we have $\varphi_{z}(a_s, P_{-n})=\varphi_{z}(a_t, P_{-n})$ for all $z \in \Theta$, as required.
This proves the lemma and completes the verification of \textbf{Step 2}.
\end{proof}

\noindent
\textbf{Step 3} consists of one lemma, and completes the verification of the necessity part.
We fix a unanimous and strategy-proof RSCF $\varphi: \mathbb{D}^n \rightarrow \Delta(A)$, where $\mathbb{D}$ is the regular domain of Theorem \ref{thm:rulecharacterization}.
Clearly, $\mathbb{D}$ is a minimally rich domain and a $(\underline{k}, \overline{k})$-hybrid domain, and $\varphi$ satisfies the tops-only property.
Accordingly, we can construct a well defined RSCF $\psi$ on the hybrid domain $\mathbb{D}_{\prec}(\underline{k}, \overline{k})$:
    given $(P_1, \dots, P_n) \in \big[\mathbb{D}_{\prec}(\underline{k}, \overline{k})\big]^n$,
    \begin{align*}
    \psi(P_1, \dots, P_n) = \varphi(\bar{P}_1, \dots, \bar{P}_n) \;\textrm{for all}\; (\bar{P}_1, \dots, \bar{P}_n) \in \mathbb{D}^n\; \textrm{such that}
     ~\,r_1(\bar{P}_1) = r_1(P_1), \dots, r_1(\bar{P}_n) = r_1(P_n).
    \end{align*}
    Evidently, by construction, $\psi$ inherits unanimity and the tops-only property of $\varphi$.
    According to Lemma \ref{lem:PFBR} in \textbf{Step 1}, we know that if $\psi$ is locally strategy-proof, then $\psi$ is a $(\underline{k}, \overline{k})$-PFBR.
    Moreover importantly, note that if $\psi$ is shown to be a $(\underline{k}, \overline{k})$-PFBR,
    the construction of $\psi$ immediately implies that $\varphi$ is a $(\underline{k}, \overline{k})$-PFBR as well.
    Therefore, to complete the verification of \textbf{Step 3}, we show in the following lemma that $\psi$ is locally strategy-proof.

\begin{lemma}
RSCF $\psi$ is locally strategy-proof.
\end{lemma}

\begin{proof}
Fix arbitrary $i \in N$, $P_i, P_i' \in \mathbb{D}_{\prec}(\underline{k}, \overline{k})$ with $P_{i} \sim P_{i}'$, and $P_{-i} \in \left[\mathbb{D}_{\prec}(\underline{k}, \overline{k})\right]^{n-1}$.
We that $\psi(P_i, P_{-i})$ stochastically dominates $\psi(P_i', P_{-i})$ according to $P_{i}$.
Let $r_1(P_{i}) = a_s$ and $r_1(P_{i}') = a_t$.
If $a_s = a_t$, the tops-only property of $\psi$ implies $\psi(P_i, P_{-i}) = \psi(P_i, P_{-i})$.
Next, assume $a_s \neq a_t$.
Then, $P_{i} \sim P_{i}'$ implies $r_1(P_{i}) = r_{2}(P_{i}') = a_s$, $r_1(P_{i}') = r_2(P_{i}) = a_t$ and $r_{k}(P_{i}) = r_{k}(P_{i}')$ for all $k = 3, \dots, m$.
Then,
it suffices to show $\psi_{a_s}(P_i, P_{-i}) \geq \psi_{a_s}(P_i', P_{-i})$ and
$\psi_{a_k}(P_i, P_{-i}) = \psi_{a_k}(P_i', P_{-i})$ for all $a_k \notin \{a_s, a_t\}$.
Since both $\mathbb{D}_{\prec}(\underline{k}, \overline{k})$ and $\mathbb{D}$ are minimally rich,
we have $\bar{P}_j \in \mathbb{D}$ for each $j \neq i$ such that $r_1(\bar{P}_j) = r_1(P_j)$.
Since $P_i, P_i' \in \mathbb{D}_{\prec}(\underline{k}, \overline{k})$,
$r_1(P_{i}) = a_s$, $r_1(P_{i}') = a_t$ and $P_{i} \sim P_{i}'$, we know that $a_s$ and $a_t$ are strongly connected (according to $\mathbb{D}_{\prec}(\underline{k}, \overline{k})$).
Then, the definition of $\mathbb{D}_{\prec}(\underline{k}, \overline{k})$ implies that one of the three cases must occur:
(i) $a_s, a_t \in L$ and $|s-t| = 1$, (ii) $a_s, a_t \in R$ and $|s-t| = 1$, or (iii) $a_s, a_t \in M$.
The first two cases are symmetric, and hence we focus on the verification of the first case.
In the first case, since $\mathbb{D}$ is a $(\underline{k}, \overline{k})$-hybrid domain,
$|s-t| = 1$ implies that $a_s$ and $a_t$ are also strongly connected (according to $\mathbb{D}$).
Hence, we have $\bar{P}_i, \bar{P}_{i}' \in \mathbb{D}$ such that $r_1(\bar{P}_i) =r_2(\bar{P}_{i}')= a_s$, $r_1(\bar{P}_{i}')=r_2(\bar{P}_i) =a_t$ and $r_k(\bar{P}_i) =r_k(\bar{P}_i') $ for all $k = 3, \dots, m$.
Then, by the construction of $\psi$ and strategy-proofness of $\varphi$, we have
$\psi_{a_s}(P_i, P_{-i}) = \varphi_{a_s}(\bar{P}_i, \bar{P}_{-i})  \geq \varphi_{a_s}(\bar{P}_i', \bar{P}_{-i}) = \psi_{a_s}(P_i', P_{-i})$, and
$\psi_{a_k}(P_i, P_{-i}) = \varphi_{a_k}(\bar{P}_{i}, \bar{P}_{-i}) = \varphi_{a_k}(\bar{P}_{i}', \bar{P}_{-i})= \psi_{a_k}(P_i', P_{-i})$ for all $a_k \notin \{a_s, a_t\}$, as required.
Last, let $a_s, a_t \in M$. Fix $\bar{P}_i ,\bar{P}_i' \in \mathbb{D}$ with $r_1(\bar{P}_i) = a_s$ and $r_1(\bar{P}_{i}')=a_t$ by minimal richness.
Then, by the construction of $\psi$ and strategy-proofness of $\varphi$,
we have $\psi_{a_s}(P_i, P_{-i}) =\varphi_{a_s}(\bar{P}_{i}, \bar{P}_{-i})\geq \varphi_{a_s}(\bar{P}_{i}', \bar{P}_{-i}) = \psi_{a_s}(P_i, P_{-i})$, as required.
Moreover, by the construction of $\psi$ and Lemma \ref{lem:refineuncompromise} established on $\varphi$,
we have $\psi_{a_k}(P_i, P_{-i}) =\varphi_{a_k}(\bar{P}_i, \bar{P}_{-i})=\varphi_{a_k}(\bar{P}_i', \bar{P}_{-i}) = \psi_{a_k}(P_i', P_{-i})$ for all $a_k \notin \{a_s, a_t\}$, as required.
Therefore, $\psi$ is locally strategy-proof.
This completes the verification of the lemma and hence proves the necessity part.
\end{proof}

\section{Two Clarifications}\label{app:clarification}

\noindent
\textbf{Clarification 1}:
For all $1 \leq \underline{k}< \overline{k}\leq m$,
the $(\underline{k}, \overline{k})$-hybrid domain $\mathbb{D}_{\prec}(\underline{k}, \overline{k})$ is a no-restoration domain.\medskip

\begin{proof}
Note that Propositions 4.1 and 4.2 of \citet{S2013} imply that both the complete domain and the single-peaked domain are no-restoration domains.
Since the hybrid domain $\mathbb{D}_{\prec}(\underline{k}, \overline{k})$
expands to the complete domain when $\overline{k}-\underline{k} = m-1$ and reduces to the single-peaked domain when $\overline{k}-\underline{k} = 1$,
the fact holds evidently in these two cases.
Henceforth, we assume $1<\overline{k}-\underline{k}<m-1$.
We introduce some new notation.
Given a preference $P_i$, let $a_sP_i!a_t$ denote that $a_sP_ia_t$ and there exists no $a_r \in A\backslash \{a_s, a_t\}$ such that $a_sP_ia_r$ and $a_rP_ia_t$.
		
For any pair of distinct preference $P_i, P_i' \in \mathbb{D}_{\prec} (\underline{k}, \overline{k})$,
if we can identify some pair of alternatives $a_p, a_q \in A$ such that
$a_pP_i!a_q$ and $a_qP_i'a_p$, construct a preference $P_i''$ by locally switching $a_p$ and $a_q$ in $P_i$ (i.e., $P_i \sim P_i''$ and $a_qP_i!a_p$), and verify $P_i'' \in \mathbb{D}_{\prec} (\underline{k}, \overline{k})$,
then we have a $(\underline{k}, \overline{k})$-hybrid preference $P_i''$ that is one-step closer to $P_i'$ than $P_i$.
Thus, given a pair of distinct preference $P_i, P_i' \in \mathbb{D}_{\prec} (\underline{k}, \overline{k})$,
by recursively applying the construction of new preferences, we eventually generate a path
that connects $P_i$ and $P_i'$, and more importantly has no $\{a_s, a_t\}$-restoration for \emph{all} $a_s, a_t \in A$.
This of course implies that the hybrid domain $\mathbb{D}_{\prec}(\underline{k}, \overline{k})$ is a no-restoration domain.
Therefore, in the rest of proof, we fix two distinct preferences
$P_i, P_i' \in \mathbb{D}_{\prec} (\underline{k}, \overline{k})$.
Then, we will identify some pair of alternatives $a_p, a_q \in A$ such that
$a_pP_i!a_q$ and $a_qP_i'a_p$, construct a preference $P_i''$ by locally switching $a_p$ and $a_q$ in $P_i$, and show that $P_i''$ is a $(\underline{k}, \overline{k})$-hybrid preference.
For notational convenience, let $r_1(P_i) = a_s$ and $r_1(P_i') = a_t$.
We first provide an observation that will be repeatedly applied.
\begin{observation}\label{obs}\rm
Given $a_p, a_q \in A$ such that $a_pP_i!a_q$ (it is possible that $a_p = a_s$),
let $P_i''$ be a preference by locally switching $a_p$ and $a_q$ in $P_i$ (i.e., $P_i \sim P_i''$ and $a_qP_i''!a_p$).
If one of the following three conditions is satisfied:
\begin{itemize}
\item[\rm (i)] $r_1(P_i) = r_1(P_i'')$, and $a_p \prec a_s \prec a_q$ or $a_q \prec a_s \prec a_p$,
\item[\rm (ii)] $r_1(P_i) = r_1(P_i'') \in M$, $\{a_p, a_q\} \nsubseteq L$ and $\{a_p, a_q\} \nsubseteq R$, and
\item[\rm (iii)] $r_1(P_i) \neq r_1(P_i'')$, and either $\{a_p, a_q\} \subseteq L$ and $|p-q| = 1$, or $\{a_p, a_q\} \subseteq R$ and $|p-q| = 1$, or $\{a_p, a_q\} \subseteq M$,
\end{itemize}
then it is true that $P_i''  \in \mathbb{D}_{\prec} (\underline{k}, \overline{k})$.
\hfill$\Box$
\end{observation}

According to $r_1(P_i) = a_s$ and $r_1(P_i') = a_t$, there are three situations: $a_s = a_t$, $a_s \prec a_t$ and $a_t \prec a_s$.
Note that the last two situations are symmetric, and we hence focus on the verification of the first two situations.
		
		We first assume $a_s = a_t$.
		We identify $1< k \leq m$ such that $r_l(P_i) = r_l(P_i')$ for all $l = 1, \dots, k-1$, and
		$r_k(P_i) \neq r_k(P_i')$. Let $r_k(P_i') = a_q$.
Then, it is clear that $a_q = r_{\nu}(P_i)$ for some $k<\nu\leq m$.
		Meanwhile, let $r_{\nu-1}(P_i) = a_p$.
        It is clear that $a_qP_i'a_p$.
		We then generate a preference $P_i''$ by locally switching $a_p$ and $a_q$ in $P_i$.
		Thus, we know $P_i \sim P_i''$, $a_pP_i!a_q$, $a_qP_i''!a_p$ and $a_qP_i'a_p$.
        Note that $r_1(P_i) = r_1(P_i'') = r_1(P_i')$.
		We next show $P_i'' \in \mathbb{D}_{\prec}(\underline{k}, \overline{k})$.
		Suppose not, i.e., $P_i'' \notin \mathbb{D}_{\prec}(\underline{k}, \overline{k})$.
		On the one hand, since $P_i$ and $P_i''$ share the same peak and differ exactly on the relative rankings of $a_p$ and $a_q$,
		$P_i \in \mathbb{D}_{\prec}(\underline{k}, \overline{k})$ and $P_i'' \notin \mathbb{D}_{\prec}(\underline{k}, \overline{k})$ imply that
        $a_qP_i''a_p$ violates the definition of $(\underline{k}, \overline{k})$-hybridness.
		On the other hand, since $P_i''$ and $P_i'$ share the same peak and the same relative ranking of $a_p$ and $a_q$,
		$P_i' \in \mathbb{D}_{\prec} (\underline{k}, \overline{k})$ implies that $a_qP_i''a_p$ does not violate the definition of $(\underline{k}, \overline{k})$-hybridness. Contradiction.
		Therefore, $P_i'' \in \mathbb{D}_{\prec}(\underline{k}, \overline{k})$.
		
		Next, we assume $a_s \prec a_t$.
		According to $a_{\underline{k}}$ and $a_{\overline{k}}$, we consider the four possible cases:
		(1) $a_s \prec a_{\underline{k}}$,
		(2) $a_{\overline{k}}\preceq a_s$,
		(3) $a_{\underline{k}}\preceq a_s \prec a_{\overline{k}}\preceq a_t$ and
		(4) $a_{\underline{k}}\preceq a_s \prec a_t \prec a_{\overline{k}}$.

		In case (1), we notice $a_s\prec a_{s+1} \preceq a_{\underline{k}}$ and $a_s\prec a_{s+1} \preceq a_t$.
		Let $a_{s+1} = r_k(P_i)$ for some $1< k \leq m$ and $r_{k-1}(P_i) = a_p$.
		Thus, $a_pP_i!a_{s+1}$.
		Since $r_1(P_i) = a_s \prec a_{\underline{k}}$ and $P_i$ is $(\underline{k}, \overline{k})$-hybrid,
        $a_pP_ia_{s+1}$ implies $a_p\preceq a_s$.
		Thus, $a_p\preceq a_s\prec a_{s+1} \preceq a_{\underline{k}}$ and $a_p\preceq a_s\prec a_{s+1} \prec a_t$,
		which imply $a_{s+1}P_i'a_p$ by the definition of $(\underline{k}, \overline{k})$-hybridness.
		By locally switching $a_p$ and $a_{s+1}$ in $P_i$, we generate a preference $P_i''$.
		Thus, $P_i \sim P_i''$, $a_pP_i!a_{s+1}$, $a_{s+1}P_i''!a_p$ and $a_{s+1}P_i'a_p$.
		We last show $P_i'' \in \mathbb{D}_{\prec}(\underline{k}, \overline{k})$.
		If $r_1(P_i'') = r_1(P_i) = a_s$, according to $a_p \prec a_s \prec a_{s+1}$,
		Observation \ref{obs}(i) implies $P_i'' \in \mathbb{D}_{\prec}(\underline{k}, \overline{k})$.
		If $r_1(P_i'') \neq r_1(P_i)$, it is true that $r_1(P_i) = a_s = a_p$, $r_1(P_i'') = a_{s+1}$, $a_p, a_{s+1} \in L$ and $|p-(s+1)| = 1$.
Then, Observation \ref{obs}(iii) implies $P_i'' \in \mathbb{D}_{\prec}(\underline{k}, \overline{k})$.
		
		The verification of case (2) is similar to that of case (1), and we hence omit it.
		
		In case (3), let $a_{\overline{k}} = r_k(P_i)$ for some $1< k \leq m$ and $r_{k-1}(P_i) = a_p$.
		Thus, $a_pP_i!a_{\overline{k}}$.
		Since $a_{\underline{k}} \preceq a_s \prec a_{\overline{k}}$ and $P_i$ is $(\underline{k}, \overline{k})$-hybrid,
		$a_pP_ia_{\overline{k}}$ implies $a_p \prec a_{\overline{k}}$.
		Thus, we know either $a_p\prec a_{\underline{k}} \prec a_{\overline{k}} \preceq a_t$ which implies $a_{\overline{k}}P_i'a_{\underline{k}}$ and
		$a_{\underline{k}}P_i'a_p$ by the definition of $(\underline{k}, \overline{k})$-hybridness,
		or $a_{\underline{k}} \preceq a_p \prec a_{\overline{k}} \preceq a_t$ which implies $a_{\overline{k}}P_i'a_p$ by the definition of $(\underline{k}, \overline{k})$-hybridness.
        Overall, $a_{\overline{k}}P_i'a_p$.
		By locally switching $a_p$ and $a_{\overline{k}}$ in $P_i$, we generate a preference $P_i''$.
		Thus, $P_i \sim P_i''$, $a_pP_i!a_{\overline{k}}$, $a_{\overline{k}}P_i''!a_p$ and $a_{\overline{k}}P_i'a_p$.
		We last show $P_i'' \in \mathbb{D}_{\prec}(\underline{k}, \overline{k})$.
		If $r_1(P_i'') = r_1(P_i) = a_s$, since  $a_{\overline{k}} \notin L$ and $a_p\notin  R$, Observation \ref{obs}(ii) implies $P_i'' \in \mathbb{D}_{\prec}(\underline{k}, \overline{k})$.
		If $r_1(P_i'') \neq r_1(P_i)$, it is true that $r_1(P_i)=a_s=a_p$, $r_1(P_i'') = a_{\overline{k}}$, and $\{a_p, a_{\overline{k}}\} \subseteq M$.
Then, Observation \ref{obs}(iii) implies $P_i'' \in \mathbb{D}_{\prec}(\underline{k}, \overline{k})$.
		
		In case (4), recall $r_1(P_i') = a_t$.
Let $a_t = r_k(P_i)$ for some $1< k \leq m$ and $r_{k-1}(P_i) = a_p$.
		By locally switching $a_p$ and $a_t$ in $P_i$, we generate a preference $P_i''$.
		Thus, $P_i \sim P_i''$, $a_pP_i!a_t$, $a_tP_i''!a_p$ and $a_tP_i'a_p$.
		We last show $P_i'' \in \mathbb{D}_{\prec}(\underline{k}, \overline{k})$.
		If $r_1(P_i'') = r_1(P_i) = a_s$, since $a_t \notin L\cup R$,
		Observation \ref{obs}(ii) implies $P_i'' \in \mathbb{D}_{\prec}(\underline{k}, \overline{k})$.
		If $r_1(P_i'') \neq r_1(P_i)$, it is true that $r_1(P_i) = a_s =a_p$, $r_1(P_i'') = a_t$ and $\{a_p, a_t\} \subseteq M$. Then, Observation \ref{obs}(iii) implies $P_i'' \in \mathbb{D}_{\prec}(\underline{k}, \overline{k})$.

Overall, we have constructed a $(\underline{k}, \overline{k})$-hybrid preference $P_i''$ by locally switching some alternatives $a_p$ and $a_q$ that are consecutively ranked in $P_i$ (i.e., $a_pP_i!a_q$) and oppositely ranked in $P_i'$ (i.e., $a_qP_i'a_p$).
This completes the verification.
\end{proof}

\medskip
\noindent
\textbf{Clarification 2}:
Let $\Omega$ be a set of linear orders over $A$.
Assume that $|\Omega| \geq 2$, and no pair of linear orders in $\Omega$ is a pair of complete reversed preferences.
The multiple single-peaked domain $\mathbb{D}_{\Omega}$ is a hybrid domain.\medskip
	
\begin{proof}
We can assume w.l.o.g.~that the natural order $\prec$ is contained in $\Omega$ by relabelling alternatives as necessary.
We label all linear orders as $\Omega = \{\prec_1, \dots, \prec_q\}$, where $q \geq 2$.
Given a linear order $\prec_v\,\in \Omega$, let $\min^{\prec_v}(A) \equiv a_s$ if $a_s \prec_v a_k$ for all $k \neq s$ and
$\max^{\prec_v}(A) \equiv a_t$ if $a_k \prec_v a_t$ for all $k \neq t$.
Then, we can assume w.l.o.g.~that for each $\prec_v \,\in \Omega$,
the index of $\min^{\prec_v}(A)$ is smaller than the index of $\max^{\prec_v}(A)$, i.e.,
$[\min^{\prec_v}(A) = a_s\; \textrm{and}\; \max^{\prec_v}(A) = a_t] \Rightarrow [s<t]$.\footnote{Suppose that we have a linear order $\prec'\, \in \Omega$ such that $\min^{\prec'}(A) = a_t$ , $\max^{\prec'}(A) = a_s$ and $t> s$. Then, we can induce a linear order $\prec''$ such that $\prec''$ and $\prec'$ are completely reversed, i.e., $[a_k \prec'' a_{k'}] \Leftrightarrow [a_{k'} \prec' a_k]$.
Then, we have $\min^{\prec''}(A) = a_s$ , $\max^{\prec''}(A) = a_t$ and $t> s$.
Meanwhile, since $\mathbb{D}_{\prec'} = \mathbb{D}_{\prec''}$, we can replace $\prec'$ in $\Omega$ by $\prec''$, and replace $\mathbb{D}_{\prec'}$ in $\mathbb{D}_{\Omega}$ by $\mathbb{D}_{\prec''}$.}
Now, according to $\Omega$, we find that one of the following two cases must occur:
\begin{itemize}
\item[\rm (1)] There exists $1\leq \tilde{k}< m-1$ such that $a_k \prec_v a_{k+1}$
for all $1 \leq k' < \tilde{k}$ and all $\prec_v \in \Omega$, and
$a_{s} \prec_{\eta} a_{\tilde{k}+1}$ for some $s>\tilde{k}+1$ and some $\prec_{\eta} \,\in \Omega$, or

\item[\rm (2)] there exists no $1\leq \tilde{k}< m-1$ such that $a_k \prec_v a_{k+1}$
for all $1 \leq k < \tilde{k}$ and $\prec_v \in \Omega$.
\end{itemize}
In the first case, let $\underline{k} = \tilde{k}$, while in the second case, let $\underline{k} = 1$.
Symmetrically, we can also find that one of the following two cases must occur:
\begin{itemize}
\item[\rm (a)] There exists $2< \hat{k}\leq m$ such that $a_{k-1} \prec_v a_k$
for all $\hat{k}< k \leq m$ and all $\prec_v \in \Omega$, and
$a_{\hat{k}-1} \prec_{\eta} a_{s}$ for some $s<\hat{k}-1$ and some $\prec_{\eta}\, \in \Omega$, or

\item[\rm (b)] there exists no $2< \hat{k}\leq m$ such that $a_{k-1} \prec_v a_k$
for all $\hat{k}< k \leq m$ and all $\prec_v \in \Omega$.
\end{itemize}
In case (a), let $\overline{k} = \hat{k}$, while in case (b), let $\overline{k} = m$.
Since $|\Omega| \geq 2$ and no two linear orders of $\Omega$ are completely reversed,
it must be the case that $\overline{k} -\underline{k}>1$.

According to the natural order $\prec$ and thresholds $a_{\underline{k}}$ and $a_{\overline{k}}$,
we refer to the $(\underline{k}, \overline{k})$-hybrid domain $\mathbb{D}_{\prec}(\underline{k}, \overline{k})$.
We show $\mathbb{D}_{\Omega}  \subseteq \mathbb{D}_{\prec}(\underline{k}, \overline{k})$.
It suffices to show $\mathbb{D}_{\prec_v} \subseteq \mathbb{D}_{\prec}(\underline{k}, \overline{k})$ for all $\prec_v \,\in \Omega$.
Fix an arbitrary $\prec_v \,\in \Omega$ and arbitrary $P_i \in \mathbb{D}_{\prec_v}$.
We show that $P_i$ is a $(\underline{k}, \overline{k})$-hybrid preference.
First, given $a_r, a_s \in [a_1, a_{\underline{k}}] \equiv L$, let $a_r \prec a_s \prec r_1(P_i)$ or
$r_1(P_i) \prec a_s \prec a_r$ hold.
Thus, case (1) must occur, and we hence have $a_r \prec_v a_s \prec_v r_1(P_i)$ or
$r_1(P_i) \prec_v a_s \prec_v a_r$, which consequently by single-peakedness w.r.t.~$\prec_v$ implies $a_sP_ia_r$.
Symmetrically, if $a_r, a_s \in [a_{\underline{k}},a_m] \equiv R$,
we have $[a_r \prec a_s \prec r_1(P_i)$ or $r_1(P_i) \prec a_s \prec a_r]$$\Rightarrow [a_sP_ia_r]$ as well.
This verifies condition (i) of Definition \ref{defn:hybriddomain}.
Next, given $r_1(P_i) \in L$ and $a_l \in [a_{\underline{k}+1}, a_{\overline{k}}] \equiv M\backslash \{a_{\underline{k}}\}$, we show $a_{\underline{k}}P_ia_l$.
If $r_1(P_i) = a_{\underline{k}}$, it is evident that $a_{\underline{k}}P_ia_l$.
Next, we consider the case that $r_1(P_i) \in L\backslash \{a_{\underline{k}}\}$.
Thus, case (1) must occur, and we hence have $r_1(P_i) \prec_v a_{\underline{k}}$ and $a_{\underline{k}} \prec_v a_l$, which consequently by single-peakedness w.r.t.~$\prec_v$ imply $a_{\underline{k}}P_ia_l$.
Symmetrically, if $r_1(P_i) \in R$, we have $[a_r \in M\backslash \{a_{\overline{k}}\}] \Rightarrow [a_{\overline{k}}P_ia_r]$.
This verifies condition (ii) of Definition \ref{defn:hybriddomain}.
Therefore, $P_i$ is a $(\underline{k}, \overline{k})$-hybrid preference, as required.

The multiple single-peaked domain $\mathbb{D}_{\Omega}$ satisfies the path-inclusion condition
since the single-peaked domain $\mathbb{D}_{\prec}$ is contained in $\mathbb{D}_{\Omega}$, and
induces the vertex-path $(a_1, \dots, a_k,a_{k+1}, \dots, a_m)$ in $G_{\approx}^A(\mathbb{D}_{\Omega})$ that includes all alternatives of $A$.
Last, we show the no-leaf condition.
For notational convenience, given $\prec_v\, \in \Omega$, let $a_s \prec_v!~ a_t$ denote that $a_s \prec_v a_t$, and there exists no $a_k$ such that $a_s \prec_v a_k$ and $a_k \prec a_t$.
Clearly, $a_k \prec!~ a_{k+1}$ for all $k = 1, \dots, m-1$.
If case (1) above occurs,
then according to the natural order $\prec \,\in \Omega$ and the alternative $a_{\underline{k}+1}$,
we must have another linear order $\prec_v \,\in \Omega$ and an alternative $a_s \in A$ with $s>\underline{k}+1$ such that $a_{\underline{k}} \prec_v!~ a_s$.
According to the single-peaked domains $\mathbb{D}_{\prec}$ and $\mathbb{D}_{\prec_v}$ included in $\mathbb{D}_{\Omega}$,
we know that $a_{\underline{k}}$ is strongly connected to both $a_{\underline{k}+1}$ and $a_k$.
Since we have shown $\mathbb{D}_{\Omega} \subseteq \mathbb{D}_{\prec}(\underline{k}, \overline{k})$,
$a_{\underline{k}} \approx a_{\underline{k}+1}$ and $a_{\underline{k}} \approx a_k$ imply $a_{\underline{k}+1}, a_k \in M$.
Therefore, $a_{\underline{k}}$ is not a leaf in $G_{\approx}^M(\mathbb{D}_{\Omega})$.
If case (2) above occurs, the inclusion of the natural order $\prec$ in $\Omega$ implies that there exists a linear order $\prec_v \,\in \Omega$ such that $\min^{\prec_v}(A) \neq a_1$.
Moreover, since $\max^{\prec_v}(A) \neq a_1$, we have $a_s, a_t \in A$ such that $a_s \prec_v!~ a_1$ and
$a_1 \prec_v!~ a_t$.
Then, according to the single-peaked domain $\mathbb{D}_{\prec_v}$ included in $\mathbb{D}_{\Omega}$,
we know that that $a_1$ is strongly connected to both $a_s$ and $a_t$.
Thus, it must be the case that $a_s, a_t \in M$.
Therefore, $a_1$ is not a leaf in $G_{\approx}^M(\mathbb{D}_{\Omega})$.
Symmetrically, according to cases (a) and (b) above, we can infer that $a_{\overline{k}}$ is not a leaf in $G_{\approx}^M(\mathbb{D}_{\Omega})$.
Last, note that the path-inclusion condition implies that for each $\underline{k}< k < \overline{k}$,
$a_k$ is never a leaf in $G_{\approx}^M(\mathbb{D}_{\Omega})$. Therefore, the no-leaf condition is satisfied.
In conclusion, $\mathbb{D}_{\Omega}$ is a $(\underline{k}, \overline{k})$-hybrid domain.
\end{proof}
	
}	

\end{document}